\ifdefined\TITsupplement
\documentclass[11pt,journal,onecolumn]{IEEEtran}
\else
\documentclass[12pt,journal,onecolumn]{IEEEtran}
\fi

\RequirePackage{amsthm,amsmath,amsfonts,amssymb,mathtools,mathrsfs}
\RequirePackage[numbers,sort&compress]{natbib}
\RequirePackage{graphicx,array,booktabs}
\usepackage{pdfpages}
\theoremstyle{plain}
\newtheorem{theorem}{Theorem}
\newtheorem{proposition}{Proposition}

\newtheorem{lemma}{Lemma}
\newtheorem{corollary}{Corollary}
\theoremstyle{definition}
\newtheorem{definition}{Definition}
\newtheorem{remark}{Remark}
\newtheorem{assumption}{Assumption}

\newtheorem{assumptionW}{Assumption}

\newtheorem{example}{Example}
\DeclareMathOperator{\Ree}{Re}
\DeclareMathOperator{\Imm}{Im}
\DeclareMathOperator{\dist}{dist}

\DeclareMathOperator*{\argmin}{arg\,min}
\DeclareMathOperator{\spanR}{span_{\mathbb R}}
\DeclareMathOperator{\spanC}{span_{\mathbb C}}
\newcommand{\dd}{\,\mathrm d}
\newcommand{\eps}{\varepsilon}
\newcommand{\CN}{\mathcal{CN}}

\begin{document}
\title{Fisher Information and the Minimum Mean-Square Error for Locating Closely Spaced Point Sources}
\author{Chuling~Wen, Jian~Lu, Qingtang~Jiang, and Charles~K.~Chui%
\thanks{This work was supported in part by the National Natural Science Foundation
of China under Grant 12571569 and Grant U21A20455, in part by the Natural
Science Foundation of Guangdong Province of China under Grant
2024A1515011913, in part by the Shenzhen Key Laboratory of Advanced
Machine Learning and Applications under Grant SYSPG20241211173920042.}%
\thanks{C.~Wen and J.~Lu are with the School of Mathematical Sciences,
Shenzhen University, Shenzhen 518060, China (e-mail: chulwen@163.com;
jianlu@szu.edu.cn). J.~Lu is also with the National Center for Applied
Mathematics Shenzhen (NCAMS), Shenzhen 518055, China. Q.~Jiang is with the
School of Mathematical Sciences, Zhejiang Normal University, Jinhua 321004,
China (e-mail: jiangq@zjnu.edu.cn). C.~K.~Chui is an independent researcher,
Menlo Park, CA 94025, USA (e-mail: ckchui@stanford.edu).
J.~Lu is the corresponding author.}}
\markboth{}{Wen \MakeLowercase{\textit{et al.}}: Fisher Information and the MMSE for Closely Spaced Point Sources}
\maketitle
\begin{abstract}
We determine the optimal average error of locating closely spaced sources from noisy superpositions of a known smooth response with unknown complex coefficients. Relative phases control how changes in one location can be offset by the other unknowns, so local Fisher information alone does not determine the attainable error. For every fixed number $K\ge2$ of sources with independent uniform phases, we prove matching Bayes risk bounds jointly in separation and noise, for all sufficiently small separations and all noise levels. Odd source counts have two risk regimes; even counts have an intermediate regime in which mean squared error is linear in the noise standard deviation. A projected least-squares estimator attains these orders without knowing the phases or the prior density. Bounded nonuniform phase laws preserve the orders under stated support conditions, whereas priors concentrated near poorly identified configurations can change the exponents. A common finite comparison yields both a global inverse bound and paths controlled by the original prior. These give sufficient conditions under which the risk has the order of the pointwise capped Cram\'er--Rao bound, and a lower-tail calculation explains the power and logarithmic regimes. We also give conditional guarantees when the cluster center, separation, and source count are unknown, with an additional low-noise requirement for count selection. A calibrated covariance factor model illustrates the same criterion beyond source localization.

\end{abstract}
\IEEEpeerreviewmaketitle
\section{Introduction}
\label{sec:intro}

\IEEEPARstart{L}{ocating} point sources that are closer than the width of
the observed response is a basic task in line spectral estimation,
direction-of-arrival estimation and coherent imaging
\citep{stoica1989music,donoho1992superres,candes2014superres,shahram2006statistical}.
When the sources are closely spaced, their contributions overlap. The
accuracy of their locations is then governed by how well the data
separate nearby configurations. The Fisher information and the
Cram\'er--Rao bound are the standard tools for quantifying this accuracy
\citep{kay1993estimation,lee1992crb,smith2005resolution}. In this paper,
we study the optimal average accuracy of the source locations through
the following two questions.
\begin{itemize}
\item \emph{Closely spaced sources in complex Gaussian noise.} Consider
$K\ge2$ sources with locations $\tau_1<\cdots<\tau_K$, amplitude moduli
$a_j>0$ and phases $\varphi_j$, observed through a known response $\psi$
at known sampling locations $w_k$:
\begin{equation}
 Y_k=\sum_{j=1}^K a_j e^{i\varphi_j}\psi(w_k,\tau_j)+\eta_k,
 \qquad 1\le k\le n,
 \label{eq:intromodel}
\end{equation}
where $\eta_k\overset{\rm iid}{\sim}\CN(0,\sigma^2)$. The sources form a
cluster $\tau_j=\tau_0+\eps s_j$ with center $\tau_0$, scale $\eps$ and
normalized offsets $s_j$ separated by at least $\rho>0$. A prior $\pi$ on
the normalized offsets, moduli and phases generates one configuration.
The prior does not depend on $\eps$, $n$ or $\sigma$, and the
configuration stays fixed during the observations. With the effective
noise level $\zeta_n=\sigma/\sqrt n$, we aim to determine the minimum
mean-square error
\[
 R_{\pi,j}(\eps,\zeta_n)
 =\inf_{\widehat\tau_j}
       \mathbb E_{\pi,\eta}|\widehat\tau_j(Y)-\tau_j|^2
\]
of the $j$th location jointly in $\eps$ and $\zeta_n$. The infimum is
over all estimators based on $Y=(Y_1,\dots,Y_n)$, including those that
know the prior. The responses covered include line spectra,
$\psi(w,\tau)=e^{i\pi w\tau}$, and smooth kernels used in array
processing and coherent imaging (Assumption~\ref{ass:sampling}). Since
$K$, $\tau_0$ and $\eps$ are known for now, $R_{\pi,j}\lesssim\eps^2$.
The question is how the risk falls below this cap as the noise
decreases.
\item \emph{Scalar targets with degenerate information.} More generally,
let $(P_x)_{x\in\Theta}$ be a parametric experiment with noise scale
$r_n$, and let $g(x)$ be a scalar target with trivial squared error $L$.
Suppose that the efficient Fisher information for $g$ is of order
$I(x)/r_n^2$. We ask when the minimum mean-square error
$\inf_{\widehat g}\mathbb E_\pi|\widehat g-g(x)|^2$ has the order of the
capped Cram\'er--Rao bound $\mathbb E_\pi\min\{L,r_n^2/I(x)\}$. Local
information alone does not settle this when $I$ vanishes on part of the
support of $\pi$.
\end{itemize}

Most analyses of closely spaced sources evaluate the Cram\'er--Rao bound
at a fixed configuration or average it over phases, or bound the error
uniformly over configurations. None of these gives the minimum
mean-square error under a prior when the information degenerates on
part of the parameter space. The same question arises for other
parametric models whose information vanishes through cancellation among
the unknowns.

A natural benchmark is the local Cram\'er--Rao bound
$\zeta_n^2/\{2J_j(x)\}$, where $2J_j(x)/\zeta_n^2$ is the efficient
Fisher information for $\tau_j$ at the configuration $x$
(Section~\ref{lc:efficient-info}). Towards determining the minimum
mean-square error, two difficulties stand out.
\begin{itemize}
\item \emph{Degenerate information near phase cancellation.} The relative
phases determine how much a change in one location can be offset by
changing the other locations and coefficients. Near configurations where
this cancellation is complete, the efficient information degenerates.
The local bound then exceeds the trivial error $\eps^2$ on a set of
configurations whose prior probability depends on the noise level. The
phase-averaged Cram\'er--Rao bound of \citet{swingler1993frequency}
averages the inverse information over the phases. Capping this average
afterwards can miss a whole risk regime.
\item \emph{From local information to the risk of every estimator.}
Local information describes one configuration, while the Bayes risk is
attained by estimators that may use the prior. An upper bound must
control location errors over the whole candidate class, not only near
the truth. A lower bound must compare configurations that keep their
prior mass. Neither follows from a pointwise information calculation,
and both must hold uniformly near degenerate phases and at every noise
level.
\end{itemize}

\paragraph{Main contributions}
This paper determines the minimum mean-square error of closely spaced
sources jointly in separation and noise level. It also gives sufficient
conditions under which this error is described by the Fisher
information. Our main contributions are summarized below.
\begin{enumerate}
\item For every fixed $K$ with independent uniform phases, we prove
matching upper and lower bounds on $R_{\pi,j}$ for all sufficiently small
separations and all noise levels (Theorem~\ref{thm:randomphase}). Odd
source counts have two risk regimes, with
$R_{\pi,j}\asymp\min\{\eps^2,\zeta_n^2\eps^{-(3K-3)}\}$. Even counts have
an additional intermediate regime. In it the mean-square error is
linear in the noise standard deviation. For two sources,
$R_{\pi,j}\asymp\min\{\eps^2,\zeta_n,\zeta_n^2\eps^{-3}\}$. A projected
least-squares estimator attains these orders without knowing the phases
or the prior density. Phase densities with an upper bound and a local
lower bound near one reference configuration preserve them
(Proposition~\ref{prop:bounded-phase-prior}).
\item The risk exponents are set by the prior mass near degenerate
configurations, not by the source count alone. For locally supported
power-weighted phase priors with every even $K$, we determine the risk
of the sources in the reweighted set. For a singular three-source prior,
we determine the risk of every source. Both include logarithmic
boundary cases (Theorem~\ref{thm:prior-effect}). For the priors of
Theorems~\ref{thm:randomphase} and~\ref{thm:prior-effect}, the risk has
the order of the capped Cram\'er--Rao bound
$\mathbb E_\pi\min\{\eps^2,\zeta_n^2/J_j(x)\}$.
\item We give a criterion for the second question
(Theorem~\ref{risk:transfer}). It has two conditions on an information
surrogate $S\lesssim I$. The first is an inverse bound. The change of
the target between the truth and any other parameter value is at most
a constant times the smaller of $\sqrt L$ and their signal distance
divided by $\sqrt S$ at the truth. The second is a matching condition.
Some family of paths charges every set at most a fixed multiple of its
prior mass. Its two-point lower bound is at least a constant times
$\mathbb E_\pi\min\{L,r_n^2/S\}$, with a constant independent of $r_n$.
Under these conditions,
\[
 \inf_{\widehat g}\mathbb E_\pi|\widehat g-g(x)|^2
 \asymp\mathbb E_\pi\min\{L,r_n^2/I(x)\},
\]
and an estimator that does not use the prior attains this order.
Suppose that the prior probability of $\{I<u\}$ is of order
$(u/s_0)^\kappa$ for $u$ between twice the smallest value of $I$ and a
top scale $s_0$. Averaging the inverse information before capping then
keeps this order uniformly in the smallest value of $I$ exactly when
$\kappa>1$ (Proposition~\ref{prop:kappa-law}). The source results are
instances, with paths that fix low-order moments of the source
configuration (Section~\ref{lc:criterion-sources}). A calibrated
covariance factor model with $m$ reference variables is a second
instance. Its risk has order $n^{-m/4}$, $n^{-1}\log(en)$ and $n^{-1}$
for $m<4$, $m=4$ and $m>4$ (Corollary~\ref{factor:risk}).
\item When the center, the separation, and a source count bounded by a
known $K_{\max}$ are unknown, we give guarantees for one fine-scale
cluster. Count selection needs an additional low-noise requirement
(Section~\ref{sec:known-count-adaptation}).
\end{enumerate}

\paragraph{Organization}
Section~\ref{sec:model} states the model and the risk.
Section~\ref{sec:main-results} presents our main results for source
localization, and Section~\ref{sec:related-works} reviews related works.
Section~\ref{sec:info} proves Theorems~\ref{thm:randomphase}
and~\ref{thm:prior-effect} from the key lemmas; the appendices contain
the remaining proofs. Section~\ref{sec:criterion} states the criterion
and the factor model, and Section~\ref{sec:known-count-adaptation}
treats an unknown center, separation and source count.
Section~\ref{sec:discussion} discusses future directions. The
supplementary material contains numerical illustrations, their settings
and a table of conditions, and no proofs.

\paragraph{Notation}
We write $u\lesssim v$ for $u\le Cv$ and $u\asymp v$ for both
comparisons, with scale-independent constants whose allowed dependence
is stated with each result. Source locations are also called nodes. We
write $a\vee b=\max\{a,b\}$ and distinguish real and complex spans by
$\spanR$ and $\spanC$. Table~\ref{lc:notation} summarizes the common
notation.

\begin{table}[t]
\centering
\caption{Summary of Common Notation}
\label{lc:notation}
\begin{tabular}{@{}ll@{}}
\toprule
Symbol & Meaning\\
\midrule
$K$, $\tau_j$, $c_j=a_je^{i\varphi_j}$ & source count, locations, complex coefficients\\
$\tau_0$, $\eps$, $s_j$ & cluster center, separation scale, normalized offsets\\
$\psi$, $\nu_n$, $g_r$ & response, sampling design, Taylor coefficients at $\tau_0$\\
$\zeta_n=\sigma/\sqrt n$ & effective noise level\\
$\mathcal X_K$, $\pi$ & class of normalized configurations, prior on it\\
$R_{\pi,j}(\eps,\zeta_n)$ & minimum mean-square error of $\tau_j$, \eqref{eq:intro-risk}\\
$\Pi$ & projection onto $\spanC\{g_0,\ldots,g_{2K-1}\}$ in $L^2(\nu_n)$\\
$2J_j(x)/\zeta_n^2$ & efficient Fisher information for $\tau_j$; $J_{\min}=\min_jJ_j$\\
$\Psi_K$, $\zeta_{*,K}(\eps)$ & risk order under uniform phases, noise level of the cap\\
$L$, $S$, $\mathcal A_\pi(r;S)$ & cap, information surrogate, $\mathbb E_\pi\min\{L,r^2/S\}$ (Section~\ref{crit:setting})\\
$\kappa$ & lower-tail exponent of the information under the prior\\
\bottomrule
\end{tabular}
\end{table}

\section{Problem Formulation}
\label{sec:model}
In this section, we formulate the observation model, the class of
configurations and the Bayes risk. We then state the condition on the
response and define the efficient Fisher information.

\subsection{Observation Model}
Suppose that we are interested in $K\ge2$ point sources with locations
$\tau_j$, amplitude moduli $a_j>0$ and phases $\varphi_j$, observed
through a known response $\psi(w,\tau)$ at sampling locations $w_k$ in a
set $\Omega$, for example an interval. Let
$\nu_n=n^{-1}\sum_k\delta_{w_k}$ denote the empirical measure of the
sampling locations. What we have observed are the noisy samples
$Y_k=\mu_\theta(w_k)+\eta_k$, $1\le k\le n$, where
\begin{equation}
\label{eq:model}
 \mu_\theta(w)=\sum_{j=1}^Kc_j\psi(w,\tau_j),\qquad
 \eta_k\overset{\rm iid}{\sim}\CN(0,\sigma^2),\qquad
 \theta=(a_j,\varphi_j,\tau_j)_{j=1}^K\in\mathbb R^{3K},
\end{equation}
where $c_j=a_jq_j$, $a_j>0$ and $q_j=e^{i\varphi_j}$. More generally, a
probability measure $\nu$ on $\Omega$ is called a sampling design, and we
put
\[
 \langle h_1,h_2\rangle_\nu=\Ree\int\overline{h_1}h_2\dd\nu,
 \qquad\|h\|_\nu^2=\int|h|^2\dd\nu,\qquad\zeta_n=\sigma/\sqrt n.
\]

\subsection{Cluster Class, Prior, and Bayes Risk}
\label{sec:class-estimator}

The sources form a cluster, $\tau_j=\tau_0+\eps s_j$, with center
$\tau_0$, scale $\eps$ and normalized offsets $s_j$. Fix a known compact
interval $\mathcal S_0$, a minimum normalized gap $\rho>0$ and
amplitude bounds $0<\underline a<1<\overline a$. The class of normalized
configurations is
\begin{equation}
\label{eq:class}
 \mathcal X_K=\{(s,c):s_j\in\mathcal S_0,\ s_{j+1}-s_j\ge\rho,\ \underline a\le|c_j|\le\overline a\}:
\end{equation}
the phases are unrestricted. It is assumed throughout this section and
Sections~\ref{sec:main-results}--\ref{sec:info} that the estimator knows
$K,\tau_0,\eps,\mathcal S_0,\rho$ and the amplitude bounds, but not the
individual locations or coefficients. Assume that $\mathcal X_K$ contains an
interior configuration with unit moduli. Fix $T<\infty$ such that
$\mathcal S_0\subset[-T,T]$. For $x=(s,c)\in\mathcal X_K$ put
\[
 \theta_x=\bigl(|c_j|,\arg c_j,\tau_0+\eps s_j\bigr)_{j=1}^K,\qquad
 \mu_x=\sum_jc_j\psi(\cdot,\tau_0+\eps s_j).
\]

For the baseline phase law, let $(s_1,\ldots,s_K,a_1,\ldots,a_K)$ have a fixed
density $f$ on the node--modulus class in~\eqref{eq:class}, with
$f\ge c_f>0$ almost everywhere on a neighborhood of some interior point,
and draw the phases independently
of these variables and of one another, uniformly on $[0,2\pi)$. This prior
on normalized parameters does not depend on $\eps,n,\sigma$. Because it
is placed on normalized offsets, changing $\eps$ rescales the
configurations without changing their prior weights.
Sections~\ref{lc:phase-laws} and~\ref{sec:prior-mass} treat other phase
laws. For all these priors the accuracy of each ordered location is
measured by the following risk.

\begin{definition}[Bayes localization risk]
\label{def:bayes-risk}
Let $\pi$ be a prior on offsets, moduli and phases. Formally, $\pi$ is
a probability measure on $\mathcal X_K$ that does not depend on
$\eps,n,\sigma$ or the noise. One realization of $\pi$ generates the
configuration, which is held fixed throughout the observations. For
$1\le j\le K$, the Bayes localization risk of the $j$th ordered location
is
\begin{equation}
 R_{\pi,j}(\eps,\zeta_n)
 =\inf_{\widehat\tau_j(Y)}\mathbb E_{\pi,\eta}|\widehat\tau_j-\tau_j|^2,
 \label{eq:intro-risk}
\end{equation}
where the expectation is over the prior and the noise, and the infimum
is over all measurable functions $\widehat\tau_j$ of the full data
$Y=(Y_1,\ldots,Y_n)$, including estimators that use the prior. In
Section~\ref{sec:known-count-adaptation} the prior is a law of the
physical configuration $(\tau_0,s,c)$, which includes the center. There
$R_{\pi,j}$ is defined by the same infimum~\eqref{eq:intro-risk} with
the expectation under this law.
\end{definition}

Since the loss is quadratic and $\tau_j$ is bounded, the infimum
in~\eqref{eq:intro-risk} is attained by the posterior mean
$\mathbb E_\pi[\tau_j\mid Y]$. Thus $R_{\pi,j}$ is the minimum mean-square
error (MMSE) of $\tau_j$.

\subsection{Response Condition}
\label{sec:obs-model}
The response condition is stated through the Taylor coefficients at
the cluster center $\tau_0$ and their design Gram matrices,
\[
 g_r=\frac{1}{r!}\,\partial_\tau^r\psi(\cdot,\tau_0)\quad(0\le r\le2K),\qquad
 \Gamma_m(\nu)=\Bigl(\int\overline{g_a}\,g_b\dd\nu\Bigr)_{0\le a,b<m}.
\]

\begin{assumption}[Response regularity]
\label{ass:sampling}
$\nu$ is a probability measure on $\Omega$, and for some $r_0>0$,
$A_\psi<\infty$ and $\lambda_0>0$:
\begin{enumerate}
\item the curve $\tau\mapsto\psi(\cdot,\tau)$ is $C^{2K}$ in
      $L^2(\nu)$ on $\mathcal J=(\tau_0-r_0,\tau_0+r_0)$, and
      $\|\partial_\tau^r\psi(\cdot,\tau)\|_\nu\le A_\psi$
      for all $\tau\in\mathcal J$ and $0\le r\le2K$;
\item $\Gamma_{2K}(\nu)\succeq\lambda_0I_{2K}$.
\end{enumerate}
A family of designs satisfies the assumption \emph{uniformly} if one
choice of $r_0,A_\psi,\lambda_0$ serves for all its members. In the
risk statements below the family is the design sequence
$(\nu_n)_{n\ge1}$, so the constants do not change as $\nu_n$ changes
with $n$.
\end{assumption}

\begin{remark}[Scope of Assumption~\ref{ass:sampling}]
Part~1 asks for finitely many uniformly bounded derivatives. A bounded
holomorphic extension to a fixed complex neighborhood is sufficient,
by Cauchy's estimates, but is not required. Part~2 requires that no
nontrivial combination of the response and its first $2K-1$ location
derivatives at $\tau_0$ vanish on the design. For line spectra take
$\psi(w,\tau)=e^{i\pi w\tau}$ on $\Omega=[-w_{\max},w_{\max}]$. Part~1
holds with $A_\psi=\max\{1,(\pi w_{\max})^{2K}\}$. The coefficients
$g_r(w)=(i\pi w)^re^{i\pi\tau_0w}/r!$ give
$\Gamma_{2K}=\overline{\Theta}H\Theta$ with the moment matrix
$H=(\int w^{a+b}\dd\nu)_{a,b<2K}$ and
$\Theta=\operatorname{diag}((i\pi)^r/r!)_{r<2K}$. Hence part~2 holds
uniformly exactly when $H\succeq\lambda_0'I$ uniformly over the designs.
Consider the centered grid $w_k=k/n$ with odd $n\ge2K$. Its empirical
designs have at least $2K$ support points and converge to the uniform
measure on $[-1/2,1/2]$, whose moment matrix is positive definite.
Hence the assumption holds uniformly. This covers sampled sinusoids, uniform
linear arrays and delays on equally spaced subcarriers.
\end{remark}

For the Gaussian response
$\psi(w,\tau)=\exp[-(w-\tau)^2/(2\varsigma^2)]$,
$g_r(w)=\psi(w,\tau_0)p_r(w-\tau_0)$ with a polynomial $p_r$ of degree
$r$. A combination of $g_0,\ldots,g_{2K-1}$ that vanishes at $2K$
distinct sampling points is therefore zero. Hence each such fixed design
satisfies Assumption~\ref{ass:sampling}. A varying design family still
needs a common Gram bound.

Assumption~\ref{ass:sampling} concerns one center. The extensions in
Section~\ref{sec:known-count-adaptation} use a uniform Gram condition
on a fixed search window when the center, separation, or source count
is unknown. Lemmas~\ref{scope:gw-bridge} and~\ref{scope:regular-block}
give the corresponding projection and a sampling design without
aliases on that window.

\subsection{Efficient Fisher Information}
\label{lc:efficient-info}
The Fisher information of $\theta$ in~\eqref{eq:model} is
\begin{equation}
\label{eq:fisher}
 \mathcal I(\theta)_{ab}=\frac{2n}{\sigma^2}\langle\chi_a,\chi_b\rangle_{\nu_n},
 \qquad \chi_a=\partial_{\theta_a}\mu_\theta .
\end{equation}
For nonsingular $\mathcal I(\theta)$, the information on $\theta_a$ with all other
coordinates treated as nuisance parameters is
\begin{equation}
\label{eq:profiled}
 \bigl[(\mathcal I^{-1})_{aa}\bigr]^{-1}
 =\frac{2n}{\sigma^2}\dist_{\nu_n}^2\bigl(\chi_a,\spanR\{\chi_b:b\ne a\}\bigr).
\end{equation}
This Schur complement is the efficient information
\citep{kay1993estimation,hochwald1994concentrated,vandervaart1998asymptotic,murphy2000profile},
also called equivalent Fisher information in localization
\citep{shen2010fundamental}. The quantity $J_j$ below removes the noise
and amplitude factors. Set
\[
 E_j=\psi(\cdot,\tau_j),\qquad E'_j=\partial_\tau\psi(\cdot,\tau_j).
\]
The tangents are $\partial_{a_k}\mu=q_kE_k$, $\partial_{\varphi_k}\mu=ia_kq_kE_k$ and
$\partial_{\tau_k}\mu=a_kq_kE_k'$, and the nuisance tangent space of
$\tau_j$ is
\begin{equation}
\label{eq:tangentspace}
 \mathcal T_j(\varphi)=\spanC\{E_1,\ldots,E_K\}
 +\spanR\{q_kE_k':k\ne j\}.
\end{equation}

\begin{lemma}[Efficient node information]
\label{lem:profiled}
Let $\mathcal I(\theta)$ be nonsingular; under Assumption~\ref{ass:sampling} with
$\nu=\nu_n$, nonsingularity holds for all small $\eps>0$. Then the efficient Fisher
information for $\tau_j$ is $2a_j^2J_j(\eps,\varphi)/\zeta_n^2$, where
\begin{equation}
\label{eq:Jdef}
 J_j(\eps,\varphi)=
 \min_{z\in\mathbb C^K,\;b\in\mathbb R^K,\;b_j=1}
 \Bigl\|\sum_{k=1}^Kz_kE_k+\sum_{k=1}^Kb_kq_kE_k'\Bigr\|_{\nu_n}^2 ,
\end{equation}
and every unbiased estimator of $\tau_j$ has variance at least
$\zeta_n^2/(2a_j^2J_j(\eps,\varphi))$.
\end{lemma}
\begin{proof}
Lemma~\ref{lem:bridge} in Section~\ref{app:bridge} for the doubled node list and confluent
Vandermonde invertibility give full tangent rank. Amplitude and phase
variations span $\mathbb CE_k$, whereas location variations span
$\mathbb Rq_kE_k'$. With the nuisance space $\mathcal T_j$ in
\eqref{eq:tangentspace}, the Schur-complement formula gives
$[(\mathcal I^{-1})_{\tau_j\tau_j}]^{-1}
=2\zeta_n^{-2}\dist_{\nu_n}^2(a_jq_jE_j',\mathcal T_j)
=2a_j^2J_j(\eps,\varphi)/\zeta_n^2$. The variance bound is the Cram\'er--Rao inequality
\citep{kay1993estimation}.
\end{proof}

Equivalently, $J_j(\eps,\varphi)=\dist_{\nu_n}^2\bigl(q_jE_j',\mathcal T_j(\varphi)\bigr)$;
we write $J_j(\eps,\varphi;\nu)$ when the design differs from $\nu_n$.
For $x=(s,c)\in\mathcal X_K$ write $J_j(x)=|c_j|^2J_j(\eps,\arg c;\nu_n)$
and $J_{\min}(x)=\min_jJ_j(x)$. The efficient Fisher information is then
$2J_j(x)/\zeta_n^2$. The factor $|c_j|^2$ is included when the argument
is $x$ and omitted when it is $(\eps,\varphi)$.

Since $R_{\pi,j}\lesssim\eps^2$, the natural benchmark built from the
information is, up to constant factors, the capped Cram\'er--Rao bound
$\mathbb E_\pi\min\{\eps^2,\zeta_n^2/J_j(x)\}$.

\section{Main Results}
\label{sec:main-results}

We now state the results on the Bayes localization risk, that is, the
minimum mean-square error of each location. We first determine the
risk under uniform phases, and then show how changing the prior near
poorly identified configurations changes its exponents. In both cases
the same projected least-squares estimator attains the bounds without
using the prior density.

\subsection{Uniform Phases}
\label{sec:randomphase}

We begin with the baseline prior of Section~\ref{sec:class-estimator}.
\label{sec:projected-estimator}%
Let $\Pi$ be the complex orthogonal projection in $L^2(\nu_n)$ onto
$\spanC\{g_0,\ldots,g_{2K-1}\}$. This projection reduces the data to
$2K$ complex coordinates and, by Lemma~\ref{lem:moment-comparison},
changes distances between noiseless signals by at most a constant
factor. We investigate the projected global least-squares estimator
\begin{equation}
\label{eq:estimator}
 \widehat x=(\widehat s,\widehat c)\in\argmin_{y\in\mathcal X_K}\|\Pi(Y-\mu_y)\|_{\nu_n}^2,
 \qquad \widehat\tau_j=\tau_0+\eps\widehat s_j ,
\end{equation}
a measurable global minimizer over nodes and complex coefficients
jointly, which exists by \citet{brown1973measurable}. The following
theorem quantifies its risk and shows that no estimator does better in
order.

\begin{theorem}[Bayes risk under uniform phases]
\label{thm:randomphase}
Fix $K\ge2$, suppose that the designs $(\nu_n)_{n\ge1}$ satisfy
Assumption~\ref{ass:sampling} uniformly, and let $\pi$ be the joint prior
of Section~\ref{sec:class-estimator}. Let $U$ be a closed ball in the
interior of the node--modulus class of~\eqref{eq:class} and $c_f>0$ such
that $f\ge c_f$ almost everywhere on $U$. There is $\eps_0>0$ such that,
for every $j\le K$, every $0<\eps\le\eps_0$, every $n$ and every
$\zeta_n>0$,
\begin{equation}
\label{eq:random-phase-risk}
 \begin{split}
 R_{\pi,j}(\eps,\zeta_n)&\asymp\Psi_K(\eps,\zeta_n),\\
 \Psi_K(\eps,\zeta)&=
 \begin{cases}
 \min\{\eps^2,\zeta^2\eps^{-(3K-3)}\},&K\text{ odd},\\
 \min\{\eps^2,\zeta\eps^{3-3K/2},
                  \zeta^2\eps^{-(3K-3)}\},&K\text{ even}.
 \end{cases}
 \end{split}
\end{equation}
The projected global least-squares estimator~\eqref{eq:estimator},
which uses neither the phases nor the prior density, attains this order.
The threshold $\eps_0$ and the constant in the upper bound depend only
on $K$, $\mathcal S_0$, $\rho$, $\underline a$, $\overline a$ and the
constants $r_0,A_\psi,\lambda_0$ of Assumption~\ref{ass:sampling}. The
upper bound holds for every law of the nodes and moduli on the class
when the phases are independent and uniform. The constant in the lower
bound depends in addition only on $U$ and $c_f$. Hence the lower bound
holds uniformly over all node--modulus densities with $f\ge c_f$ almost
everywhere on the fixed ball $U$. No constant depends on
$\eps,n,\sigma,\zeta_n$ or otherwise on $f$.
\end{theorem}

\begin{proof}See Section~\ref{sec:randomphase-proof} and
Appendices~\ref{geom:source-lower} and~\ref{sec:weighted-phase-details}. The upper bound is
Lemma~\ref{cor:conditional-mse} integrated with the information ellipse.
The lower bound is Lemma~\ref{geom:fiber-paths}(i)--(ii)
with~\eqref{spine:two-integrals}.\end{proof}

\begin{remark}[Unknown center, separation and source count]
While $K$, $\tau_0$ and $\eps$ are taken as known in
Theorem~\ref{thm:randomphase}, they are often unknown in practice.
Under additional sampling and coefficient conditions, the orders below
the cap remain attainable without knowing the center or the separation.
Section~\ref{sec:known-count-adaptation} also treats an unknown source
count bounded by a known $K_{\max}$. That case needs a further noise
threshold and charges an explicit penalty for count errors.
\end{remark}

\begin{figure}[t]
\centering
\includegraphics[width=\textwidth]{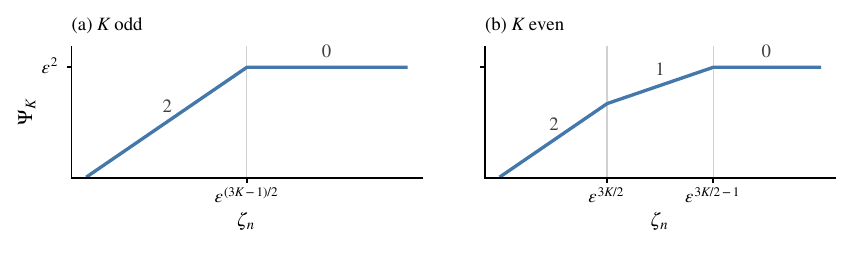}
\caption{Orders of the Bayes risk $\Psi_K(\eps,\zeta_n)$ in
Theorem~\ref{thm:randomphase} on logarithmic axes. Numbers on the
segments are slopes in $\log\zeta_n$; the plateau is the cap $\eps^2$.}
\label{fig:risk-regimes}
\end{figure}

\subsubsection{Implications}
Theorem~\ref{thm:randomphase} determines the Bayes risk of every
location for all small separations and all noise levels. We now single
out several of its implications.
\begin{itemize}
\item \emph{A parity law.} Write $\zeta_{*,K}(\eps)=\eps^{(3K-1)/2}$ for
odd $K$ and $\zeta_{*,K}(\eps)=\eps^{3K/2-1}$ for even $K$. For
$\eps\le1$ this is the noise level at which $\Psi_K$ reaches the cap
$\eps^2$. Below this level the risk of an odd number of sources is
proportional to $\zeta_n^2$. An even number of sources has an
intermediate range in which the mean-square error is linear in the noise
standard deviation (Figure~\ref{fig:risk-regimes}). For two sources and
a fixed per-observation noise level $\sigma$,
$R_{\pi,j}\asymp\min\{\eps^2,\sigma n^{-1/2},\sigma^2n^{-1}\eps^{-3}\}$.
In the intermediate range, halving the mean-square error requires four
times as many observations. In the low-noise range it requires twice as
many.
\item \emph{Averaging the Cram\'er--Rao bound before capping.} Recall from
Section~\ref{lc:efficient-info} that the efficient Fisher
information for $\tau_j$ is $2J_j(x)/\zeta_n^2$. For odd $K$,
Theorem~\ref{thm:randomphase} and
Proposition~\ref{thm:avg-inverse-info} give
$R_{\pi,j}\asymp\min\{\eps^2,\zeta_n^2\mathbb E_\pi[1/J_j]\}$. Thus
averaging the Cram\'er--Rao bound before capping it has the correct
order. For even $K$ the same two results show that this average exceeds
the risk in the intermediate range. The factor is at most
$\eps^{-1/2}$ and is attained at $\zeta_n=\eps^{(3K-1)/2}$.
\item \emph{Where the intermediate range comes from.} For $K=2$, by
Lemma~\ref{lem:weighted}, uniformly in the phases and nodes,
\begin{equation}
 J_1(\eps,\varphi)\asymp\eps^2(\vartheta^2+\eps^2),
 \qquad
 \vartheta=\operatorname{dist}(\varphi_2-\varphi_1,\pi\mathbb Z).
 \label{eq:intro-info}
\end{equation}
If $\eps^3\ll\zeta_n\ll\eps^2$, the local Cram\'er--Rao bound reaches
$\eps^2$ on the arc $\vartheta\lesssim\zeta_n\eps^{-2}$, and the
probability of this arc times $\eps^2$ gives the term $\zeta_n$.
Averaging the local bound first gives $\zeta_n^2\eps^{-3}$. Capping
this average afterwards therefore misses the intermediate range. In
general, write $\kappa$ for the \emph{lower-tail exponent} of the
information under the prior. It means that the prior probability that
the information falls below a level $s$ behaves as $s^\kappa$
(Section~\ref{crit:kappa} gives the precise
definition~\eqref{crit:lower-tail}). For the information ellipse of
Section~\ref{sec:info-def} and even $K$, $\kappa=1/2$ over a nonempty
range conditionally on the nuisance phases, and
Lemma~\ref{geom:ellipse}(2) gives the term linear in $\zeta_n$. For odd
$K$, Section~\ref{sec:randomphase-proof} shows that averaging the
inverse information before capping already gives the risk.
\end{itemize}

\subsubsection{Dependent and Nonuniform Phase Laws}
\label{lc:phase-laws}
The uniform phase law can be replaced by any phase density that is
bounded above and bounded below near one reference configuration.

\begin{proposition}[Dependent and nonuniform phase laws]
\label{prop:bounded-phase-prior}
\label{cor:bounded-phase-prior}
\label{sec:bounded-phase-law}
Keep a node--modulus density $f(s,a)$ on the class~\eqref{eq:class},
let $U_K$ be the uniform probability measure on the phase torus, and
replace the uniform phase law by $h(\varphi\mid s,a)\,dU_K(\varphi)$, where
$h$ is measurable and integrates to one. The phases may be dependent,
and their law may depend on $(s,a)$. For a target $j$, take an interior
node--modulus point with the reference phases of
Appendix~\ref{geom:source-lower}. These phases are zero on a set
$A\ni j$ of size $\lfloor K/2\rfloor+1$ and $\pi/2$ elsewhere. Let
$\mathcal N_j$ be a fixed open product neighborhood of this point whose
closure lies in the interior of the class. Assume
\[
 h(\varphi\mid s,a)\le M_h\ \text{a.e.},\qquad
 f(s,a)h(\varphi\mid s,a)\ge f_j^->0\ \text{a.e. on }\mathcal N_j .
\]
Here $f$, $h$ and $M_h$ are independent of $\eps,n,\sigma$. Then
\[
 R_{\pi,j}(\eps,\zeta_n)\asymp\Psi_K(\eps,\zeta_n),
 \qquad 0<\eps\le\eps_0,\quad \zeta_n>0,
\]
attained by~\eqref{eq:estimator}; besides the
quantities of Theorem~\ref{thm:randomphase}, $\eps_0$ and the constants
depend only on $M_h$, $\mathcal N_j$ and $f_j^-$.
\end{proposition}

\begin{proof}
We have $\pi\le M_h\pi_0$ with
$\pi_0=f\,ds\,da\,dU_K$, and the upper bound of
Theorem~\ref{thm:randomphase} holds for every node--modulus law. For the
lower bound, take the paths of Appendix~\ref{geom:source-lower} about
this reference point, in a box inside $\mathcal N_j$. On $\mathcal N_j$
the density of $\pi$ in the coordinates $(s,\Ree c,\Imm c)$ is
$fh/\{(2\pi)^K\prod_la_l\}$. It is at least $f_j^-/(2\pi\overline a)^K$.
Hence the chart Jacobian bounds the occupation measure of the paths by a
constant times $\pi$, and~\eqref{spine:two-integrals} applies.
\end{proof}

Neither independence nor a global lower bound is needed. The condition
holds under the hypothesis of Theorem~\ref{thm:randomphase} when the
node--modulus projection of $\mathcal N_j$ lies in $U$, and for phase
densities $h\propto e^{V}$ with $|V|\le V_0$ (von Mises laws of bounded
concentration and their correlated versions) whenever $f$ is bounded
below on that projection.

\subsubsection{Comparisons With Prior Works}
\label{sec:comparison}
\label{sec:intro-known}
Phase effects on the Cram\'er--Rao bound of closely spaced sources have
been studied before, mostly at a fixed configuration. We compare
Theorem~\ref{thm:randomphase} with the closest results.
\begin{itemize}
\item For a real sinusoid with a uniform phase,
\citet[Section~II-B1]{swingler1993frequency} observed that averaging
the Fisher information before inversion returns the isolated-source
bound. He instead used the phase-averaged Cram\'er--Rao bound,
proportional to the inverse cube of the separation, and deduced an
inverse-fifth-power signal-to-noise requirement. He found the same
inverse-cube law numerically for two complex sinusoids in complex
noise (Section~II-B2). These have the separation dependence of the
low-noise term $\zeta_n^2\eps^{-3}$ of Theorem~\ref{thm:randomphase}
for $K=2$, and were obtained without proof and without the
Bayes risk. Our results establish
the averaged inverse-information order for every fixed $K$, the
parity law, and attainability throughout the noise range.
\item \citet{lee1992crb} obtained the information exponent $2K-2$ for
general smooth responses under a condition on coefficient variation
across snapshots. \citet{yau1992worst} identified the phases maximizing
the Cram\'er--Rao bound at a given separation, with equal phases least
favorable under symmetric sampling. For two complex sinusoids, the
small-separation orders of the worst-phase and best-phase bounds derived
by \citet{dilaveroglu1998nonmatrix} are the inverse fourth and inverse
second powers of the separation. They agree with~\eqref{eq:intro-info}
at $\vartheta=0$ and at fixed $\vartheta>0$. Under deterministic
bounded noise on Fourier data, the relative phase changes the
two-source resolution orders for source-number detection and location
estimation \citep{he2026twopoint}. The present question is the
contribution of the phases under a prior, including shrinking
neighborhoods of degenerate phases.
\item For deterministic bounded noise, \citet{liu2021computational}
gives number-detection and support-recovery scales proportional to the
corresponding powers $1/(2K-2)$ and $1/(2K-1)$ of the
noise-to-minimum-amplitude ratio, in inverse-bandwidth units. In the
present model, a fixed sufficiently small ratio of root mean squared
error to separation is reached at
\[
 \eps\asymp
 \begin{cases}
   \zeta_n^{2/(3K-1)},&K\text{ odd},\\
   \zeta_n^{2/(3K-2)},&K\text{ even}.
 \end{cases}
\]
These are prior-average scales, not uniform guarantees over phases.
For line spectra, the inverse square-root information at aligned
phases has the separation exponent of the cluster-node error
amplification $\eps^{-(2K-2)}$ in
\citet[Theorem~2.10]{batenkov2020near}.
\end{itemize}

\subsection{Effect of the Prior}
\label{sec:prior-mass}

We now ask whether the exponents in Theorem~\ref{thm:randomphase} are
determined by the source count. Of the two priors below, the first
reweights phases near an information-degenerate hypersurface by a power
of its defining function. The second places a power singularity at
aligned three-source configurations.

\subsubsection{Power Weights Near an Even-Source Degeneracy}
\label{sec:power-weights}

We first reweight the prior near an information-degenerate phase set.
Fix an even $K$ and a partition $A\cup A^{c}=\{1,\ldots,K\}$
with $|A|=K/2+1$ and $|A^{c}|=K/2-1$. At an admissible interior reference point,
set the phases to zero on $A$ and to $\pi/2$ on $A^{c}$. Define
\begin{equation}
\label{eq:phase-prior-determinant}
 \mathscr D(s,\varphi)=\det\begin{pmatrix}
 (s_a^m\cos\varphi_a)_{0\le m<K/2,\,a\le K}\\
 (s_a^m\sin\varphi_a)_{0\le m<K/2,\,a\le K}
 \end{pmatrix}.
\end{equation}
Locally, $\mathscr D=0$ is a smooth hypersurface, and
on the rectangles below $J_j\asymp\eps^{3K-4}(\mathscr D^2+\eps^2)$ for
$j\in A$; thus $\mathscr D=0$ is the set of phases at which the
information of these nodes drops. Let
$\mathcal S_\beta\times\mathcal G\times\mathcal P$ be fixed rectangles
of node offsets, moduli and phases around this reference
point, with closure inside the admissible class and small in the sense
of Theorem~\ref{thm:prior-effect}(a), let $\beta>-1$ be
fixed, and let $f$ be a probability density supported on
$\mathcal S_\beta\times\mathcal G$, with $0<c_f\le f\le C_f<\infty$
there. The prior $\pi_\beta$ has, in unwrapped
phase coordinates, the density
\begin{equation}
\label{eq:physical-weighted-phase-prior}
 p_\beta(s,a,\varphi)=f(s,a)
 \frac{|\mathscr D(s,\varphi)|^\beta}{Z_\beta(s)}
 \mathbf1_{\mathcal P}(\varphi),\qquad
 Z_\beta(s)=\int_{\mathcal P}|\mathscr D(s,\varphi)|^\beta\,d\varphi,
\end{equation}
and is independent of $\eps,n,\sigma$ and the noise. A negative exponent
$\beta$ places more prior mass near $\mathscr D=0$, and a positive one
less.

\subsubsection{A Singular Joint Prior With Three Sources}
\label{sec:singular-prior-k3}

The following singular prior produces an intermediate range for three
sources. Fix an interior configuration of three
separated normalized nodes with positive real weights $\mathsf w_j$. Its six real
moments $\mathfrak r_k=\sum_{j=1}^3\mathsf w_js_j^k$, $0\le k\le5$, form local coordinates.
Choose a fixed small moment rectangle $\mathcal R$ whose inverse remains
inside the admissible node and weight class. For $\mathfrak r\in\mathcal R$, set
\[
 c_j=\mathsf w_j\{1+i(t_0+vs_j+us_j^2)\},\qquad
 (t_0,v,u)\in T_0\times(0,v_0)\times(0,u_0),
\]
where $T_0$ is a small interval about zero and the fixed support lies
near an aligned configuration, within the interior of $\mathcal X_3$.
Here $t_0+vs_j+us_j^2$ is the ratio $\operatorname{Im}c_j/\operatorname{Re}c_j$,
interpolated quadratically in the nodes, and $u$ is its quadratic
coefficient; the aligned configurations are those with $v=u=0$. The nine
real coordinates $(\mathfrak r,t_0,v,u)$ parametrize a neighborhood of the aligned
configuration. For fixed $0<\gamma<1$, the prior $\pi_\gamma$ has the
density
\begin{equation}
 p_\gamma(\mathfrak r,t_0,v,u)=Z_\gamma^{-1}f(\mathfrak r,t_0,v,u)u^{-\gamma},
 \qquad 0<c_f\le f\le C_f<\infty,
 \label{eq:k3-singular-prior}
\end{equation}
relative to $d\mathfrak r\,dt_0\,dv\,du$ on this support, independent of
$\eps,n,\sigma$ and the noise; it is singular along the surface $u=0$.

\subsubsection{Risk Under the Weighted and Singular Priors}
The following theorem determines the Bayes risk under the two priors
just described, with the same estimator as before.

\begin{theorem}[Risk under power-weighted and singular priors]
\label{thm:prior-effect}
Suppose that the designs $(\nu_n)_{n\ge1}$ satisfy
Assumption~\ref{ass:sampling} uniformly.
\begin{enumerate}
\item[(a)] Let $K$ be even and let $\pi_\beta$, $\beta>-1$, be the power-weighted
phase prior~\eqref{eq:physical-weighted-phase-prior} of
Section~\ref{sec:power-weights}, with index set $A$ and the fixed small
rectangles $\mathcal S_\beta\times\mathcal G\times\mathcal P$ of that
section, chosen as in Steps~1 and~2 of the proof in
Appendix~\ref{sec:weighted-phase-details}: $\mathscr D=0$ is a graph
over the other coordinates with $|\partial_{\varphi_{a_0}}\mathscr D|$
bounded above and below for some $a_0\in A$, the closure lies in the
moment chart of Lemma~\ref{weighted:moment-chart}, and the minors used
in Step~2 stay bounded away from zero on the rectangles. (On them
$J_j\asymp\eps^{3K-4}(\mathscr D^2+\eps^2)$ for $j\in A$; the paths of
Step~4 run in a fixed box inside them.) There is $\eps_0>0$ such
that, for each $j\in A$, every $0<\eps\le\eps_0$ and every $\zeta_n>0$,
\begin{equation}
\label{eq:local-weighted-phase-risk}
 R_{\pi_\beta,j}(\eps,\zeta_n)\asymp\Psi_{\beta,K}(\eps,\zeta_n)
 :=\eps^2\int_0^1
 \min\left\{1,\frac{\lambda^2}{t^2+\eps^2}\right\}t^\beta\,dt,
 \qquad \lambda=\frac{\zeta_n}{\zeta_{*,K}(\eps)}=\frac{\zeta_n}{\eps^{3K/2-1}}.
\end{equation}
For $\lambda\ge1$ the risk is of order $\eps^2$, and for $0<\lambda\le1$,
\begin{equation}
\label{eq:weighted-phase-regimes}
 \frac{R_{\pi_\beta,j}(\eps,\zeta_n)}{\eps^2}\asymp
 \begin{cases}
 \lambda^2(\eps\vee\lambda)^{\beta-1},&-1<\beta<1,\\
 \lambda^2\{1+\log[1/(\eps\vee\lambda)]\},&\beta=1,\\
 \lambda^2,&\beta>1.
 \end{cases}
\end{equation}
\item[(b)] Let $K=3$ and let $\pi_\gamma$, $0<\gamma<1$, be the singular
prior~\eqref{eq:k3-singular-prior} of Section~\ref{sec:singular-prior-k3}.
There is $\eps_0>0$ such that, for every $j\le3$, every
$0<\eps\le\eps_0$ and every $\zeta_n>0$,
\begin{equation}
 R_{\pi_\gamma,j}(\eps,\zeta_n)
 \asymp\eps^2\Phi_\gamma(\eps,\lambda),
 \qquad \lambda=\zeta_n/\zeta_{*,3}(\eps)=\zeta_n/\eps^4,
 \label{eq:k3-singular-risk}
\end{equation}
where $\Phi_\gamma(\eps,\lambda)=1$ for $\lambda\ge1$, and for
$0<\lambda\le1$,
\begin{equation}
 \Phi_\gamma(\eps,\lambda)=
 \begin{cases}
  \lambda^2,&0<\gamma<1/2,\\
  \lambda^2\{1+\log[1/(\eps\vee\lambda)]\},&\gamma=1/2,\\
  \lambda^2(\eps\vee\lambda)^{1-2\gamma},&1/2<\gamma<1.
 \end{cases}
 \label{eq:k3-singular-profile}
\end{equation}
\end{enumerate}
In both parts the projected global least-squares
estimator~\eqref{eq:estimator}, minimizing over the full class
$\mathcal X_K$, attains this order without knowing the prior.
In part~(a), $\eps_0$ depends only on $K$, the class~\eqref{eq:class},
the constants of Assumption~\ref{ass:sampling}, the reference
configuration and the rectangles
$\mathcal S_\beta\times\mathcal G\times\mathcal P$. The constants depend
in addition only on $\beta$ and the lower bound $c_f$. Part~(a)
therefore holds uniformly over all densities $f$ of
Section~\ref{sec:power-weights} with a given $c_f$, whatever $C_f$. In
part~(b), $\eps_0$ depends only on the class~\eqref{eq:class}, the
constants of Assumption~\ref{ass:sampling} and the support
$\mathcal R\times T_0\times(0,v_0)\times(0,u_0)$. Since $p_\gamma$ is
normalized, the constants depend in addition only on $\gamma$ and the
ratio $C_f/c_f$. This ratio enters the lower bound through the paths of
Lemma~\ref{geom:fiber-paths}(iii). Part~(b) therefore holds uniformly
over all densities $f$ with a given ratio $C_f/c_f$. No constant
depends on $\eps,n,\sigma,\zeta_n$.
\end{theorem}

\begin{proof}See Section~\ref{sec:randomphase-proof}. The prior-specific phase
integrals and paths are in Appendices~\ref{sec:weighted-phase-details}
and~\ref{prf:k3-singular}.\end{proof}

\subsubsection{Implications}
In short, Theorem~\ref{thm:prior-effect} shows that the prior mass near
$\mathscr D=0$ sets the intermediate law. For $-1<\beta<0$ the risk
decays more slowly in $\zeta_n$ than under uniform phases, and for
$0<\beta<1$ it decays faster. At $\beta=1$ a logarithmic factor appears,
and for $\beta>1$ the intermediate range disappears. For three sources,
the singular prior with $1/2<\gamma<1$ creates an intermediate range,
which is absent under uniform phases. Mixing the prior of part~(a) in a
fixed proportion with a uniform-phase baseline gives the order of the
larger component risk (Corollary~\ref{cor:full-support-weighted-phase}),
so the slower law for $-1<\beta<0$ persists under priors with full phase
support.

Both theorems are instances of the criterion of
Section~\ref{sec:criterion}. Section~\ref{lc:criterion-sources} records
where Section~\ref{sec:info} verifies its conditions for each prior
above.

\section{Related Works}
\label{sec:related-works}

The accuracy of locating closely spaced sources has been studied from
several angles. A first line of work describes the local accuracy
through the Cram\'er--Rao bound. It gives information exponents with
multiple snapshots \citep{lee1992crb}, phase-dependent closed forms
\citep{dilaveroglu1998nonmatrix}, the least favorable phases
\citep{yau1992worst} and phase-averaged bounds
\citep{swingler1993frequency}. Section~\ref{sec:comparison} compares
them with Theorem~\ref{thm:randomphase}. \citet{hockmann2025sparse}
proved a conjecture of Ferreira Da Costa and Mitra on the transition
point between easier and harder super-resolution. The transition is
defined through the size of the Cram\'er--Rao bound. He also extended
the result to higher dimensions. Resolution criteria based on estimating or testing a
two-source separation \citep{smith2005resolution,shahram2006statistical}
use a loss different from the error of an individual location.

A second line concerns worst-case super-resolution, in which the error
is controlled uniformly over configurations
\citep{donoho1992superres,candes2014superres,tang2015minimax,batenkov2020near,liu2021computational},
together with results on conditioning, support stability and ESPRIT
\citep{dacosta2020stable,li2020esprit,yang2023esprit}.
These works address different error criteria; Section~\ref{sec:comparison}
relates the deterministic scales of \citet{liu2021computational} and
\citet{batenkov2020near} to the prior-average scales of
Theorem~\ref{thm:randomphase}.

Nonstandard rates also occur in finite mixtures, singular-information
models, and incoherent imaging
\citep{heinrich2018strong,ho2019singularity,wu2020gaussian,rotnitzky2000singular,tsang2018conservative}.
The aligned-phase balance
$n^{-1/2}\eps^{-(2K-2)}\asymp\eps$ has the exponent
$n^{-1/(4K-2)}$ of the Gaussian-mixture rate with known common variance
in \citet{wu2020gaussian} and of the rate for $K$ components near a
single true component in \citet{heinrich2018strong}; the experiments
and losses are different. Singular-information and weak-identification
analyses
\citep{rotnitzky2000singular,ekvall2022confidence,andrews2012estimation}
describe local behavior or drifting sequences, rather than the
Bayes localization risk jointly in separation and noise.
Distribution-dependent sample complexity in multireference alignment
\citep{bandeira2020optimal,abbe2019multireference} and prior effects
near singular sets \citep{watanabe2009algebraic,rousseau2011overfitted}
provide related phenomena. Theorem~\ref{thm:prior-effect} specifies
the effect here through matching squared-error risks for individual
source locations.

Finally, Ziv--Zakai-type testing arguments, the van Trees inequality and
$f$-informativity bound the Bayes risk below
\citep{zivzakai1969,chazan1975improved,bell1997extended,zhang2023zzb,jeong2025zzb,gill1995applications,chen2016bayes},
and moduli of continuity give minimax rates \citep{donoho1991geometrizing2}
and two-point benchmarks at each instance
\citep{cailow2015framework,duchiruan2024local}. These provide related
lower-bound techniques; Section~\ref{lb:criterion-comparison} states how
Theorem~\ref{risk:transfer} differs from the two-point benchmarks of
\citet{cailow2015framework} and \citet{duchiruan2024local} and from
asymptotic Bayesian Cram\'er--Rao analyses.

\section{Analysis}
\label{sec:info}
\label{sec:mechanism}
This section outlines the proofs of Theorems~\ref{thm:randomphase}
and~\ref{thm:prior-effect}. The main tools are the key lemmas of
Section~\ref{lc:key-lemmas}. They give a finite comparison of source
signals, a statistical conversion of finite comparisons into risk
bounds, and an evaluation of the resulting power integrals. The analysis
then follows the chain
\[
\begin{gathered}
 \text{moment identity}
 \ \Longrightarrow\ \text{inverse bound and prior-dominated paths}\\
 \ \Longrightarrow\ \text{two risk integrals}
 \ \Longrightarrow\ \text{Theorems~\ref{thm:randomphase} and~\ref{thm:prior-effect}}.
\end{gathered}
\]
The statistical conversion is classical. The model-specific work is the
uniform finite comparison, the integration over difficult phases, and
the control of both endpoints under the original prior.
Appendices~\ref{geom:source-lower}, \ref{sec:weighted-phase-details}
and~\ref{prf:k3-singular} establish the prior-specific phase averages
and prior-preserving paths, with uniform constants.

\subsection{Key Lemmas}
\label{lc:key-lemmas}
The key step is the finite--tangent comparison of
Lemma~\ref{geom:prony-secant}. It rests on a single moment identity
that holds uniformly near degenerate phases. The identity controls both
finite differences of source signals and their derivatives.
Lemma~\ref{lem:bridge} transfers these moment bounds to response norms.
The tangent form of the identity gives the efficient information. Its
finite form gives the inverse bound for the upper risk bound. Fixing
low moments gives the paths for the lower bound.
Lemma~\ref{spine:risk} converts finite comparisons into upper and lower
risk bounds through minimum distance and two-point testing.
Lemma~\ref{risk:power-integral} evaluates the resulting power integrals.

\subsubsection{One Finite Comparison for Signals and Tangents}
\label{app:bridge}

We first transfer moment norms to response norms by Hermite--Newton
interpolation, uniformly over node lists with repetitions.

For $1\le m\le2K$, let $\Pi_m$ denote the complex orthogonal projection
in $L^2(\nu)$ onto the span of $g_0,\ldots,g_{m-1}$, so that
$\Pi=\Pi_{2K}$ in
Section~\ref{sec:projected-estimator}. For nodes $x_i=\tau_0+\eps t_i$,
call $\Lambda$ admissible if it is a complex linear combination of
$f\mapsto f^{(e)}(x_i)$, with $e$ below the multiplicity of $x_i$.
Write $\mathsf v_r=\Lambda[(\tau-\tau_0)^r]$. Thus an atomic combination
has $\mathsf v_r=\eps^r\sum_iA_it_i^r$, while
with $E_j=\psi(\cdot,\tau_0+\eps s_j)$ and
$E'_j=\partial_\tau\psi(\cdot,\tau_0+\eps s_j)$, the combination
$\sum_j\alpha_jE_j+\eps\sum_j\beta_jE_j'$ has
$\mathsf v_r=\eps^r(\sum_j\alpha_js_j^r+r\sum_j\beta_js_j^{r-1})$.

\begin{lemma}[Divided differences control the norm]
\label{lem:bridge}
Let $\nu$ satisfy Assumption~\ref{ass:sampling}, $1\le m\le2K$ and $T<\infty$.
Then there are $\eps_0>0$ and $C<\infty$, depending only on $m$, $T$ and the
constants of the assumption, such that for $0<\eps\le\eps_0$, all
$t_1,\ldots,t_m\in[-T,T]$ and every admissible $\Lambda$ (a shorter list is
padded with points of $[-T,T]$ on which $\Lambda$ puts no weight),
\begin{equation}
\label{eq:bridge}
 C^{-1}|\mathsf v|\le\|\Pi_m\Lambda\psi\|_\nu\le\|\Lambda\psi\|_\nu\le C|\mathsf v|,
 \qquad |\mathsf v|^2=\sum_{r<m}|\mathsf v_r|^2 .
\end{equation}
\end{lemma}

\begin{proof}
Let $H(\tau)=\sum_{j<m}b_j\prod_{a=1}^j(\tau-x_a)$ be the
Hermite--Newton interpolant of $\psi(\cdot,\tau)$ on the node list,
with repetitions. The divided-difference integral over a simplex
of volume $1/j!$ \citep[Eq.~(52)]{deboor2005}, applied by scalarization,
gives
$b_j=g_j+O_{L^2}(\eps)$ and $\|b_j\|\le C$, uniformly even at repeated
nodes. The lower coefficients of $\prod_{a=1}^j(\tau-x_a)$ in powers of
$\tau-\tau_0$ are $O(\eps)$, so
$H(\tau)=\sum_{l<m}(g_l+e_l)(\tau-\tau_0)^l$ with $\|e_l\|\le C\eps$.
Hermite interpolation matches every derivative used by $\Lambda$, hence
\begin{equation}
 \Lambda\psi=\Lambda H
 =\sum_{l<m}g_l\mathsf v_l+O_{L^2}(\eps)|\mathsf v|.
 \label{eq:prf-common-principal}
\end{equation}
The constants are independent of node gaps and the coefficients of
$\Lambda$. The principal part has norm at least
$\sqrt{\lambda_0}|\mathsf v|$ and is fixed by $\Pi_m$. Taking $\eps_0$
sufficiently small proves both norm bounds.
\end{proof}

For the exponential kernel with $\tau_0=0$, $\mathsf v_r$ is the $r$th derivative
at the origin of the exponential sum $\Lambda\psi$ divided by $(i\pi)^r$.

\subsubsection{Prony Coordinates and Finite Differences}
We next write source signals in Prony coordinates, in which finite
differences and tangents satisfy one identity. The Prony mapping and
Prony varieties are studied in
\citet{batenkov2014geometry,goldman2018prony}. Write
$m_l(x)=\sum_a c_as_a^l$,
$\omega_x(t)=\prod_a(t-s_a)=t^K+\sum_{l<K}p_l(x)t^l$,
$M=(m_0,\ldots,m_{K-1})$ and $\xi(x)=(M(x),p(x))$; complex coordinates
are realified, so $\xi$ has dimension $3K$.

\begin{lemma}[The common finite--tangent comparison]
\label{geom:prony-secant}
\label{geom:prony-target}
On the separated normalized source class, $\xi$ is one-to-one with
a smooth local inverse and uniformly bounded Jacobian and inverse
Jacobian on compact admissible subsets.
\begin{enumerate}
\item Define
\[
 \mathcal A_{\eps,x}(\Delta M,\Delta p)=
 \left((\eps^l\Delta M_l)_{l<K},\;
 \left(\eps^{K+q}\sum_{l<K}m_{q+l}(x)\Delta p_l\right)_{q<K}\right).
\]
Uniformly over admissible pairs and sufficiently small $\eps$,
\begin{equation}
 \|\Pi(\mu_{x'}-\mu_x)\|_{\nu_n}
 \asymp\|\mu_{x'}-\mu_x\|_{\nu_n}
 \asymp|\mathcal A_{\eps,x}(\xi(x')-\xi(x))|.
 \label{geom:source-secant}
\end{equation}
\item In a sufficiently small fixed closed Prony box $\mathcal U$ about an
interior configuration, let the ordered roots remain in disjoint
intervals. For every pair in this box,
\[
 |\tau_j(x')-\tau_j(x)|
 \asymp\eps\left|\sum_{l<K}s_j(x)^l\Delta p_l\right|.
\]
If the original prior has a density bounded below by a positive
constant in this neighborhood, uniform probability $\pi_0$ on the
box, pushed back by the inverse chart, is dominated by a fixed
multiple of that prior.
\end{enumerate}
\end{lemma}
\begin{proof}
The ordered roots and Vandermonde inversion give the local inverse
chart, with uniform derivative bounds by separation.
Write $h=\Delta\xi$, $v=(\eps^l\Delta m_l)_{l<2K}$, and let
$(T_\eps(p)z)_l=z_l$ for $l<K$ and
$(T_\eps(p)z)_{K+q}=z_{K+q}+\sum_{l<K}p_l\eps^{K-l}z_{q+l}$ for $q<K$.
This is unit lower triangular, with $T_\eps(p)=I+O(\eps)$ uniformly.
Let $\mathsf E=\operatorname{diag}(I_K,-I_K)$. Subtracting the moment
recurrences, or differentiating one, gives the single identity
\begin{equation}
 T_\eps(p')v=\mathsf E\mathcal A_{\eps,x}h
      =T_\eps(p)DF_\eps(\xi)h,
 \quad F_\eps=(\eps^lm_l)_{l<2K}.
 \label{geom:prony-finite-tangent}
\end{equation}
In unscaled coordinates the finite-difference form reads
\begin{equation}
 \Delta m_{K+q}+\sum_{l<K}p_l'\Delta m_{q+l}
      =-\sum_{l<K}m_{q+l}\Delta p_l,\qquad q<K.
 \label{geom:prony-difference}
\end{equation}
Thus $|v|\asymp|\mathcal A_{\eps,x}h|$. Lemma~\ref{lem:bridge}
on the union of the node lists proves~\eqref{geom:source-secant}.

For item~2,
$\sum_{l<K}s_j^l\Delta p_l
=\omega_{x'}(s_j)=(s_j-s_j')\prod_{a\ne j}(s_j-s_a')$.
The last product is bounded above and away from zero on the root box.
The bounded, nonzero chart Jacobian gives the prior-domination claim.
\end{proof}

The tangent and finite forms of~\eqref{geom:prony-finite-tangent} hold
uniformly over all phases. The fixed-moment paths also require control
of the original prior along the paths, which is verified for each prior
below.

\subsubsection{One Statistical Conversion}
\label{spine:conversion}
\label{geom:section}
\label{unified:section}
This subsection states the statistical step used for all risk bounds
of Section~\ref{sec:main-results}. Let
$x\in\Theta$, let $F$ be a signal map on a search class
$\Theta^+\supset\Theta$, and let $g$ be a real target. For a Borel
probability $\pi$ on $\Theta$, put
$R_\pi^*(r)=\inf_{\widehat g}\int\mathbb E_x|\widehat g-g(x)|^2\pi(dx)$,
where the infimum is over all measurable functions of the full data.
Let $r,L_{\rm KL}>0$, $\bar q\ge0$, and suppose
\begin{equation}
 \operatorname{KL}(P_x,P_{x'})\le L_{\rm KL}r^{-2}\|F(x)-F(x')\|^2
 \ (x,x'\in\Theta),
 \qquad \sup_{x\in\Theta}\mathbb E_x\|T-F(x)\|^2\le\bar q r^2.
 \label{spine:experiment}
\end{equation}
All maps and estimators below are measurable, and path families are
jointly measurable in their label and time.
For sources, $F=\Pi\mu$, $T=\Pi Y$, $r=\zeta_n$ and $\bar q=2K$;
Lemma~\ref{lem:bridge} bounds the KL divergence of the full observations
by the projected signal difference. Thus the lower bound concerns all
estimators of the full data.

\begin{lemma}[Finite comparisons give both risk bounds]
\label{spine:risk}
Assume~\eqref{spine:experiment}.
\begin{enumerate}
\item[(a)] If, for $C_0>0$ and all $x\in\Theta$, $y\in\Theta^+$,
\begin{equation}
 |g(y)-g(x)|\le C_0w_x(\|F(y)-F(x)\|),
 \label{risk:modulus}
\end{equation}
where $w_x\ge0$ is
nondecreasing and $w_x(ts)\le tw_x(s)$ for $t\ge1$, a minimum-distance
fit $\widehat x\in\argmin_{y\in\Theta^+}\|T-F(y)\|$ satisfies
\[
 \mathbb E_x|g(\widehat x)-g(x)|^2
 \le C_0^2(1+4\bar q)w_x(r)^2.
\]
The same conclusion holds, with a constant depending also on
$\alpha_1,\alpha_2$, when the fit obeys
$\|F(\widehat x)-F(x)\|\le\alpha_1\|T-F(x)\|+\alpha_2r$.
\item[(b)] Set $\varrho_0=(2L_{\rm KL})^{-1/2}$. If a finite measure
$Q$ on $\Theta^2$ has both marginals at most $B\pi$, $B>0$, and is supported
on $\|F(x')-F(x)\|\le\varrho_0r$, then
\begin{equation}
 R_\pi^*(r)\ge\frac1{16B}\int|g(x')-g(x)|^2\,Q(dx,dx').
 \label{risk:pairs}
\end{equation}
In particular, let $\gamma_y:[0,1]\to\Theta$ have absolutely
continuous signal and target, a target derivative of fixed sign,
finite label measure $\Upsilon$, and measurable $b_y,d_y\ge0$ such that
\begin{equation}
 \begin{gathered}
 |(g\circ\gamma_y)'|\ge b_y,\qquad
 \|(F\circ\gamma_y)'\|\le d_y\quad\text{a.e.},\\
 \int\!\int_0^1\delta_{\gamma_y(s)}\,ds\,\Upsilon(dy)\le B\pi.
 \end{gathered}
 \label{crit:path-family}
\end{equation}
Then
\begin{equation}
 \frac{c_L}{B}\int b_y^2\min\{1,r^2/d_y^2\}\,\Upsilon(dy)
 \le R_\pi^*(r)
 \le C\int w_x(r)^2\,\pi(dx),
 \label{spine:two-integrals}
\end{equation}
where the upper bound assumes~(a),
$c_L=\min\{1/4,\varrho_0^2\}/32$, and the minimum is one if $d_y=0$.
\end{enumerate}
\end{lemma}
\begin{proof}
Minimum-distance optimality gives
$\|F(\widehat x)-F(x)\|\le2e$, $e=\|T-F(x)\|$.
Hence the squared error is at most
$C_0^2\max\{1,2e/r\}^2w_x(r)^2$, whose expectation proves~(a);
replace $2e$ by $\alpha_1e+\alpha_2r$ for an approximate fit.
Pinsker's inequality and~\eqref{spine:experiment} give
$\operatorname{TV}(P_x,P_{x'})\le1/2$ on the support of $Q$.
The two-point squared-error bound gives a sum of pointwise risks at
least $|g(x')-g(x)|^2/8$. Integrating and using the two marginal
bounds proves~(b).

For paths take $h_y=\min\{1/2,\varrho_0r/d_y\}$, with $h_y=1/2$
when $d_y=0$, and pair
$\gamma_y(s)$ with $\gamma_y(s+h_y)$, with $s$ uniform on $[0,1/2]$.
Both endpoint measures are at most $2B\pi$: for the second, translate
$s$ by the label-dependent constant $h_y$. The signal difference is
at most $\varrho_0r$ and the target difference at least $b_yh_y$.
The pair bound proves~\eqref{spine:two-integrals}.
\end{proof}

Lemma~\ref{spine:risk} is based on minimum distance and two-point
testing \citep[Chapter~2]{tsybakov2009nonparametric}. Its lower bound
uses only the Kullback--Leibler condition. It is a prior-weighted
two-point argument restricted to pairs whose endpoints both carry prior
mass. The extended Ziv--Zakai bound \citep{bell1997extended} is also a
prior-weighted two-point argument. The two integrals need not match for
an arbitrary prior. Section~\ref{sec:criterion} identifies when both are
comparable to the prior average of the capped inverse information.

\subsubsection{Truncated Power Integrals}
Both risk bounds are evaluated through the following scalar integral.
\begin{lemma}[Truncated power integrals]
\label{risk:power-integral}
For fixed $q>0$, $0\le\eta\le1/2$, and $\lambda>0$, put
\[
 I_q(\eta,\lambda)=\int_0^1
       \min\{1,\lambda^2/(t^2+\eta^2)\}t^{q-1}\,dt.
\]
For $0<\lambda\le1$ and $h=\eta\vee\lambda$,
\[
 I_q(\eta,\lambda)\asymp
 \begin{cases}
 \lambda^2h^{q-2},&q<2,\\
 \lambda^2\{1+\log(1/h)\},&q=2,\\
 \lambda^2,&q>2.
 \end{cases}
\]
For $\lambda\ge1$ it is comparable to one. Constants depend only on $q$.
\end{lemma}
\begin{proof}
Split at $h$. The two parts have sum comparable to
$\lambda^2[h^{q-2}+\int_h^1t^{q-3}\,dt]$; integrating proves the formula.
For $\lambda\ge1$, the integrand's minimum lies between $4/5$ and $1$.
\end{proof}

More generally, Tonelli's theorem gives, for any nonnegative $X$,
\begin{equation}
 \mathbb E\min\{1,\lambda^2/X^2\}
 =2\lambda^2\int_\lambda^\infty\mathbb P\{X\le s\}s^{-3}\,ds.
 \label{risk:capped-tail}
\end{equation}

\subsection{Analysis for the Source Model}
\label{lc:source-analysis}
With the key lemmas in hand, we verify their hypotheses for sources:
the efficient information and its phase average, the inverse bound
over the whole candidate class, and paths that keep the prior mass.

\subsubsection{Information and Its Phase Average}
\label{sec:info-def}
\label{sec:results}
\label{sec:avg-info}
Section~\ref{lc:efficient-info} defined the efficient information
$J_j$. We now estimate it uniformly in the phases and average it over
them.

\emph{Uniform two-sided estimate.}
Exponents computed at a fixed configuration do not control phases that
approach a degenerate configuration as $\eps$ decreases. The estimate
below is uniform in the phases because it keeps all $K$ moment orders,
each with its own power of $\eps$. Put
$d_k=\prod_{\ell\ne k}(s_k-s_\ell)$, $q_k=e^{i\varphi_k}$,
write $J_j(\eps,q;\nu)=J_j(\eps,\varphi;\nu)$, and for $b\in\mathbb R^K$
let
\[
 Q_{\eps,q}(b)=\sum_{r=0}^{K-1}\eps^{2r}
       \Bigl|\sum_{k=1}^Kd_ks_k^rq_kb_k\Bigr|^2 .
\]

\begin{lemma}[Uniform moment estimate]
\label{lem:weighted}
Let $\nu$ satisfy Assumption~\ref{ass:sampling}, $\rho>0$, and let the nodes
range over a compact set with $s_{j+1}-s_j\ge\rho$. Then, for all
$q\in(\mathbb S^1)^K$ and $0<\eps\le\eps_0$,
\begin{equation}
 \label{eq:uniform-moment-main}
 J_j(\eps,q;\nu)\asymp\eps^{2K-2}\min_{b_j=1}Q_{\eps,q}(b),
\end{equation}
with $\eps_0$ and constants depending only on $K$, $\rho$, the compact set
and the constants of Assumption~\ref{ass:sampling}, uniformly over all
phases, including $\varphi=\varphi(\eps)$.
\end{lemma}

\begin{proof}See Appendix~\ref{prf:weighted}.\end{proof}
Appendix~\ref{prf:weighted} also gives the Cauchy--Binet ratio of
minors~\eqref{eq:profile-minors} used for the weighted priors of
Section~\ref{sec:prior-mass}.

\emph{The information ellipse.}
For the scaled Vandermonde columns
$h_a=(1,\eps s_a,\ldots,\eps^{K-1}s_a^{K-1})^T$ put
$\mathcal V_j=\spanR\{e^{i\varphi_a}h_a:a\ne j\}$.
By~\eqref{eq:uniform-moment-main}, $J_j(x)\asymp a_j^2\widetilde J_j$ with
$\widetilde J_j=\eps^{2K-2}d_j^2\dist(e^{i\varphi_j}h_j,\mathcal V_j)^2$.
Let $\mathbf R$ be the real $2\times2$ Gram matrix of the projections of
$h_j,ih_j$ onto $\mathcal V_j^\perp$ and $u_\varphi=(\cos\varphi,\sin\varphi)^T$.
Then $\widetilde J_j=\eps^{2K-2}d_j^2u_{\varphi_j}^T\mathbf Ru_{\varphi_j}$: for
fixed nuisance phases $\varphi_{-j}$, the information is a quadratic
form in the target phase, with ellipse matrix $\mathbf R$. Integrating the
target phase gives
\begin{equation}
\label{eq:main-ellipse}
 \mathbb E_{\varphi_j}[\widetilde J_j^{-1}\mid\varphi_{-j}]
 =\frac{\eps^{-2K+2}}{d_j^2\sqrt{\det\mathbf R}} .
\end{equation}

\begin{lemma}[Information ellipse]
\label{geom:ellipse}
\label{lem:phase-integrals}
\begin{enumerate}
\item Uniformly over the node class, the coefficient moduli and all
nuisance phases $\varphi_{-j}$, $\sqrt{\det\mathbf R}=\eps^{K-1}|d_j|\{1+O(\eps)\}$.
\item If $K$ is even, then, uniformly over the node class,
$\mathbb E_{\varphi_{-j}}\lambda_{\max}(\mathbf R)^{-1/2}\le C\eps^{1-K/2}$, with $C$
depending only on $K$, the offset bound $T$ and $\rho$.
\end{enumerate}
\end{lemma}

\begin{proof}See Appendix~\ref{geom:source-lower}.\end{proof}
Write $\widetilde\varphi=\varphi_j-\varphi_0$, where
$u_{\varphi_0}$ is a unit eigenvector for $\lambda_{\min}(\mathbf R)$; then
$u_{\varphi_j}^T\mathbf Ru_{\varphi_j}
=\lambda_{\min}+(\lambda_{\max}-\lambda_{\min})\sin^2\widetilde\varphi$.
The conditional tail has exponent $1/2$ between the two eigenvalue
scales when their ratio is at least eight. This follows from
$\mathbb P\{|\sin\widetilde\varphi|\le z\}=(2/\pi)\arcsin z$;
the global phase average still requires both parts of
Lemma~\ref{geom:ellipse}.

\begin{proposition}[Phase-averaged inverse information]
\label{thm:avg-inverse-info}
Suppose that the designs $(\nu_n)_{n\ge1}$ satisfy
Assumption~\ref{ass:sampling} uniformly. For every fixed
$K\ge2$, every $j\le K$, and sufficiently small $\eps$, uniformly over
the admissible nodes and coefficient moduli in~\eqref{eq:class}, with
independent uniform phases,
\begin{equation}
\label{eq:random-phase-information}
 \mathbb E_\varphi[J_j(x)^{-1}]\asymp\eps^{-(3K-3)}.
\end{equation}
\end{proposition}

\begin{proof}
Lemma~\ref{geom:ellipse}(1) and~\eqref{eq:main-ellipse} give
$\mathbb E_{\varphi_j}[\widetilde J_j^{-1}\mid\varphi_{-j}]
=\eps^{-(3K-3)}|d_j|^{-3}\{1+O(\eps)\}$ uniformly in the nuisance
phases, and $J_j(x)\asymp a_j^2\widetilde J_j$.
\end{proof}

\subsubsection{The Inverse Bound}
\label{sec:statistics}
\label{prf:stat}

Throughout this subsection, $(\nu_n)$ satisfies Assumption~\ref{ass:sampling}
uniformly, and unless otherwise stated the constants and the upper bound
on $\eps$ depend only on $K$, the constants of that assumption and the
parameter class~\eqref{eq:class}. Write $m_\ell(x)=\sum_ac_as_a^\ell$
for the moments of $x=(s,c)\in\mathcal X_K$.

\begin{lemma}[Moment comparison]
\label{lem:moment-comparison}
There is $C<\infty$ such that, for sufficiently small $\eps$, all $n$ and
all $x,x'\in\mathcal X_K$,
\[
 C^{-1}V(x,x')\le\|\Pi(\mu_x-\mu_{x'})\|_{\nu_n}
 \le\|\mu_x-\mu_{x'}\|_{\nu_n}\le CV(x,x'),
\]
where $V(x,x')^2=\sum_{\ell=0}^{2K-1}\eps^{2\ell}|m_\ell(x)-m_\ell(x')|^2$.
\end{lemma}

\begin{proof}
The difference $\mu_x-\mu_{x'}$ combines the response $\psi$ at the $2K$ points
$\tau_0+\eps s_a$ and $\tau_0+\eps s'_a$, and its scaled moments are
$\eps^\ell\{m_\ell(x)-m_\ell(x')\}$; Lemma~\ref{lem:bridge}
with $m=2K$ gives the claim.
\end{proof}

\begin{proposition}[Inverse bound near a configuration]
\label{thm:inverse}
Suppose Assumption~\ref{ass:sampling} holds uniformly for $(\nu_n)$.
For the class $\mathcal X_K$ in~\eqref{eq:class} and projection $\Pi$
defined in Section~\ref{sec:projected-estimator}, there are $c_*,C_*>0$ such
that for small $\eps$, all $n$, all
$x,x'\in\mathcal X_K$ (phases arbitrary, possibly depending on $\eps$) and
$1\le j\le K$,
\[
 \|\Pi(\mu_x-\mu_{x'})\|_{\nu_n}\le c_*\eps\sqrt{J_{\min}(x)}
 \ \Longrightarrow\
 \eps|s_j-s'_j|\le\frac{C_*\|\Pi(\mu_x-\mu_{x'})\|_{\nu_n}}{\sqrt{J_j(x)}} .
\]
\end{proposition}

\begin{proof}See Appendix~\ref{prf:inverse}; the argument uses an annihilating polynomial.\end{proof}

The class diameter extends the inverse bound to every pair:
\begin{equation}
 \label{eq:source-global-modulus}
 |\tau_j(x')-\tau_j(x)|^2
 \le C\min\left\{\eps^2,
 \frac{\|\Pi(\mu_{x'}-\mu_x)\|_{\nu_n}^2}{J_{\min}(x)}\right\},
 \qquad x,x'\in\mathcal X_K.
\end{equation}
Indeed, when the hypothesis of Proposition~\ref{thm:inverse} holds, use
$J_j\ge J_{\min}$; otherwise the second entry is at least
$c_*^2\eps^2$, so the class diameter applies.

\begin{lemma}[Conditional mean-square bound]
\label{cor:conditional-mse}
For the measurable minimizer in~\eqref{eq:estimator} and every
$x\in\mathcal X_K$,
\begin{equation}
 \label{eq:rp-conditional-mse}
 \mathbb E_x|\widehat\tau_j-\tau_j|^2
 \le C\min\{\eps^2,\zeta_n^2/J_{\min}(x)\}.
\end{equation}
\end{lemma}

\begin{proof}
Compactness, continuity and measurable selection
\citep{brown1973measurable} give a measurable minimizer.
Apply Lemma~\ref{spine:risk}(a), with $g=\tau_j$, $F=\Pi\mu$ and
$w_x(r)=\min\{\eps,r/\sqrt{J_{\min}(x)}\}$, to
\eqref{eq:source-global-modulus}, using
$\mathbb E\|\Pi\eta\|_{\nu_n}^2=2K\zeta_n^2$.
\end{proof}

Consequently, for every prior $\pi$ on $\mathcal X_K$, the
estimator~\eqref{eq:estimator} has Bayes risk at most
$C\,\mathbb E_\pi\min\{\eps^2,\zeta_n^2/J_{\min}\}$.

\subsubsection{Fixed-Moment Paths}
\label{sec:fiber-paths}
\label{sec:source-verification}

The lower bounds use paths along which the first few complex moments
stay fixed. On such a path the target node moves at speed of order
$\eps$. By Lemma~\ref{lem:bridge}, the mean $\mu$ moves only through the
first moment that changes. That moment carries a high power of
$\eps$. The prior enters only through its density ratio along the path
and through the chart Jacobian. No derivative of the prior density is
taken. Hence the weights $|\mathscr D|^\beta$ with $-1<\beta<0$ and
$u^{-\gamma}$ are allowed. The three source constructions differ only
in how the moving direction keeps the prior density comparable. Affine
fibers serve the uniform law, and a chart that preserves $\mathscr D$
serves $\pi_\beta$. For $\pi_\gamma$ the paths are shears that keep $u$
and $v$ within fixed ratios (see~\eqref{k3p:prior-preservation}).
The affine fibers come from the following consequence of
Lemma~\ref{geom:prony-secant}.

\begin{corollary}[Affine moment fibers]
\label{geom:affine-moment-fiber}
Let $K\le K_M\le\lfloor3K/2\rfloor$ and, in this corollary only,
$M=(m_0,\ldots,m_{K_M-1})$. Let
$\mathcal A(M)_{q,l}=m_{q+l}$ for $q<K_M-K$, $l<K$.
At an interior separated configuration where the realification of $\mathcal A(M)$ has
row rank $2(K_M-K)$, nearby sources have smooth coordinates $(M,t)$, $t\in\mathbb R^{3K-2K_M}$,
in which $p=p_*(M)+N(M)t$ and
$\operatorname{range}N(M)=\ker\mathcal A(M)$.
The chart and inverse have bounded derivatives on fixed smaller
rectangles, independently of $\eps,n$.
\end{corollary}
\begin{proof}
A fixed invertible real minor solves
$\mathcal A(M)p=-(m_K,\ldots,m_{K_M-1})$ affinely.
Lemma~\ref{geom:prony-secant} recovers roots and coefficients from
$p$ and the first $K$ moments. The recurrence then recovers each
remaining entry of $M$, proving that these maps are smooth inverses.
Shrinking the rectangle gives the bounds. For $K_M=K$ there are no
constraints, and the chart is that of Lemma~\ref{geom:prony-secant}.
\end{proof}

\begin{lemma}[Fixed-moment paths]
\label{geom:fiber-paths}
In each case below there are a jointly measurable path family
$\gamma_y:[0,1]\to\operatorname{supp}\pi$ with a finite measure $\Upsilon$ on the labels and
constants independent of $\eps,n,\zeta_n$ such that the target
derivative has a fixed sign on each path, its magnitude is at least
$b_y$, the speed of the mean satisfies
$\|(\mu\circ\gamma_y)'\|_{\nu_n}\le d_y$, and the occupation measure
$\int\!\int_0^1\delta_{\gamma_y(t)}\,dt\,\Upsilon(dy)$ is at most $B\pi$.
\begin{enumerate}
\item[(i)] \emph{Uniform phases, $K$ odd, every $j$.} The paths fix
 $(m_0,\ldots,m_{(3K-3)/2})$ and move along the affine moment fiber of
 Corollary~\ref{geom:affine-moment-fiber}; $b_y=c\eps$,
 $d_y=C\eps^{(3K-1)/2}$, and $\Upsilon$ is a fixed multiple of Lebesgue
 measure on the remaining chart coordinates. Hence
 $R_{\pi,j}\ge c\min\{\eps^2,\zeta_n^2\eps^{-(3K-3)}\}$.
\item[(ii)] \emph{The prior $\pi_\beta$, $K$ even, $\beta>-1$, $j\in A$;
 for $\beta=0$ also uniform phases with $K$ even.} In the chart
 $(\mathfrak b,z,u,v)$ of Lemma~\ref{weighted:moment-chart}, the paths
 are $u$-translations at fixed $(\mathfrak b,z,v)$. They fix
 $(m_0,\ldots,m_{3K/2-2})$ and change $m_{3K/2-1}$ at rate
 $i\mathfrak b$; the target speed is at least $c\eps$,
 $d_y=C\eps^{3K/2-1}(\mathfrak b^2+\eps^2)^{1/2}$, and
 $\Upsilon(d\mathfrak b\,dz\,dv)=|\mathfrak b|^\beta\,d\mathfrak b\,dz\,dv$.
 Hence, with $\lambda=\zeta_n/\eps^{3K/2-1}$ (and $R_{\pi,j}$ in place
 of $R_{\pi_\beta,j}$ for uniform phases),
 $R_{\pi_\beta,j}\ge c\eps^2\int_0^{\mathfrak b_0}\min\{1,\lambda^2/(\mathfrak b^2+\eps^2)\}\mathfrak b^\beta\,d\mathfrak b$.
\item[(iii)] \emph{The prior $\pi_\gamma$, $K=3$, every $j$.} On each
 dyadic layer $t_l=t_{\max}2^{-l}\ge\eps$, the paths are unit-Jacobian
 shears in the Prony coordinates
 $\xi=(m_0,\ldots,m_{K-1},p_0,\ldots,p_{K-1})$ (moments and coefficients
 of the node polynomial, Lemma~\ref{geom:prony-secant}) that fix
 $(m_0,\ldots,m_3)$ and change $m_4$ at rate $O(t_l)$; the target speed
 is at least $c\eps$, $d_y=C\eps^4t_l$, and $\Upsilon=\sum_l\pi_l$,
 where $\pi_l$ is $\pi_\gamma$ restricted to the layer, of mass of order
 $t_l^{3-2\gamma}$. Hence, with $\lambda=\zeta_n/\eps^4$,
 $R_{\pi_\gamma,j}\ge c\eps^2\sum_l t_l^{3-2\gamma}\min\{1,\lambda^2/t_l^2\}$.
\end{enumerate}
\end{lemma}

\begin{proof}
Case~(i) is proved in Appendix~\ref{geom:source-lower}. Case~(ii) is
Lemma~\ref{weighted:moment-chart} with Step~4 of
Appendix~\ref{sec:weighted-phase-details}; both are stated for every
$\beta>-1$ and use only the chart and the prior density
of Section~\ref{sec:power-weights}, not the risk bound of
Theorem~\ref{thm:prior-effect}. For $\beta=0$, the local prior is
dominated by the uniform-phase prior
(Appendix~\ref{geom:source-lower}). Case~(iii) is
Appendix~\ref{geom:k3-paths}.
\end{proof}

\subsubsection{Closing the Two Risk Bounds}
\label{sec:randomphase-proof}
Theorems~\ref{thm:randomphase} and~\ref{thm:prior-effect} follow by
evaluating both sides of~\eqref{spine:two-integrals}, with the common
upper modulus $w_x(r)=\min\{\eps,r/\sqrt{J_{\min}(x)}\}$.

\emph{Uniform phases.}
Proposition~\ref{thm:avg-inverse-info},
$J_{\min}^{-1}\le\sum_aJ_a^{-1}$ and Lemma~\ref{cor:conditional-mse}
prove the upper bound of Theorem~\ref{thm:randomphase} for odd $K$; it
matches case~(i) of Lemma~\ref{geom:fiber-paths}. For even $K$,
splitting the integral over the target phase at
$|\sin\widetilde\varphi|=\{t/(L\lambda_{\max}(\mathbf R))\}^{1/2}$, with
$\widetilde\varphi$ as in Section~\ref{sec:info-def}, gives
\[
 \mathbb E_{\varphi_j}\min\{L,t/(u_{\varphi_j}^T\mathbf Ru_{\varphi_j})\}
 \le C\min\bigl\{L,\ t/\sqrt{\det\mathbf R},\ \sqrt{Lt/\lambda_{\max}(\mathbf R)}\bigr\}.
\]
With $L=\eps^2$ and $t=\zeta_n^2\eps^{-2K+2}d_j^{-2}$,
Lemma~\ref{geom:ellipse}(2) bounds the nuisance-phase average of the
last term by $C\zeta_n\eps^{3-3K/2}$, which gives the intermediate range.
The right side of~\eqref{eq:rp-conditional-mse} is at most
$C\sum_a\min\{\eps^2,\zeta_n^2/J_a(x)\}$, so the same bounds apply to
least squares. The integral in case~(ii) with $\beta=0$ is
$\eps^2I_1(\eps,\lambda)$, $\lambda=\zeta_n/\eps^{3K/2-1}$, and
Lemma~\ref{risk:power-integral} with $q=1$ gives
$\eps^2I_1(\eps,\lambda)\asymp\Psi_K(\eps,\zeta_n)$, which proves the
even-$K$ lower bound.
The paths give the integral over $[0,\mathfrak b_0]$. Since
$f(\mathfrak b)=\min\{1,\lambda^2/(\mathfrak b^2+\eps^2)\}$ is
nonincreasing, the substitution $\mathfrak b=\mathfrak b_0t$ gives, for
$\mathfrak b_0\le1$,
$\int_0^{\mathfrak b_0}f(\mathfrak b)\mathfrak b^\beta\,d\mathfrak b
=\mathfrak b_0^{\beta+1}\int_0^1f(\mathfrak b_0t)t^\beta\,dt
\ge\mathfrak b_0^{\beta+1}\int_0^1f(t)t^\beta\,dt$, and the integral
over $[0,1]$ is a lower bound if $\mathfrak b_0>1$. This replacement
imposes no condition on $\eps$.

\emph{Weighted and singular priors.}
For $\pi_\beta$ the two-sided
estimate~\eqref{eq:local-even-information-normal-form} gives
$J_{\min}\asymp J_j\asymp\eps^{3K-4}(\mathscr D^2+\eps^2)$ for $j\in A$,
and both integrals reduce to
$\eps^2 I_{\beta+1}(\eps,\zeta_n/\eps^{3K/2-1})$ (Steps~3--4 of
Appendix~\ref{sec:weighted-phase-details}). For $\pi_\gamma$, the
surrogate~\eqref{k3p:information} is integrated
in~\eqref{k3p:upper-integral}--\eqref{k3p:tail-profile}, and the dyadic
sum of case~(iii) has the same distribution-function order, that of
$\eps^2 I_{3-2\gamma}(\eps,\zeta_n/\eps^4)$.
Lemma~\ref{risk:power-integral} evaluates all these branches, including
the logarithmic cases; Table~\ref{geom:verification-table} records the
inputs for each prior.
Section~\ref{lc:criterion-sources} relates these integrals to the
capped Cram\'er--Rao bound.

\begin{table*}[t]
\centering
\caption{The two inputs to Lemma~\ref{spine:risk} for the source priors.
In every source row the inverse bound
is~\eqref{eq:source-global-modulus} with $S=J_{\min}$ and $L\asymp\eps^2$,
and the lower bounds come from prior-dominated paths. The rows cover
only the targets of Theorems~\ref{thm:randomphase}
and~\ref{thm:prior-effect}; for $\pi_\beta$ and its mixtures, the risk
order for targets $j\notin A$ is open.}
\label{geom:verification-table}
\small
\begin{tabular}{@{}p{0.22\linewidth}p{0.24\linewidth}p{0.2\linewidth}p{0.26\linewidth}@{}}
\toprule
Model and prior & Surrogate & Paths & Matching\\
\midrule
Uniform phases, $K$ odd
 & ellipse, Lemma~\ref{geom:ellipse}(1)
 & Lemma~\ref{geom:fiber-paths}(i)
 & $\sum_a\min\{\eps^2,\zeta_n^2\mathbb E J_a^{-1}\}$,
   Proposition~\ref{thm:avg-inverse-info}\\
Uniform phases, $K$ even
 & ellipse, Lemma~\ref{geom:ellipse}(1)--(2)
 & Lemma~\ref{geom:fiber-paths}(ii), $\beta=0$
 & capped ellipse integral, Section~\ref{sec:randomphase-proof}\\
Bounded phase density ratio
 & as above & as above
 & Proposition~\ref{prop:bounded-phase-prior}\\
$\pi_\beta$, $K$ even, $j\in A$
 & \eqref{eq:local-even-information-normal-form}
 & Lemma~\ref{geom:fiber-paths}(ii)
 & Steps~3--4 of Appendix~\ref{sec:weighted-phase-details}\\
$\pi_\gamma$, $K=3$
 & \eqref{k3p:information}
 & Lemma~\ref{geom:fiber-paths}(iii)
 & \eqref{k3p:upper-integral}--\eqref{k3p:tail-profile}\\
Mixtures, $j\in A$
 & components & components
 & Corollary~\ref{cor:full-support-weighted-phase}\\
\bottomrule
\end{tabular}
\end{table*}

\section{Fisher Information and the Bayes Risk: A Criterion}
\label{sec:criterion}

This section abstracts the proof chain of Section~\ref{sec:info}. It
gives conditions under which the Bayes risk of a scalar
target has the order of the capped Cram\'er--Rao bound
$\mathbb E_\pi\min\{L,r_n^2/I\}$, also when the information degenerates
on part of the prior's support. It compares the two bounds of
Lemma~\ref{spine:risk} with efficient information through two
conditions on an information surrogate. The first is an inverse bound.
The second is a matching lower bound from prior-dominated paths.
Section~\ref{sec:info} supplies both conditions for sources.
Section~\ref{factor:section} supplies them for a one-factor model whose
information also degenerates through cancellation among the unknowns.

\subsection{Setting}
\label{crit:setting}

Use the experiment and risk of Section~\ref{spine:conversion}, with
$P_x=P_{n,x}$, $T=T_n$ and $r=r_n$, where
$\Theta\subset\Theta^+$ are compact metric spaces and
$F:\Theta^+\to\mathbb R^q$ or $\mathbb C^q$ and
$g:\Theta^+\to\mathbb R$ are continuous. For path families as in
Lemma~\ref{spine:risk}(b) put
$\mathcal L_\Upsilon(r)=\int b_y^2\min\{1,r^2/d_y^2\}\,\Upsilon(dy)$,
the minimum being $1$ when $d_y=0$.

\emph{Gaussian observations.} Let $Y=\mu(x)+\zeta Z$ with standard real
Gaussian noise $Z$, or circular complex Gaussian noise $Z$ with
$\mathbb E|Z_k|^2=1$, so that
$\operatorname{KL}(P_x,P_{x'})=\varsigma\|\mu(x)-\mu(x')\|^2/(2\zeta^2)$
with $\varsigma=1$ or $2$. For a linear map $\Pi$ of norm at most one and
rank $q$ (the complex rank for complex $F$), take $r_n=\zeta$, $F=\Pi\mu$ and $T_n=\Pi Y$. Then
\eqref{spine:experiment} holds with $\bar q=q$ and
$L_{\rm KL}=\varsigma C_1^2/2$ whenever
\begin{equation}
 \|\mu(x)-\mu(x')\|\le C_1\|\Pi(\mu(x)-\mu(x'))\|,\qquad x,x'\in\Theta;
 \label{crit:stat-comparison}
\end{equation}
for sources this is Lemma~\ref{lem:moment-comparison}. The dimension
enters only through the rank $q$ of $\Pi$.

\emph{Cap and local modulus.} Put $L=(\operatorname{diam}g(\Theta^+))^2$.
For $x\in\Theta$ and $\delta>0$ the local modulus of the target with
respect to the mean map $F$ is
\[
 \varpi_x(\delta)=\sup\{|g(y)-g(x)|:\ y\in\Theta^+,\
 \|F(y)-F(x)\|\le\delta\}.
\]
Its envelope $\bar\varpi_x(\delta)=\sup_{u>0}\varpi_x(u)\min\{1,\delta/u\}$
is the least majorant of $\varpi_x$ that is nondecreasing and satisfies
$\bar\varpi_x(t\delta)\le t\bar\varpi_x(\delta)$ for $t\ge1$.

\emph{Efficient information.} For the information statements, $\Theta^+$
lies in an open set $U\subset\mathbb R^p$, $F$ and $g$ are $C^1$ on a
neighborhood of $\Theta$, and $\mathcal M\subset U$ is a $C^1$ embedded
submanifold containing $\Theta$ (for instance $\mathcal M=U$ when $\pi$
has a density). The efficient information of $g$ relative to
$\mathcal M$ is
\[
 I(x)=\inf\{\|DF(x)h\|^2:\ h\in T_x\mathcal M,\ Dg(x)h=1\}\in[0,\infty],
\]
with $\inf\emptyset=\infty$. For measurable $S:\Theta\to[0,\infty]$ the
capped Cram\'er--Rao bound is
\begin{equation}
 \mathcal A_\pi(r;S)=\int\min\{L,r^2/S(x)\}\,\pi(dx),
 \qquad r^2/0=\infty,\ r^2/\infty=0 .
 \label{crit:capped-functional}
\end{equation}

\begin{remark}[Relation to Fisher information]
\label{crit:fisher-rel}
Suppose that the Fisher information matrix $\mathcal I_x$ satisfies
\[
 c_Fr_n^{-2}\|DF(x)h\|^2\le h^T\mathcal I_xh\le C_Fr_n^{-2}\|DF(x)h\|^2,
 \qquad h\in T_x\mathcal M .
\]
The efficient Fisher information of $g$ on
$\mathcal M$ and $r_n^{-2}I(x)$ are infima of these two quadratic forms
over one affine set. Hence the former lies between $c_Fr_n^{-2}I(x)$
and $C_Fr_n^{-2}I(x)$. Thus $\mathcal A_\pi(r_n;I)$ is comparable to the
prior average of the Cram\'er--Rao bound for $g$ capped at $L$. For
Gaussian observations, differentiating~\eqref{crit:stat-comparison} at
relative-interior points of $\Theta$ in $\mathcal M$ gives
$c_F=\varsigma$ and $C_F=\varsigma C_1^2$. The comparison therefore
holds $\pi$-almost everywhere when $\pi$ charges no relative-boundary
point of $\Theta$ in $\mathcal M$. For sources take $g=\tau_j$,
$F=\Pi\mu$ and $r_n=\zeta_n$. Lemma~\ref{lem:profiled} identifies
$J_j(x)$ as $\inf\{\|D\mu_xh\|_{\nu_n}^2:Dg(x)h=1\}$. Differentiating
Lemma~\ref{lem:moment-comparison} along curves in $\mathcal X_K$ gives
$\|D\mu_xh\|\le C^2\|\Pi D\mu_xh\|$. Since $\|\Pi\|\le1$,
$C^{-4}J_j\le I\le J_j$ at interior points.
\end{remark}

\subsection{The Criterion}
\label{crit:criterion}

Both conditions are stated through an \emph{information surrogate}.
This is a measurable $S:\Theta\to[0,\infty]$ used in place of $I$
in~\eqref{crit:capped-functional}. The \emph{inverse bound} is global
in $y$. It requires a constant $C_U$ such that
\begin{equation}
 |g(y)-g(x)|\le C_U\min\Bigl\{\sqrt L,\frac{\|F(y)-F(x)\|}{\sqrt{S(x)}}\Bigr\},
 \qquad x\in\Theta,\ y\in\Theta^+ .
 \label{crit:profile}
\end{equation}
Here $\delta/0=\infty$ for every $\delta\ge0$. The
\emph{matching condition} asks for a path family with
\begin{equation}
 \mathcal A_\pi(r_n;S)\le C_M\,\mathcal L_\Upsilon(r_n).
 \label{crit:matching}
\end{equation}
All orders below are uniform over any family of models, priors and
scales $r_n$ along which $L_{\rm KL}$, $\bar q$, $C_U$, $C_M$, $B$ and
the constant $C_S$ of Theorem~\ref{risk:transfer}(b) stay bounded. In
particular, the matching condition is required with $C_M$ and $B$
independent of $r_n$. The following theorem states the criterion. It is
the main result of this section.

\begin{theorem}[Inverse bounds, prior-dominated paths and the
capped Cram\'er--Rao bound]
\label{risk:transfer}
Assume~\eqref{spine:experiment}.
\begin{enumerate}
\item[(a)] \emph{Paths and the efficient information.} If each $\gamma_y$ is absolutely
continuous in a chart of $\mathcal M$, then
$\mathcal L_\Upsilon(r)\le B\,\mathcal A_\pi(r;I)$ for every $r>0$.
\item[(b)] \emph{Capped Cram\'er--Rao criterion.} Suppose
\eqref{crit:profile} holds, \eqref{crit:matching} holds for a path
family as in~(a), and $S\le C_SI$ $\pi$-almost everywhere with
$C_S\ge1$. Then
\[
 R_\pi^*(r_n)\asymp\mathcal A_\pi(r_n;I)\asymp\mathcal A_\pi(r_n;S)
 \asymp\mathcal L_\Upsilon(r_n),
\]
and the minimum-distance estimator of Lemma~\ref{spine:risk}(a), which
does not use $\pi$, attains this order.
\item[(c)] \emph{Canonical modulus.} The envelope $w_x=\bar\varpi_x$ is
admissible in Lemma~\ref{spine:risk}(a) with $C_0=1$.
\end{enumerate}
The constants in~(b) depend only on $L_{\rm KL},\bar q,C_U,B,C_M,C_S$.
\end{theorem}
\begin{proof}See Appendix~\ref{risk:transfer-proof}.\end{proof}

Lemma~\ref{crit:SleI} gives $C_S=\max\{1,C_U^2\}$ at relative-interior
points, hence $\pi$-almost everywhere when $\pi$ charges no
relative-boundary point of $\Theta^+\cap\mathcal M$.

\begin{lemma}[The surrogate does not exceed the information]
\label{crit:SleI}
Let~\eqref{crit:profile} hold at $x\in\Theta$ with $C_U>0$, and let $x$ lie in the
relative interior of $\Theta^+\cap\mathcal M$ in $\mathcal M$. Then
$S(x)\le C_U^2I(x)$.
\end{lemma}
\begin{proof}See Appendix~\ref{crit:proofs}.\end{proof}

\begin{proposition}[Sufficient conditions]
\label{crit:suff}
\begin{enumerate}
\item[(a)] \emph{A sufficient condition for~\eqref{crit:profile}.} Suppose there are
measurable $S_t,S_v\ge0$ and constants $c_v>0$, $C_Q<\infty$ such that,
for $x\in\Theta$ and $y\in\Theta^+$,
\[
 \|F(y)-F(x)\|\le c_v\sqrt{LS_v(x)}
 \ \Longrightarrow\
 |g(y)-g(x)|\le C_Q\frac{\|F(y)-F(x)\|}{\sqrt{S_t(x)}} ,
\]
with $\delta/0=\infty$. Then~\eqref{crit:profile} holds with $S=\min\{S_t,S_v\}$ and
$C_U=\max\{1,C_Q,1/c_v\}$.
\item[(b)] \emph{A sufficient condition for~\eqref{crit:matching}.} Suppose a path family with
occupation bound $B$ satisfies, for constants $c_P,b_P>0$,
$C_P,C_G<\infty$ and a Borel set $G\subset\Theta$,
\[
\begin{gathered}
 b_y^2\ge c_PL,\qquad d_y^2\le C_PL\,S(\gamma_y(s))\ \text{for a.e. }s,\\
 \int\!\int_0^1\delta_{\gamma_y(s)}\,ds\,\Upsilon(dy)\ge b_P\,\pi|_G,
\end{gathered}
\]
and
\[
 \int_{\Theta\setminus G}\min\{L,r_n^2/S\}\dd\pi
 \le C_G\int_G\min\{L,r_n^2/S\}\dd\pi.
\]
Then~\eqref{crit:matching} holds
with $C_M=(1+C_G)\max\{1,C_P\}/(c_Pb_P)$.
\end{enumerate}
\end{proposition}
\begin{proof}See Appendix~\ref{crit:proofs}.\end{proof}

Part~(a) has the form of the source inverse bound of
Proposition~\ref{thm:inverse}, with $S_t=J_j$ and $S_v=J_{\min}$. Then
$S=J_{\min}\le J_j\le C^4I$ by Remark~\ref{crit:fisher-rel}, so
$C_S=C^4$.

\begin{example}[Two limitations of local information]
\label{crit:quadratic-example}
Let $Y=x^2+\zeta Z$, $Z\sim N(0,1)$, $0<\zeta\le1$, and estimate $x$.
In both experiments $I(x)=4x^2$, but
\[
\begin{array}{c|c|c}
 \pi & \mathcal A_\pi(\zeta;I)&R_\pi^*(\zeta)\\ \hline
 \operatorname{Unif}[0,1],\ L=1
   &\zeta-\zeta^2/4&\asymp\zeta^{3/2}\\
 \operatorname{Unif}[-1,1],\ L=4
   &2\zeta-\zeta^2/4&1/3 .
\end{array}
\]
Thus capped inverse information can overestimate or underestimate the
Bayes risk. For $S=I$, the inverse bound holds on $[0,1]$, but matching
cannot hold uniformly as $\zeta\downarrow0$; on $[-1,1]$ the inverse
bound fails because $x$ and $-x$ have identical signals.
\end{example}
\begin{proof}
On $[0,1]$, $\widehat x=\min\{1,\sqrt{Y_+}\}$ has risk at most
$\min\{\zeta,\zeta^2/x^2\}$ at $x$, so $R_\pi^*\le2\zeta^{3/2}$. The
pushforward of $\pi|_{[0,\sqrt\zeta/2]}$ under
$x\mapsto(x,x+\sqrt\zeta/2)$ has marginals at most $\pi$ ($B=1$; here
$L_{\rm KL}=1/2$, so $\varrho_0=1\ge3/4$), signal differences at most
$3\zeta/4$ and integrated squared target difference $\zeta^{3/2}/8$,
so~\eqref{risk:pairs} gives $R_\pi^*\ge\zeta^{3/2}/128$. Also
$|y-x|\le|y^2-x^2|/x$ gives~\eqref{crit:profile} with $C_U=2$, and
Theorem~\ref{risk:transfer}(b) rules out uniform matching. On $[-1,1]$
the posterior is even, so $R_\pi^*=\mathbb E x^2=1/3$. The
inverse-information integrals split at $\zeta/2$ and $\zeta/4$.
\end{proof}

\subsection{When the Averaged Cram\'er--Rao Bound Suffices}
\label{crit:kappa}

The capped Cram\'er--Rao bound~\eqref{crit:capped-functional} caps the bound at
each configuration before averaging. The averaged bound
$\min\{L,r^2\mathbb E_\pi[1/S]\}$ caps only after averaging. Their
relation depends only on the law of $S$ under $\pi$. We say that $S$ has
\emph{lower-tail exponent} $\kappa>0$ under $\pi$ between scales
$0\le s_1\le s_0/8$ if $S\ge s_1$ $\pi$-almost surely and there are
constants $0<c_-\le\min\{1,C_+\}$ with
\begin{equation}
 \pi\{S<u\}\le C_+(u/s_0)^\kappa\quad(u>0),\qquad
 \pi\{S<u\}\ge c_-(u/s_0)^\kappa\quad(2s_1\le u\le s_0).
 \label{crit:lower-tail}
\end{equation}
With $S=I$, $\kappa$ is the lower-tail exponent of the Fisher
information under the prior. Up to a factor $2$, $s_1$ is the smallest
value of $S$. Its essential infimum under $\pi$ lies in $[s_1,2s_1)$, or
is $0$ if $s_1=0$.

\begin{proposition}[Lower-tail law]
\label{prop:kappa-law}
Let $S\ge0$ be measurable, $L,t>0$,
$A(t)=\mathcal A_\pi(t^{1/2};S)=\mathbb E_\pi\min\{L,t/S\}$ and
$\bar A(t)=\min\{L,t\,\mathbb E_\pi[1/S]\}$.
\begin{enumerate}
\item[(a)] $A(t)\le\bar A(t)$. With $F_S(u)=\pi\{S<u\}$,
\[
 \begin{gathered}
 A(t)=L\pi\{S\le t/L\}+t\int_{\{S>t/L\}}S^{-1}\dd\pi
 =t\int_{t/L}^\infty\frac{F_S(u)}{u^2}\dd u,\\
 t\,\mathbb E_\pi[1/S]=t\int_0^\infty\frac{F_S(u)}{u^2}\dd u .
 \end{gathered}
\]
\item[(b)] Under~\eqref{crit:lower-tail}, put $H=\max\{s_1,t/L\}$. For
$t\le Ls_0$,
\[
 A(t)\asymp\begin{cases}
 ts_0^{-\kappa}H^{\kappa-1},&\kappa<1,\\
 (t/s_0)\{1+\log(s_0/H)\},&\kappa=1,\\
 t/s_0,&\kappa>1,
 \end{cases}
\]
and $t\,\mathbb E_\pi[1/S]$ has the same orders with $H$ replaced by
$s_1$ (and is infinite if $s_1=0$ and $\kappa\le1$). For $t\ge Ls_0$,
$A(t)\asymp\bar A(t)\asymp L$; for $t\le Ls_1$, $A(t)=t\,\mathbb E_\pi[1/S]$.
\item[(c)] Consequently there is $C$, depending only on
$(\kappa,c_-,C_+)$ and not on $L$, $s_0$ or $s_1$, such that $\bar A(t)\le CA(t)$ for all $t>0$ and
every law of $S$ satisfying~\eqref{crit:lower-tail} with these constants
and some $s_1\in[0,s_0/8]$, if and only if $\kappa>1$.
\end{enumerate}
The constants in~(b) depend only on $(\kappa,c_-,C_+)$.
\end{proposition}
\begin{proof}See Appendix~\ref{crit:proofs}.\end{proof}

Part~(a) is~\eqref{risk:capped-tail} after a change of variable;
part~(b) is a power integral. Under Theorem~\ref{risk:transfer}(b), the
same orders apply to the risk. Part~(c) concerns uniformity as
$s_1\downarrow0$; its proof constructs laws for which uniform
equivalence fails when $\kappa\le1$. For fixed $s_1>0$ equivalence always
holds, but its constant may grow with $s_0/s_1$. For sources this ratio
is a power of $\eps$. The intermediate noise exponent is $2\kappa$
when $\kappa<1$, with a logarithmic correction at $\kappa=1$.

\begin{lemma}[Lower-tail exponent from a vanishing order]
\label{crit:gauge}
Let $S\ge0$ and $\mathsf r$ be measurable, with
$\mathsf r_1\le\mathsf r\le1$ $\pi$-almost surely for some
$\mathsf r_1\ge0$. Suppose that, for $\mathsf p,\mathsf d,s_0>0$ and
constants $0<c_1\le C_1$, $0<c_2\le C_2$,
\[
 c_1s_0\mathsf r^{\mathsf p}\le S\le C_1s_0\mathsf r^{\mathsf p}\
 \text{a.s.},\qquad
 \pi\{\mathsf r<t\}\le C_2t^{\mathsf d}\ (t>0),\qquad
 \pi\{\mathsf r<t\}\ge c_2t^{\mathsf d}\ (2\mathsf r_1\le t\le1),
\]
and put $\kappa=\mathsf d/\mathsf p$ and
$s_1=s_0\mathsf r_1^{\mathsf p}$. If
$s_1\le\min\{2^{-\mathsf p-2},1/8\}s_0$, the conclusions of
Proposition~\ref{prop:kappa-law}(b) hold for $S$, with $\asymp$ in place
of the equality in its last clause and with constants depending only on
$\mathsf p,\mathsf d,c_1,C_1,c_2,C_2$. Otherwise
$c_1\min\{2^{-\mathsf p-2},1/8\}s_0\le S\le C_1s_0$ almost surely.
\end{lemma}
\begin{proof}See Appendix~\ref{crit:proofs}.\end{proof}

Here $\mathsf r$ measures the distance to the set where the information
degenerates, $\mathsf p$ is the order of vanishing and $\mathsf d$ the
exponent of the prior mass of $\{\mathsf r<t\}$.

\begin{remark}[Prior mass and the risk exponents]
\label{rem:kappa}
In every row of Table~\ref{tab:kappa} the surrogate satisfies the
hypotheses of Lemma~\ref{crit:gauge}. It is comparable to the efficient
information. For the factor model this holds where $|v|$ is bounded
below (Section~\ref{factor:section}).
Proposition~\ref{prop:kappa-law}(b) evaluates $A(t)$ in the noise ranges
of Theorem~\ref{thm:prior-effect}. For even $K$ the power law of row~1
holds only conditionally on $\varphi_{-j}$. Hence
Theorem~\ref{thm:randomphase} uses the full phase integral of
Section~\ref{sec:randomphase-proof}. Rows~1--3 have $L=\eps^2$ and
$t=\zeta_n^2$. They use the supports of Theorem~\ref{thm:prior-effect}
and $j\in A$ in its part~(a). The matrix $\mathbf R$ and the angle
$\widetilde\varphi$ are as after Lemma~\ref{geom:ellipse}. In row~2,
$\mathscr D$ is divided by its supremum, and its mass exponent comes
from Step~1 of the proof in Appendix~\ref{sec:weighted-phase-details}.
In row~3, $u$ and $\bar\varrho_\eps$ are as in
Appendix~\ref{prf:k3-singular}. Row~4 has
$L=(\overline d-\underline d)^2$ and $t=n^{-1}$.
\end{remark}

\begin{table*}[t]
\centering
\caption{Lower-tail exponents of the surrogates
(Lemma~\ref{crit:gauge}): $S\asymp s_0\mathsf r^{\mathsf p}$ with
$\mathsf r_1\le\mathsf r\le1$, $\pi\{\mathsf r<t\}\asymp t^{\mathsf d}$
for $2\mathsf r_1\le t\le1$, $\kappa=\mathsf d/\mathsf p$ and
$s_1=s_0\mathsf r_1^{\mathsf p}$, up to fixed constants; $s_0$ and $s_1$
are the scales of~\eqref{crit:lower-tail}. Remark~\ref{rem:kappa} gives
the setting of each row.}
\label{tab:kappa}
\small
\begin{tabular}{@{}p{0.30\linewidth}lllll@{}}
\toprule
Model and prior & $s_0$ & $\mathsf r$ & $\mathsf r_1$
 & $(\mathsf p,\mathsf d)$ & $\kappa$\\
\midrule
Even $K$, uniform target phase
 & $\lambda_{\max}(\mathbf R)$ &$|\sin\widetilde\varphi|\vee\mathsf r_1$
 & $(\lambda_{\min}/\lambda_{\max})^{1/2}$ & $(2,1)$ & $1/2$\\
Even $K$, local $|\mathscr D|^\beta$ prior
 & $\eps^{3K-4}$ & $|\mathscr D|\vee\eps$ & $\eps$ & $(2,1+\beta)$
 & $(1+\beta)/2$\\
$K=3$, specified $u^{-\gamma}$ prior
 & $\eps^6$ & $\bar\varrho_\eps$ & $\eps$ & $(2,3-2\gamma)$
 & $(3-2\gamma)/2$\\
Calibrated factor model
 & $\|A\|^4$ & $\|a\|/\|A\|$ & $0$ & $(4,m)$ & $m/4$\\
\bottomrule
\end{tabular}
\end{table*}

\subsection{Application to Theorems~\ref{thm:randomphase} and~\ref{thm:prior-effect}}
\label{lc:criterion-sources}
Theorems~\ref{thm:randomphase} and~\ref{thm:prior-effect} are instances
of Theorem~\ref{risk:transfer} with the surrogate $S=J_{\min}=\min_kJ_k$.
Section~\ref{sec:info} supplies both conditions. The inverse
bound~\eqref{eq:source-global-modulus} holds for every prior on
$\mathcal X_K$. The fixed-moment paths of Lemma~\ref{geom:fiber-paths}
lie on the Prony varieties of \citet{goldman2018prony}.
Table~\ref{geom:verification-table} records which result verifies which
condition for each prior. For every prior and target $j$ covered by the
two theorems, Theorem~\ref{risk:transfer}(b) therefore makes the risk
of the order of $\mathbb E_\pi\min\{\eps^2,\zeta_n^2/J_{\min}\}$. Since
$J_{\min}\le J_j\le C^4I$ (Remark~\ref{crit:fisher-rel}) and the capped
functional decreases in the surrogate, the same holds with $J_j$.
Theorem~\ref{risk:transfer}(a) with $F=\mu$, whose efficient
information for $\tau_j$ is $J_j$, also shows that no prior-dominated
path family gives a lower bound above
$B\,\mathbb E_\pi\min\{\eps^2,\zeta_n^2/J_j\}$.

In the language of Proposition~\ref{prop:kappa-law}, the surrogate
$J_{\min}$ has lower-tail exponent $\kappa=(1+\beta)/2$ under
$\pi_\beta$ and $\kappa=(3-2\gamma)/2$ under $\pi_\gamma$, with the
scales of Table~\ref{tab:kappa}. For every $\zeta_n$,
$\mathcal A_\pi(\zeta_n;J_j)\asymp\mathcal A_\pi(\zeta_n;J_{\min})$ with
constants independent of $\eps$. For small $\zeta_n$ both functionals
equal $\zeta_n^2$ times the prior average of the inverse surrogate.
Hence $\mathbb E_\pi[1/J_j]\asymp\mathbb E_\pi[1/J_{\min}]$, and the
statements about the averaged bound hold with the target information
$J_j$. At $\beta=0$ the scales are $s_0=\eps^{3K-4}$ and
$s_1=\eps^{3K-2}$, with $\kappa=1/2$. They agree with row~1 of
Table~\ref{tab:kappa} after multiplication by $\eps^{2K-2}d_j^2$, at
nuisance phases with $\lambda_{\max}(\mathbf R)\asymp\eps^{K-2}$.

\subsection{A One-Factor Model}
\label{factor:section}

Theorem~\ref{risk:transfer} does not use the source geometry. As an
instance outside sources, let $\mathbf X_1,\ldots,\mathbf X_n\in\mathbb R^{m+1}$
be independent $N(0,\Sigma)$ vectors with
\begin{equation}
 \Sigma=\operatorname{diag}(d_0,d_1,\ldots,d_m)+uu^T,\qquad u=(v,a^T)^T .
 \label{factor:covariance}
\end{equation}
The target is the residual variance $d_0$ of variable~$0$, whose loading
is $v$. The variables $1,\ldots,m$ are references with loadings $a$. The
residual variances of the references in
$\mathcal C\subset\{1,\ldots,m\}$ are known from calibration. The
parameters $d_0$, $u$ and the other $d_j$ are unknown, with
$|a_j|\le A_j$, $|v|\le V$ and $d_j\in[\underline d,\overline d]$ for
some $0<\underline d<\overline d$. The prior $\pi$ has a density on this
box that is bounded above and below by positive constants. Let $\mathsf P$ keep the entries
$(j,k)$, $j\ne k$, and $(j,j)$, $j\in\mathcal C$, of a symmetric
$m\times m$ matrix and set the others to zero, and put
\begin{equation}
 S(a)=\|\mathsf P(aa^T)\|_F^2=\|a\|^4-\textstyle\sum_{j\notin\mathcal C}a_j^4 .
 \label{factor:surrogate}
\end{equation}
The matrix $G=\mathsf P(aa^T)$ is read off $\Sigma$ and the known $d_j$,
and $z=(\Sigma_{j0})_{j\ge1}=va$. Since $\mathsf P(zz^T)=v^2G$,
\begin{equation}
 d_0=\Sigma_{00}-z^TGz/\|G\|_F^2\qquad\text{whenever }S(a)>0;
 \label{factor:identification}
\end{equation}
at $S(a)=0$ only $d_0+v^2$ is identified. With all references
calibrated, $S=\|a\|^4$.

As for sources, the information about the target becomes small through
cancellation. A change of $d_0$ can be partly offset by a change of the
other unknowns. Only the place where the information degenerates
differs. For sources it lies near phase configurations at which a moved
node can be compensated. These are the hypersurface $\mathscr D=0$ for
the even-source priors of Section~\ref{sec:power-weights} and the
aligned set for the three-source prior of
Section~\ref{sec:singular-prior-k3}. Here it lies near the zero set of
$S$ at points with $v\ne0$ (Remark~\ref{factor:exact} and the paragraph
after Lemma~\ref{factor:path-family}). This zero set is $\{a=0\}$ when
all references are calibrated. At
$a=0$ the only
entry of $\Sigma$ that involves $d_0$ or $v$ is $\Sigma_{00}=d_0+v^2$,
so for $v\ne0$ a change of $d_0$ is offset exactly by a change of~$v$
that keeps $d_0+v^2$ fixed.

The clipped plug-in estimator forms $\widehat G$ and $\widehat z$ from the sample
second-moment matrix
$\widehat\Sigma_n=n^{-1}\sum_i\mathbf X_i\mathbf X_i^T$ as $G$ and $z$
are formed from $\Sigma$, evaluates the right side
of~\eqref{factor:identification} at $\widehat\Sigma_n$ and clips it to
$[\underline d,\overline d]$, with $(\underline d+\overline d)/2$ when
$\widehat G=0$. Appendix~\ref{factor:proofs} verifies
the conditions of Theorem~\ref{risk:transfer} for the surrogate $S$.

\begin{corollary}[Risk for the target residual variance]
\label{factor:risk}
Let all references be calibrated and $m\ge1$ be fixed. Then
$R_{\pi,d_0}(n)=R_\pi^*(n^{-1/2})$ of Section~\ref{crit:setting}, for
$g=d_0$, satisfies
\begin{equation}
 R_{\pi,d_0}(n)\asymp\mathbb E_\pi\min\{1,n^{-1}\|a\|^{-4}\}
 \asymp\mathcal A_\pi(n^{-1/2};I)\asymp
 \begin{cases}
 n^{-m/4},&m<4,\\
 n^{-1}\log(en),&m=4,\\
 n^{-1},&m>4,
 \end{cases}
 \label{factor:rates}
\end{equation}
attained by the clipped plug-in estimator at $\widehat\Sigma_n$, which
uses neither the prior nor $n$. The constants depend on $m$, the box and
the density bounds.
\end{corollary}

Here $\kappa=m/4$ for $S$ and $I$ (Table~\ref{tab:kappa}), so by
Proposition~\ref{prop:kappa-law} the averaged Cram\'er--Rao bound has
the order of the risk if and only if $m>4$.

\subsection{Comparisons With Prior Works}
\label{lb:criterion-comparison}

Theorem~\ref{risk:transfer} uses classical lower-bound arguments
(Section~\ref{sec:related-works}); we compare it with the closest
benchmarks. For pointwise estimation of a convex function in white noise,
\citet[Theorems~1 and~3]{cailow2015framework} compare the local
two-point minimax benchmark with a squared modulus and give one
estimator attaining it at every function. For general families,
\citet[Proposition~1 and (7)]{duchiruan2024local} give Hellinger-modulus
bounds for the local two-point minimax risk, of matching order under
polynomial growth. Theorem~\ref{risk:transfer} instead controls risk
under a fixed prior. Its part~(a) bounds the path integral by capped
efficient information. Its part~(b) makes the two comparable once the
all-candidate inverse bound and original-prior matching are verified.

The capped decomposition also resembles the interval-error method,
which combines interval-error probabilities with local-MSE
approximations \citep{richmond2006threshold}; here its comparison with
optimal Bayes risk requires the two conditions of
Theorem~\ref{risk:transfer}(b).
\citet[Theorem~6]{benhaim2009bias} show that the number of independent
observations times the optimal-bias bound tends to the prior average of
the trace of the inverse Fisher information of one observation.
\citet[Proposition~3]{aharon2024asymptotically} show that a weighted
Bayesian Cram\'er--Rao bound is asymptotically attained by the maximum
a posteriori estimator when the information is positive and finite.
Both are statements for a fixed model.
Proposition~\ref{prop:kappa-law} instead concerns the information tail
at every noise level, uniformly in the smallest value of the
information.

\section{Adaptation to an Unknown Center, Separation, and Source Count}
\label{sec:consequences}
\label{sec:known-count-adaptation}

In this section, we remove the assumption that the center and the
separation are known, and then the assumption that the source count is
known. Here the estimator knows amplitude bounds $0<a_*\le A_*$, a window
$W=[\bar\tau-D_W/2,\bar\tau+D_W/2]$, and either $K$ or its upper bound
$K_{\max}$, but neither $\tau_0$ nor $\eps$. The truth remains one close
cluster of the form~\eqref{eq:class}; candidates may lie anywhere in $W$.
A window Gram condition replaces Assumption~\ref{ass:sampling}; unknown
$K$ additionally requires a count-selection threshold.

\subsection{A Uniform Gram Condition on a Window}
\label{revs1a:window-comparison}

The comparison of signals must now hold for all candidate
configurations in $W$. For an ordered list $t_1\le\cdots\le t_N$ in $W$,
repetitions allowed, let $v_j(t)=[t_1,\ldots,t_j]\psi$ be the divided
differences in $\tau$, with Hermite divided differences at repeated
knots.

\begin{assumptionW}[Uniform divided-difference Gram bound on a window]
\label{scope:AW}
Fix an integer $K_{\max}\ge1$, put $N=2K_{\max}$, and let $W$ lie in the
interior of a compact interval $W^+$. For some $A_\psi<\infty$ and
$\lambda_G>0$: the curve $\tau\mapsto\psi(\cdot,\tau)$ is $C^N$ in
$L^2(\nu)$ on a neighborhood of $W^+$, with
$\|\partial_\tau^r\psi(\cdot,\tau)\|_\nu\le A_\psi$ for $\tau\in W^+$ and
$r\le N$; and for every ordered list $t\in W^N$,
\[
 G_W(t;\nu)=\Bigl(\int\overline{v_i(t)}\,v_j(t)\dd\nu\Bigr)_{i,j\le N}
 \succeq\lambda_GI_N .
\]
Uniformity over a design family is understood as in
Assumption~\ref{ass:sampling}.
\end{assumptionW}

\begin{definition}[Moment-norming projection]
\label{scope:proj}
For distinct points $w_1,\dots,w_p\in W$ with multiplicities
$k_1,\dots,k_p\ge1$ and $\sum_ik_i\le N$, let
$\Lambda$ be a complex combination of the functionals
$f\mapsto f^{(l)}(w_i)/l!$, $l<k_i$, $i\le p$ (for example, the
difference of two configurations with at most $K_{\max}$ nodes each in
$W$). Write $\Lambda\psi\in L^2(\nu)$ for $\Lambda$ applied to
$\tau\mapsto\psi(\cdot,\tau)$, and put
$m(\Lambda)=(\Lambda[(\tau-\bar\tau)^r])_{r<N}$. An orthogonal projection
$\Pi_W$ in $L^2(\nu)$, chosen without reference to the data or the true
configuration (it is determined by $\psi$, $\nu$ and $W$), is a
\emph{moment-norming projection} of (complex) rank at most $q_W$ if
$\operatorname{rank}\Pi_W\le q_W$ and, for constants
$0<c_\Pi\le C_\Pi$ and every such $\Lambda$,
\[
 c_\Pi\|m(\Lambda)\|_2\le\|\Pi_W\Lambda\psi\|_\nu
 \le\|\Lambda\psi\|_\nu\le C_\Pi\|m(\Lambda)\|_2 .
\]
\end{definition}

Assumption~\ref{scope:AW} implies Assumption~\ref{ass:sampling} at
every center in $W$ for $K\le K_{\max}$; conversely, on a sufficiently
small window the latter gives Assumption~\ref{scope:AW} with
$K_{\max}=K$, and the projection onto the span of the first $2K$ Taylor
coefficients at $\bar\tau$ (the analogue of $\Pi$) is then a
moment-norming projection (Lemma~\ref{scope:small-window}).
Lemma~\ref{scope:gw-bridge} gives a fixed-rank $\Pi_W$ on any window,
and Lemma~\ref{scope:regular-block} gives line-spectrum designs with
spacing $\mathsf h$, $\mathsf hD_W<2$.

\subsection{Estimators and Risk Guarantees}
\label{adapt:known-count-section}

Assume that the designs satisfy
Assumption~\ref{scope:AW} uniformly; it suffices to take $K_{\max}=K$
when $K$ is known. Let $\Pi_W$ be a moment-norming projection
(Definition~\ref{scope:proj}) of rank $q_\Pi\le q_W$.

\emph{Candidate fits.} For $m\ge1$ put
\[
\begin{aligned}
 \mathcal C_{m,W}&=\{(u,c'):u_1\le\cdots\le u_m\in W,\
                a_*\le|c'_l|\le A_*\},\\
 \widehat x^{(m)}&\in\argmin_{y\in\mathcal C_{m,W}}\|\Pi_W(Y-\mu_y)\|_{\nu_n}^2,
\end{aligned}
\]
with repeated nodes allowed, $\mathcal C_{0,W}=\{0\}$ and
$\mathcal C_W=\mathcal C_{K,W}$. For known $K$ the estimator is
$\widehat x=\widehat x^{(K)}$, and $\widehat\tau_j$ is its $j$th node.

\emph{Truth and target order.}
The truth has center $\tau_0\in W$ and nodes
$\tau_j=\tau_0+\eps s_j\in W$, with $(s,c)$ in the
class~\eqref{eq:class} for $\underline a=a_*$, $\overline a=A_*$, and
$c_l=a_le^{i\varphi_l}$. With $\zeta_{*,K}$ as
in Section~\ref{sec:randomphase}, the order~\eqref{eq:random-phase-risk}
is, for both parities,
$\Psi_K(\eps,\zeta)=\min\{\eps^2,\zeta^2\eps^{-(3K-3)},
\eps^2\zeta/\zeta_{*,K}(\eps)\}$. Here and below
$R_{\pi,j}$ is taken under the law $\pi$ of the physical configuration,
as at the end of Definition~\ref{def:bayes-risk}.

For the lower bound in part~(b) below, fix for each target $j$ an
interior node--modulus point with phase zero on a set $A\ni j$ of size
$\lfloor K/2\rfloor+1$ and $\pi/2$ elsewhere, and let $\pi_{0,j}$ be
uniform probability on a small closed box about it in the Prony
coordinates of Lemma~\ref{geom:prony-target}.

\begin{proposition}[Adaptation to the center and the separation]
\label{thm:known-count-adaptation}
\label{prop:knownK-lower}
Fix $j\le K$, and let $(\tau_0,s,c)$ be drawn from a law $\pi$ on the
truths described above (center $\tau_0\in W$, nodes
$\tau_0+\eps s_l\in W$, $(s,c)$ in~\eqref{eq:class}) that does not depend
on $\eps,n,\zeta_n$ or the noise.
\begin{enumerate}
\item[(a)] \emph{Upper bound.} Suppose that, given $(\tau_0,s,a)$,
the phases have a density at most $M_h$ with respect to the uniform law
$U_K$ on the phase torus. The law of $a$ given $(\tau_0,s)$ may depend on
$(\tau_0,s)$; it is arbitrary for odd $K$, and for even $K$ it has a
Lebesgue density at most $M_a$. Then there is $\eps_0>0$ such that, for
$0<\eps\le\eps_0$ and every $\zeta_n>0$,
\begin{equation}
 \mathbb E|\widehat\tau_j-\tau_j|^2\le C\min\Bigl\{D_W^2,\
 \frac{\zeta_n^2}{\eps^{3K-3}},\
 \eps^2\frac{\zeta_n}{\zeta_{*,K}(\eps)}\Bigr\},
 \label{knownK:all-noise-upper}
\end{equation}
which is at most $C\Psi_K(\eps,\zeta_n)$ when
$\zeta_n\le\zeta_{*,K}(\eps)$.
\item[(b)] \emph{Lower bound.} Suppose that, for some $p_0>0$ and
probability laws $Q_j$ on centers $t$ with
$\operatorname{dist}(t,\partial W)\ge\delta_W>0$,
\begin{equation}
 \pi(dt,dx)\ge p_0\,Q_j(dt)\,\pi_{0,j}(dx),\qquad j\le K,
 \label{knownK:prior-domination}
\end{equation}
where the boxes, $p_0$ and $\delta_W$ do not depend on the center. Then
$R_{\pi,j}\ge c\Psi_K(\eps,\zeta_n)$ for
$0<\eps\le\eps_0$ and every $\zeta_n>0$.
\end{enumerate}
Hence, for $\zeta_n\le\zeta_{*,K}(\eps)$, $R_{\pi,j}\asymp\Psi_K$ under
the hypotheses of both~(a) and~(b). The constants and $\eps_0$ may depend on
$K$, the class~\eqref{eq:class},
$a_*,A_*,M_h,M_a$, $W$, the constants of Assumption~\ref{scope:AW} and
Definition~\ref{scope:proj} (including $q_W$), and in~(b) on
$p_0$, $\delta_W$ and the boxes, but not on $\eps,n,\zeta_n$
or the law of $\tau_0$.
\end{proposition}
\begin{proof}See Appendix~\ref{prf:knownK-adaptation}.\end{proof}

Part~(a) is Lemma~\ref{spine:risk}(a) on the search class
$\mathcal C_W$ with $F=\Pi_W\mu$ and $\bar q=q_W$: for every truth $x$
and candidate $y$, Lemma~\ref{knownK:global-inverse} gives
$|u_j-\tau_j|^2\le C\min\{D_W^2,\Delta^2/J_{\min}(x),
\eps\Delta/\sqrt{J_{\min}(x)}\}$ with
$\Delta=\|\Pi_W(\mu_y-\mu_x)\|_{\nu_n}$, so~\eqref{crit:profile}
holds with $S=J_{\min}$ and cap $D_W^2$. For
even $K$, part~(a) integrates the envelope $\bar\varpi_x$ of the local
modulus (Theorem~\ref{risk:transfer}(c)) instead of this surrogate, and the conditional density bound on
the moduli controls the prior average. The lower bound~(b) is an
instance of the path bound~\eqref{spine:two-integrals}, with the paths of
Lemma~\ref{geom:fiber-paths} about each center.
Above $\zeta_{*,K}$ the bound in~(a) grows linearly as
$\eps^2\zeta_n/\zeta_{*,K}$ until it reaches $D_W^2$. With the known
center and scale of~\eqref{eq:estimator}, the risk stays at most of
order $\eps^2$. Section~S2-E of the supplement
reports simulations in this range.

\subsection{Selecting the Source Count}
\label{adapt:count-selection}

For $0\le m\le K_{\max}$ let $r_m$ be the residual norm of the
$m$-source fit above. Given a deterministic noise-scale bound
$\zeta_n\le\bar\zeta\le C_\sigma\zeta_n$ and an error level
$0<\alpha<1/2$, put
\begin{equation}
 \mathsf t=\bar\zeta\{2q_\Pi\log2+2\log(1/\alpha)\}^{1/2}.
 \label{scope:selection-threshold}
\end{equation}
Choose the smallest $m$ with $r_m\le\mathsf t$, merge coincident nodes
of $\widehat x^{(m)}$, and remove zero coefficients. Report the remaining
count $\widehat K$ and ordered nodes; report $\widehat K=0$ if no fit
qualifies. Neither the true $K$, center, nor separation enters this rule.
For a $K$-source truth and target $j\le K$, define
\begin{equation}
 \ell^{\rm cnt}_{K,j}=
 \begin{cases}
 |\widehat\tau_j-\tau_j|^2,&\widehat K=K,\\
 D_W^2,&\widehat K\ne K,
 \end{cases}
 \qquad R^{\rm cnt}_{\pi,j}=\inf_{\widehat K,\widehat\tau}
                   \mathbb E_{\pi,\eta}\ell^{\rm cnt}_{K,j}.
 \label{scope:Rcnt}
\end{equation}

\begin{proposition}[Unknown source count in the selection range]
\label{scope:unknownK}
Let $2\le K\le K_{\max}$, and assume the conditions of
Proposition~\ref{thm:known-count-adaptation}(a), with
Assumption~\ref{scope:AW} imposed through order $2K_{\max}$.
There are $c_{\rm sel}>0$ and $\eps_0\le1$ such that, if
\begin{equation}
 0<\eps\le\eps_0,\qquad 2\mathsf t\le c_{\rm sel}\eps^{2K-2},
 \label{scope:selection-range}
\end{equation}
then $\mathbb P_x(\widehat K\ne K)\le\alpha$ for every truth $x$, and
\begin{equation}
 \mathbb E\ell^{\rm cnt}_{K,j}\le
 C\min\left\{D_W^2,\frac{\zeta_n^2}{\eps^{3K-3}},
          \frac{\eps^2\zeta_n}{\zeta_{*,K}(\eps)}\right\}
       +D_W^2\alpha.
 \label{scope:count-upper}
\end{equation}
If also~\eqref{knownK:prior-domination} holds, then
$R^{\rm cnt}_{\pi,j}\ge R_{\pi,j}\ge c\Psi_K(\eps,\zeta_n)$.
Consequently, under~\eqref{knownK:prior-domination}, the selector
attains $R^{\rm cnt}_{\pi,j}\asymp\Psi_K(\eps,\zeta_n)$ when
\eqref{scope:selection-range} holds,
$\zeta_n\le\zeta_{*,K}(\eps)$, and
$\alpha=\min\{1/4,\bar\zeta^2/4\}$. Here $\pi$ is a law of $K$-source
truths.
Constants may additionally depend on $K_{\max}$ and $C_\sigma$, but not
on $\eps,n,\zeta_n$ or $\alpha$.
\end{proposition}
\begin{proof}See Appendix~\ref{prf:scope-unknownK}.\end{proof}

Appendix~\ref{prf:scope-unknownK} adds one annihilating-polynomial
comparison to the known-$K$ proof: every signal with fewer than $K$
atoms is at projected distance at least $c\eps^{2K-2}$ from the truth,
the exponent of the computational resolution limit for number
detection of deterministic line spectra in
\citet[Definitions~1--2 and Theorem~1]{liu2021computational}. For $\bar\zeta\le1$ the matching choice of
$\alpha$ requires $\bar\zeta\{\log(e/\bar\zeta)\}^{1/2}\lesssim\eps^{2K-2}$.
For $K=2$ this range retains the linear branch,
$\eps^3\lesssim\zeta_n\lesssim\eps^2/\{\log(e/\zeta_n)\}^{1/2}$, and for
$K=3$ it reaches $\zeta_{*,3}=\eps^4$ up to the same logarithmic margin.
For even $K\ge4$, $2K-2\ge3K/2$, so the range lies inside the low-noise
branch $\Psi_K=\zeta_n^2\eps^{-(3K-3)}$; for odd $K\ge5$ it is narrower
than $\zeta_n\le\zeta_{*,K}$. Thus this is not a full-noise unknown-$K$
result, and a constant-factor noise upper bound is needed.

\section{Discussion}
\label{sec:discussion}

This paper has determined the minimum mean-square error of locating
closely spaced sources from noisy superpositions of a known response,
jointly in separation and noise level. The error depends on the loss of
information near phase cancellation. It also depends on the prior mass
of those configurations. Under uniform phases this yields the parity
law. Weighted and singular priors change its exponents. The criterion
of Section~\ref{sec:criterion} isolates the two inputs of the proofs.
An inverse bound and a matching lower bound from prior-dominated paths
make the Bayes risk of the order of the capped Cram\'er--Rao bound. For
a two-sided power law of the information tail, averaging the inverse
information before capping keeps this order uniformly in the smallest
value of the information exactly when the lower-tail exponent exceeds
one. Section~\ref{factor:section} shows that
the criterion does not use the source geometry.

Moving forward, several extensions are worth pursuing. First, the
matching lower bound requires prior-mass control along the paths. The
results therefore do not assert a capped-information formula for
arbitrary priors. It would be of interest to characterize the priors
for which this formula holds. Second, the unknown-center extension
retains a single fine-scale true cluster and additional sampling
conditions. Unknown-count recovery also requires a stronger noise
threshold. Aliasing and nested clusters require different global
comparisons. Third, the least-squares estimators are global minimizers.
Algorithms with full complexity guarantees remain open. So do the
leading risk constants when separation and noise vanish jointly.
Finally, it would be of interest to verify the two conditions of the
criterion for other models whose information vanishes through
cancellation among the unknowns, such as finite mixtures near a
degenerate component.

\appendices
\numberwithin{equation}{section}
\numberwithin{theorem}{section}
\numberwithin{proposition}{section}
\numberwithin{lemma}{section}
\numberwithin{corollary}{section}
\numberwithin{remark}{section}
\numberwithin{example}{section}
\section{Proofs of Lemma~\ref{lem:weighted} and Proposition~\ref{thm:inverse}}
\label{sec:proofs}
\label{prf:weighted}
\label{prf:inverse}
Throughout the appendices, $\omega=\omega_x$ is the node polynomial of
Section~\ref{app:bridge}, and $\omega'$ is its derivative.
\begin{proof}[Proof of Lemma~\ref{lem:weighted}]
Fix $b\in\mathbb R^K$ and, for $\alpha\in\mathbb C^K$, let
$R_{\alpha,b}=\sum_k\alpha_kE_k+\eps\sum_kq_kb_kE_k'$, the tangent of
$\mu$ at $c=q$ in the direction $(\delta c,\delta s)=(\alpha,b)$, with scaled
tangent moments $w_l=\eps^l\delta m_l$, $l<2K$. By the differentiated
identity~\eqref{geom:prony-finite-tangent}, the first $K$ entries of
$T_\eps(p)w$ are $w_0,\ldots,w_{K-1}$, which range over $\mathbb C^K$ as
$\alpha$ varies, while its last $K$ entries equal
$\eps^{K+r}\sum_kq_k\omega'(s_k)s_k^rb_k$, $r<K$, independently of
$\alpha$. Lemma~\ref{lem:bridge} with each node doubled and
$T_\eps(p)=I+O(\eps)$ give
$\min_\alpha\|R_{\alpha,b}\|_\nu^2\asymp\eps^{2K}Q_{\eps,q}(b)$.
Since $J_j=\eps^{-2}\min_{\alpha,\,b_j=1}\|R_{\alpha,b}\|_\nu^2$
by~\eqref{eq:Jdef}, the claim follows; the constants are independent of
all phases.
\end{proof}

\emph{A common elimination identity.} For any full-column-rank real matrix $B$ and nonzero row $\ell$,
\begin{equation}
 \min_{\ell z=1}\|Bz\|^2
 =\frac{\displaystyle\sum_{|I|=d}\det(B_I)^2}
 {\displaystyle\sum_{|I|=d-1}\det\binom{B_I}{\ell}^{\,2}},
 \qquad d=\operatorname{rank}B.
 \label{eq:profile-minors}
\end{equation}
Indeed, the minimum is $\det(B^TB)/[\ell\operatorname{adj}(B^TB)\ell^T]$.
Cauchy--Binet gives the numerator; applying it after adjoining
$\sqrt t\ell$ and comparing coefficients of $t$ gives the denominator.
Thus a row weight $\eps^q$ contributes its squared weight to each minor.

\begin{proof}[Proof of Proposition~\ref{thm:inverse}]
Put $\delta=\|\Pi(\mu_x-\mu_{x'})\|_{\nu_n}$ and
$p(z)=\prod_k(z-s'_k)$. The finite-difference form
\eqref{geom:prony-difference} of~\eqref{geom:prony-finite-tangent} and
Lemma~\ref{lem:moment-comparison} give
$S_r:=\sum_kc_kp(s_k)s_k^r=\sum_{l=0}^Kp_l\{m_{l+r}(x)-m_{l+r}(x')\}$
and $\sum_{r<K}\eps^{2r}|S_r|^2\le C\eps^{-2K}\delta^2$.
For $p(s_j)\ne0$, the real direction
$b_k=p(s_k)d_j/[p(s_j)d_k]$, $d_k=\omega'(s_k)$, in the
amplitude-included form of Lemma~\ref{lem:weighted} gives
$J_j(x)\le C\eps^{2K-2}d_j^2p(s_j)^{-2}\sum_{r<K}\eps^{2r}|S_r|^2
\le C\delta^2/\{\eps^2p(s_j)^2\}$. Consequently, also when $p(s_j)=0$,
\begin{equation}
 |p(s_j)|\le\frac{C\delta}{\eps\sqrt{J_j(x)}}.
 \label{eq:prf-psj}
\end{equation}
If
$\delta\le c_*\eps\sqrt{J_{\min}(x)}$, then~\eqref{eq:prf-psj} gives
$|p(s_j)|\le Cc_*$ for every $j$; choose $c_*$ with $Cc_*<(\rho/3)^K$.
Since $|p(s_j)|=\prod_k|s_j-s'_k|$, each $s_j$ has a candidate root within
$\rho/3$. The intervals $(s_j-\rho/3,s_j+\rho/3)$ are disjoint, so each
contains exactly one of the $K$ candidate roots, counted with
multiplicity; the roots therefore match in order, and
$|s_j-s'_k|\ge2\rho/3$ for $k\ne j$. Finally,
$|s_j-s'_j|=|p(s_j)|/\prod_{k\ne j}|s_j-s'_k|
\le(3/(2\rho))^{K-1}|p(s_j)|\le C_*\delta/(\eps\sqrt{J_j(x)})$.
\end{proof}

\section{Paths and Phase Integrals for Uniform Phases}
\label{geom:source-lower}
\emph{Prior-dominated paths for uniform phases.}
Fix a node--modulus point in the interior of the ball $U$ of
Theorem~\ref{thm:randomphase}, on which $f\ge c_f$ almost everywhere.
All charts below are chosen about this point, with phases selected
freely by the uniform law, and may be restricted inside any prescribed
Prony neighborhood of it whose node--modulus projection lies in $U$;
their sizes and Jacobian bounds are independent of $\eps,n,\zeta_n$.

For $K=2r+1$ (Lemma~\ref{geom:fiber-paths}(i)), put phase zero on
$A\ni j$ of size $r+1$ and $\pi/2$ on its complement, and set
$M=(m_0,\ldots,m_{K+r-1})$. The recurrence map
$e\mapsto(\sum_ac_as_a^qe(s_a))_{q<r}$ on real polynomials has rank
$2r=K-1$ at the reference point, as its two weighted Vandermonde blocks
have rank $r$ each, so Corollary~\ref{geom:affine-moment-fiber} gives
coordinates $(M,u)$ with $p=p_*(M)+u e(M)$. At the reference point
$e(s_a)=0$ off $A$, while $e(s_j)\ne0$ because every $r$-column minor on
$A$ is invertible; so $\partial_us_j=-e(s_j)/\omega'(s_j)$ has fixed
sign and magnitude bounded below nearby. As $\partial_um_l=0$ for
$l<K+r$, Lemma~\ref{lem:bridge} applied to the tangent gives
\begin{equation}
 |\partial_u\tau_j|\ge c\eps,\qquad
 \|\partial_u\mu\|_{\nu_n}\le C\eps^{K+r}.
 \label{geom:source-direction}
\end{equation}
Traverse a fixed $u$-interval at constant speed and average over $M$;
the chart Jacobian and $f\ge c_f$ bound the occupation measure by
$C\pi$, which with~\eqref{geom:source-direction} proves
Lemma~\ref{geom:fiber-paths}(i).

For $K=2r$, choose the reference configuration of
Section~\ref{sec:power-weights} at the same node--modulus point, with
$j\in A$, and at $\beta=0$ shrink the weighted chart inside $U$. Its
prior with uniform node--modulus density is then dominated by the
baseline prior, as $f\ge c_f$ there, so the paths of Step~4 of
Appendix~\ref{sec:weighted-phase-details} have the required occupation
bound under $\pi$; this is the $\beta=0$ domination used in
Lemma~\ref{geom:fiber-paths}(ii) for uniform phases.

\begin{proof}[Proof of Lemma~\ref{geom:ellipse}(1)]
We use $d_a,h_a,\mathcal V_j$ of Section~\ref{sec:info-def}; all
comparisons are uniform over~\eqref{eq:class}, admissible designs and
nuisance phases. Let $\mathsf T_j$ have columns $(h_j,0)$ and $(0,h_j)$
in $\mathbb R^{2K}$, so
$\mathbf R=\mathsf T_j^T(I-P_{\mathcal V_j})\mathsf T_j$.
Collect the nuisance $h_a$ as $H$, and set $\alpha_j=\|h_j\|^2$,
$G=H^TH$, $G_0=H^T(I-h_jh_j^T/\alpha_j)H$ and
$\mathcal Q=\operatorname{diag}(e^{i\varphi_a})_{a\ne j}$. The real
nuisance Gram matrices are $N=\Ree(\mathcal Q^*G\mathcal Q)$ and, after
projection off the target line, $N_0=\Ree(\mathcal Q^*G_0\mathcal Q)$;
both are positive definite by distinct-node Vandermonde rank, and the
Schur complement in either order gives
$\det\mathbf R=\alpha_j^2\det N_0/\det N$.

Let $\check H$ collect the shortened columns
$\check h_a=(1,\eps s_a,\ldots,\eps^{K-2}s_a^{K-2})^T$, $a\ne j$, and
$\Delta=\operatorname{diag}(s_a-s_j)_{a\ne j}$. Interpolation of
$t^{K-1}$ on the nuisance nodes gives
$G=\check H^T(I+\mathbf e\mathbf e^T)\check H$, with
$\mathbf e=O(\eps)$ since the interpolation coefficients carry
positive powers of $\eps$. Also
$(h_a-h_j)/\{\eps(s_a-s_j)\}=(0,\Lambda\check h_a)$ with
$\Lambda_{uv}=(\eps s_j)^{u-v}$ ($u\ge v$), so $\Lambda=I+O(\eps)$, and
projection off $h_j$ has Gram $B=I+O(\eps^2)$ on vectors with zero
first coordinate. Thus
$G_0=\eps^2\Delta\check H^T\Lambda^TB\Lambda\check H\Delta$.
Both coefficient matrices are $I+O(\eps)$. Compare them before the
possibly ill-conditioned congruence by $\check H\Delta$:
\[
 (1-C\eps)\eps^2\Delta G\Delta\preceq G_0
 \preceq(1+C\eps)\eps^2\Delta G\Delta.
\]
Phase conjugation and real parts preserve this order. Since
$\mathcal Q\Delta=\Delta\mathcal Q$, the order gives
$\det N_0/\det N=\eps^{2K-2}d_j^2(1+O(\eps))$, and, as
$\alpha_j=1+O(\eps^2)$, the Schur identity gives
$\sqrt{\det\mathbf R}=\alpha_j\sqrt{\det N_0/\det N}=\eps^{K-1}|d_j|(1+O(\eps))$,
uniformly in the nuisance phases.
\end{proof}

\begin{proof}[Proof of Lemma~\ref{geom:ellipse}(2)]\label{supp:rp-even}
Fix $K=2r$, a target $j$, and a compact class $\mathcal S$ of ordered
node vectors satisfying $|s_a|\le T$ and $s_{a+1}-s_a\ge\rho>0$.
Identify coefficient-unit polynomials of degree below $r$ with
$S^{2r-1}\subset\mathbb C^r$. Almost every
$\varphi_{-j}\in(\mathbb R/\pi\mathbb Z)^{2r-1}$ determines exactly two
solutions $\pm\mathsf p_\varphi$ of
$\Ree(e^{i\varphi_a}\mathsf p(s_a))=0$, $a\ne j$.
To see this, set $r$ nuisance phases to zero and $r-1$ to $\pi/2$; then
the real polynomial vanishes and the imaginary one is proportional to
the product over the latter nodes, so a maximal minor is a nonzero
analytic function; the cofactor nullvector is analytic, each nuisance evaluation
is not identically zero (place that node in the first group), and
neither is the target evaluation. Rank loss and evaluation zeros
therefore form null sets.

The main step is the bound
$\mathbb E_{\varphi_{-j}}|\mathsf p_\varphi(s_j)|^{-1}\le C$,
uniformly over $\mathcal S$. It is immediate for $r=1$. Otherwise,
off the nuisance evaluation-zero sets define
$\Phi_a(\mathsf p)=\pi/2-\arg\mathsf p(s_a)\pmod\pi$, $a\ne j$.
Its generic fibers have two points, and differentiating the arguments
and Hadamard's inequality give
$|\operatorname{Jac}\Phi(\mathsf p)|\le C\prod_{a\ne j}|\mathsf p(s_a)|^{-1}$.
The area formula yields
\[
 \mathbb E_{\varphi_{-j}}|\mathsf p_\varphi(s_j)|^{-1}
 =\frac{1}{2\pi^{2r-1}}
   \int_{S^{2r-1}}\frac{|\operatorname{Jac}\Phi(\mathsf p)|}{|\mathsf p(s_j)|}\,dS(\mathsf p)
 \le C\int_{S^{2r-1}}\prod_{a=1}^{2r}|\mathsf p(s_a)|^{-1}\,dS(\mathsf p).
\]
Apply the formula away from the evaluation-zero sets, then pass to
the full nonnegative integral by monotone convergence.

Uniform Vandermonde invertibility prevents $r$ evaluations from being
simultaneously small, so
$\prod_{a=1}^{2r}|\mathsf p(s_a)|^{-1}
\le C\sum_{|A|=r-1}\prod_{a\in A}|\mathsf p(s_a)|^{-1}$.
In coefficient space write $\mathsf G=\chi U\sim\CN(0,I_r)$, with
independent uniform direction $U$ and radius $\chi$.
For each $A$, the covariance of $(\mathsf G(s_a))_{a\in A}$ has
uniform positive lower and upper bounds, so its density is at most
$Ce^{-c|z|^2}$ and
$\mathbb E\chi^{-(r-1)}\,\mathbb E\prod_{a\in A}|U(s_a)|^{-1}
=\mathbb E\prod_{a\in A}|\mathsf G(s_a)|^{-1}
\le C\prod_{a\in A}\int_{\mathbb C}e^{-c|z|^2}|z|^{-1}\,d^2z<\infty$.
Since $\mathbb E\chi^{-(r-1)}=\Gamma((r+1)/2)/\Gamma(r)>0$,
summing proves the main step uniformly over $\mathcal S$.

Write $\mathsf p_\varphi(y)=\sum_{\iota<r}\beta_\iota y^\iota$ and let $v$ have entries
$v_\iota=\overline{\beta_\iota}\,\eps^{-\iota}$ for $\iota<r$ and $0$ otherwise. For the real
inner product $\langle v,u\rangle=\Ree\sum_\iota\overline{v_\iota}u_\iota$, the vector $v$ satisfies
$\langle v,e^{i\varphi_a}h_a\rangle=\Ree\bigl(e^{i\varphi_a}\mathsf p_\varphi(s_a)\bigr)=0$
for every nuisance column, and $\|v\|^2\le\eps^{-2r+2}$. Its maximal pairing with the target's two columns
is $|\mathsf p_\varphi(s_j)|$. Thus
$\lambda_{\max}(\mathbf R)\ge\eps^{2r-2}|\mathsf p_\varphi(s_j)|^2$ and
$\mathbb E\lambda_{\max}(\mathbf R)^{-1/2}\le C\eps^{-r+1}$.
\end{proof}

\section{Proof of Theorem~\ref{thm:prior-effect} for Power-Weighted Priors}
\label{sec:weighted-phase-details}
\label{sec:local-weighted-phase}
Lemma~\ref{weighted:moment-chart} and Step~4 below prove
Lemma~\ref{geom:fiber-paths}(ii) for every $\beta>-1$. Write $K=2r$,
retain the notation of Section~\ref{sec:power-weights}, and write
$(s^*,a^*,\varphi^*)$ for the reference configuration. All rectangles
are fixed independently of $\eps,n,\sigma$; comparisons are uniform on
them and over the designs of Assumption~\ref{ass:sampling}, with
constants depending on the rectangles, $K$, $\beta$, density and model
bounds.
\begin{lemma}[A chart preserving the weighted transverse coordinate]
\label{weighted:moment-chart}
Near the reference configuration there are smooth coordinates
$(\mathfrak b,z,u,v)$, with all neighborhoods independent of $\eps$, such that
$(\mathfrak b,z)$ parametrizes $M=(m_0,\ldots,m_{3r-2})$ and $u,v$ are affine
coordinates on its two-dimensional fiber. Translation by $h$ in $u$
fixes these moments and satisfies $\Delta m_{3r-1}=i\mathfrak bh$ exactly.
Moreover $\partial_u s_j\ne0$ for every $j\in A$ at the reference
point, and $\mathscr D=\mathsf k(x)\mathfrak b$ for a smooth nonvanishing factor $\mathsf k$.
\end{lemma}
\begin{proof}
\emph{The affine fixed-moment fiber.}
Use the moments $m_l$, the node polynomial $\omega$ and its
coefficients $p$ of Lemma~\ref{geom:prony-secant}, with $M$ as in the
statement, and $d_a=\prod_{\ell\ne a}(s_a-s_\ell)$. The reference phases are zero on
$A$ ($r+1$ nodes) and $\pi/2$ on $A^{c}$ ($r-1$ nodes). For a real
polynomial variation $\delta\omega$, the recurrence matrix satisfies
$(\mathcal A(M)\delta p)_q=\sum_ac_as_a^q\delta\omega(s_a)
=-\sum_ac_ad_as_a^q\delta s_a$, $q<r-1$.
At the reference point the real and imaginary Vandermonde blocks
act on disjoint node sets and each have rank $r-1$, giving real rank $2r-2$.
Corollary~\ref{geom:affine-moment-fiber} gives full coordinates $(M,t)$
with $p=p_*(M)+N(M)t$, $t\in\mathbb R^2$, and
$\operatorname{range}N(M)=\ker\mathcal A(M)$, including $r=1$.

\emph{A transverse coordinate on the base.}
The next moment
$\chi(M,p):=m_{3r-1}=-\sum_{l<K}m_{r-1+l}p_l$ is exactly linear in $p$
on a fixed-$M$ fiber. At the reference point, the imaginary constraints force the node
variations on $A^{c}$ to vanish. The fiber derivative of $\chi$ is then
real and of rank one. Its kernel on $A$ is the nullspace of an
$r\times(r+1)$ weighted Vandermonde matrix; every component is nonzero
because every $r$-column minor is invertible.

Choose a smooth fiber basis $e_1(M),e_2(M)$ with
$\Ree(\chi_p e_1)=0$, $\Ree(\chi_p e_2)=1$ and $\chi_p e_1=i\mathfrak b(M)$;
for example, if $\eta=\Ree(\chi_p N)=(\eta_1,\eta_2)\ne0$, take
$e_1=N(-\eta_2,\eta_1)^T$ and $e_2=N\eta^T/\|\eta\|^2$.
Writing $p=p_*(M)+u e_1(M)+v e_2(M)$ gives
$\partial_u s_j\ne0$ at the reference point and $\mathfrak b(M_*)=0$.

For $F(x)=(m_0,\ldots,m_{3r-1})$, differentiate
\eqref{geom:prony-difference} and eliminate the first $K$ moments.
Since $\delta\omega(s_a)=-d_a\delta s_a$, this gives
$\det DF=C(x)\det H$, with $C$ smooth and nonzero and with column $a$ of
$H$ equal to $a_ad_a$ times the stacked vectors
$(s_a^l\cos\varphi_a)_{l<r}$ and $(s_a^l\sin\varphi_a)_{l<r}$.
Here $H$ has rank $2r-1$ at the reference point, and, for $a\in A$,
$\partial_{\varphi_a}\det H|_*=\pm(\prod_{k=1}^{K}a_kd_k)
\det V_{A\setminus\{a\}}\det V_{A^{c}\cup\{a\}}\ne0$,
where the two Vandermonde matrices have degrees $0,\ldots,r-1$.
In $(M,u,v)$, the Jacobian of $F$ is block triangular with fiber block
$\bigl(\begin{smallmatrix}0&1\\ \mathfrak b(M)&\Imm(\chi_p e_2)\end{smallmatrix}\bigr)$.
Consequently $\det DF=-C_1(x)\mathfrak b(M(x))$, with $C_1\ne0$, and the
nonzero phase derivative implies $\dd\mathfrak b(M_*)\ne0$. Write
$M=(\mathfrak b,z)$ locally. Translation in $u$ fixes this base and changes
$\chi$ by $i\mathfrak bh$ exactly. Finally,
$\det H=\pm(\prod_a a_ad_a)\mathscr D$ gives
$\mathscr D=\mathsf k(x)\mathfrak b$, with $\mathsf k$ smooth and nonzero.
\end{proof}

\begin{proof}[Proof of Theorem~\ref{thm:prior-effect}(a)]
\emph{1. The prior $\pi_\beta$ is well defined uniformly in the nodes.}
Choose $a_0\in A$. As $\partial_{\varphi_{a_0}}\mathscr D(s^*,\varphi^*)\ne0$
by the preceding proof, the implicit function theorem gives fixed rectangles
$\mathcal S_\beta$, $\Phi_{-a_0}$, and an interval $\Phi_{a_0}=[-\delta,\delta]$
such that the zero is the graph
$\varphi_{a_0}=\xi(s,\varphi_{-a_0})\in[-\delta/2,\delta/2]$,
and $|\partial_{\varphi_{a_0}}\mathscr D|$ is bounded above and
below by positive constants throughout the rectangle.
Shrink the other phase intervals first, then $\mathcal S_\beta$ and $\mathcal G$,
to keep the closure inside the moment chart. Hence
$|\mathscr D|\asymp|\varphi_{a_0}-\xi|$ and
$0<c_\beta\le Z_\beta(s)\le C_\beta<\infty$ ($\beta>-1$),
uniformly in $s$, with an interior crossing margin in every node section.

\emph{2. A two-sided information estimate on the entire rectangle.}
With $c_a=a_ae^{i\varphi_a}$, let $\mathsf M_m(x)$, $m\ge0$, be the real
$2\times K$ matrix with column $a$ equal to the real and imaginary parts
of $c_ad_as_a^m$, put
$G_\eps(x)=\sum_{m<K}\eps^{2m}\mathsf M_m(x)^T\mathsf M_m(x)$, and stack
vertically $\mathsf L=(\mathsf M_0;\ldots;\mathsf M_{r-2})$ (no rows if $r=1$),
$H=(\mathsf L;\mathsf M_{r-1})$ and $H_+=(H;\mathsf M_r)$; up to the
order of its rows, this $H$ is the matrix of the preceding proof. Lemma~\ref{lem:weighted}, with amplitude
factors absorbed into real node-variation coordinates, says
\begin{equation}
 J_a(x)\asymp\eps^{2K-2}
      \min_{\theta_a=1}\theta^TG_\eps(x)\theta.
 \label{eq:local-prior-weighted-information}
\end{equation}
Indeed, replacing the $k$th coordinate of the variable of
$Q_{\eps,q}$ in Lemma~\ref{lem:weighted} by $(a_k/a_a)\theta_k$
converts $a_a^2Q_{\eps,q}$ to the form displayed here.

At the reference point $\mathsf L,H,H_+$ have ranks $K-2,K-1,K$.
The nullvector of $H$ is nonzero in every coordinate of $A$.
Apply \eqref{eq:profile-minors} to
$B_\eps=(\mathsf M_0^T,\eps\mathsf M_1^T,\ldots,\eps^{K-1}\mathsf M_{K-1}^T)^T$.
The least-weight full minor is $\eps^{r(r-1)}\det H$, with
$\det H=\pm(\prod_a a_ad_a)\mathscr D$.
A minor using all rows of $\mathsf L$, one of $\mathsf M_{r-1}$ and one of $\mathsf M_r$
is nonzero at the reference point and has one additional power of
$\eps$; all other minors have at least that weight. Hence
$\det G_\eps\asymp\eps^{2r(r-1)}(\mathscr D^2+\eps^2)$. The denominator in \eqref{eq:profile-minors} for $\ell=e_j^T$
is at most $C\eps^{2(r-1)^2}$, the least weight of $K-1$ rows.
For $j\in A$, deleting column $j$ preserves the ranks of both
$\mathsf L$ and $H$: otherwise a row combination supported at $j$ would
annihilate the nullvector of $H$. Therefore a least-weight denominator
minor is nonzero and its coefficient stays bounded below nearby.
Using \eqref{eq:local-prior-weighted-information} gives, uniformly
on the rectangle,
\begin{equation}
 J_{\min}(x)\asymp J_j(x)
 \asymp\eps^{6r-4}(\mathscr D(s,\varphi)^2+\eps^2),
 \qquad j\in A.
 \label{eq:local-even-information-normal-form}
\end{equation}

\emph{3. One estimator supplies the full upper bound.}
Put $\lambda=\zeta_n/\eps^{3r-1}$, so that
$\min\{\eps^2,\zeta_n^2/(\eps^{6r-4}(t^2+\eps^2))\}
=\eps^2\min\{1,\lambda^2/(t^2+\eps^2)\}$.
Lemma~\ref{cor:conditional-mse}, with minimization over the full
class $\mathcal X_K$, \eqref{eq:local-even-information-normal-form},
the change of variable $\varphi_{a_0}\mapsto\mathscr D$ from Step~1 and
$Z_\beta(s)\asymp1$ give, uniformly in $(s,a)$,
$\mathbb E|\widehat\tau_j-\tau_j|^2\lesssim
\eps^2\int_0^{\bar t}\min\{1,\lambda^2/(t^2+\eps^2)\}t^\beta\,dt$.
The remaining average uses only that $f$ is a probability density.

\emph{4. Paths under the same prior $\pi_\beta$
(Lemma~\ref{geom:fiber-paths}(ii)).}
In the chart of Lemma~\ref{weighted:moment-chart},
$\mathscr D=\mathsf k(x)\mathfrak b$. The polar and chart Jacobians are bounded
above and below, so on a fixed inner product box the prior density is
\begin{equation}
 q_\beta(\mathfrak b,z,u,v)\asymp |\mathfrak b|^\beta.
 \label{eq:weighted-prior-chart}
\end{equation}
Fix $(\mathfrak b,z,v)$ and traverse the box's $u$-interval, of fixed
length, at constant speed. Shrinking the box if necessary, for every
$j\in A$, $\partial_u\tau_j$ has fixed sign and magnitude at least
$c\eps$ throughout the box. The first $3r-1$ moments stay
fixed and $\partial_um_{3r-1}=i\mathfrak b$; the remaining moment derivatives
through order $2K-1$ are uniformly bounded. Lemma~\ref{lem:bridge}
applied to the tangent gives
$\|\partial_u\mu\|_{\nu_n}\le C\eps^{3r-1}\sqrt{\mathfrak b^2+\eps^2}$.
Give the path labels $(\mathfrak b,z,v)$ the measure
$\Upsilon(d\mathfrak b\,dz\,dv)=|\mathfrak b|^\beta\,d\mathfrak b\,dz\,dv$.
Fubini and \eqref{eq:weighted-prior-chart} bound their total occupation
by $C\pi_\beta$, also for $-1<\beta<0$; no derivative of this weight
is taken. Applying \eqref{spine:two-integrals}, with fixed interval lengths
absorbed into constants, proves
$R_{\pi_\beta,j}\ge c\eps^2\int_0^{\mathfrak b_0}
\min\{1,\lambda^2/(\mathfrak b^2+\eps^2)\}\mathfrak b^\beta\,d\mathfrak b$.

Both bounds are thus $\eps^2 I_{\beta+1}(\eps,\lambda)$ up to the fixed
endpoints $\bar t$ and $\mathfrak b_0$, which change only constants
(for $\mathfrak b_0$ by the substitution $\mathfrak b=\mathfrak b_0t$ of
Section~\ref{sec:randomphase-proof}, valid for every $\beta>-1$). Lemma~\ref{risk:power-integral} proves
\eqref{eq:local-weighted-phase-risk}--\eqref{eq:weighted-phase-regimes}.
\end{proof}

\begin{corollary}[Local average inverse information]
\label{cor:local-weighted-information}
Under the prior of Theorem~\ref{thm:prior-effect}(a), for every $j\in A$,
\[
 \mathbb E_{\pi_\beta}[J_j^{-1}]
 \asymp
 \begin{cases}
  \eps^{\beta-(6r-3)},&-1<\beta<1,\\
  \eps^{-(6r-4)}[1+\log(1/\eps)],&\beta=1,\\
  \eps^{-(6r-4)},&\beta>1.
 \end{cases}
\]
The same orders hold when nodes and amplitude moduli are fixed in
their prescribed rectangles and only the conditional phases are
averaged.
\end{corollary}

\begin{proof}
By~\eqref{eq:local-even-information-normal-form}, the change of variable
$\varphi_{a_0}\mapsto\mathscr D$ of Step~1 and $Z_\beta(s)\asymp1$, uniformly
in $(s,a)$,
$\mathbb E[J_j^{-1}\mid s,a]\asymp\eps^{-(6r-4)}\int_0^{\bar t}t^\beta(t^2+\eps^2)^{-1}dt$.
The substitution $t=\eps u$ shows that the integral is of order
$\eps^{\beta-1}$ for $\beta<1$ and $1+\log(1/\eps)$ for $\beta=1$; for
$\beta>1$ the integrand is at most $t^{\beta-2}$, which is integrable,
while the part over $[\bar t/2,\bar t]$ is bounded below for $\eps\le1$,
so the integral is of order $1$. This is the fixed-node version;
integrating in $(s,a)$ gives the unconditional orders.
\end{proof}

\begin{corollary}[Adding a uniform-phase component]
\label{cor:full-support-weighted-phase}
Let $\pi_0$ be a baseline prior of Theorem~\ref{thm:randomphase}
on the same search class, and let $\pi_\beta$ be the preceding local
prior. For fixed $0<p_{\mathrm{mix}}<1$ set
$\pi_{\mathrm{mix}}=(1-p_{\mathrm{mix}})\pi_0+p_{\mathrm{mix}}\pi_\beta$.
Write $\Psi_0$ for the even-source risk order in~\eqref{eq:random-phase-risk} and $\Psi_\beta$
for~\eqref{eq:local-weighted-phase-risk}. Equivalently, $\Psi_0$ is
the integral in~\eqref{eq:local-weighted-phase-risk} with $\beta=0$
and the same $\lambda=\zeta_n/\eps^{3r-1}$. Then, for $j\in A$,
\[
 R_{\pi_{\mathrm{mix}},j}\asymp
       (1-p_{\mathrm{mix}})\Psi_0+p_{\mathrm{mix}}\Psi_\beta
       \asymp\max\{\Psi_0,\Psi_\beta\}.
\]
If the node--modulus density of $\pi_0$ is positive throughout the
admissible interior, this mixture has full phase support at every
interior node--modulus configuration in either component's support.
For $-1<\beta<0$ its risk has the heavier local order $\Psi_\beta$;
for $\beta\ge0$ it retains the baseline order $\Psi_0$.
\end{corollary}

\begin{proof}
Every estimator's mixture risk is the weighted sum of its component
risks, while the same global projected least-squares estimator attains
both component upper bounds. Hence the Bayes risk is comparable to
$(1-p_{\mathrm{mix}})\Psi_0+p_{\mathrm{mix}}\Psi_\beta$, and thus to
$\max\{\Psi_0,\Psi_\beta\}$ because $p_{\mathrm{mix}}\in(0,1)$ is fixed.
The claimed dominance follows by comparing $t^\beta$ with $1$ in the
common scalar integral, after fixed endpoint rescaling.
\end{proof}

\section{Proof of Theorem~\ref{thm:prior-effect} for the Singular Three-Source Prior}
\label{prf:k3-singular}
This appendix proves part~(b) of Theorem~\ref{thm:prior-effect}.

Write $\zeta=\zeta_n$, $\lambda=\zeta/\eps^4$ and
\[
 \omega(z)=\prod_{j=1}^3(z-s_j)
          =z^3+p_2z^2+p_1z+p_0.
\]
All rectangles are fixed independently of $\eps,\zeta,n$ and retain
separated nodes, positive real weights and Assumption~\ref{ass:sampling}.
On the support near alignment of \eqref{eq:k3-singular-prior}, take $v_0<1$
and $0<\eps\le\eps_0$, where $v_0^2+\eps_0^2\le1$; reduce $\eps_0$
finitely many times below. Constants may depend on these rectangles
and the fixed $\gamma$.

\emph{An information form uniform over the prior support.} Put $\vartheta_G=v^2+\eps^2\le1$ and $\mathfrak S=u^2+\vartheta_G^2$.
The amplitude-included form of Lemma~\ref{lem:weighted}, with
$\xi_q=\sum_l\mathsf w_l\omega'(s_l)b_ls_l^q$, gives
\[
 J_j\asymp\eps^4\min_{\ell_j\xi=1}\|B_\eps\xi\|^2,\qquad
 B_\eps=\begin{pmatrix}e_0^T\\\eps e_1^T\\\eps^2e_2^T\\
 b_0\\\eps b_1\\\eps^2b_2\end{pmatrix},\quad
 \ell_j=\frac{(s_ks_l,-s_k-s_l,1)}{\mathsf w_j\omega'(s_j)^2},
\]
where $\{j,k,l\}=\{1,2,3\}$,
$b_0=(0,v,u)$, $b_1=(-p_0u,-p_1u,v-p_2u)$ and
$b_2\xi=v\xi_3+u\xi_4$, with the recurrence
$\xi_{q+3}=-p_2\xi_{q+2}-p_1\xi_{q+1}-p_0\xi_q$.
This only subtracts $t_0$ times each real row from its imaginary row;
$b_2$ is linear in $(u,v)$ with uniformly bounded coefficients.

In the numerator of \eqref{eq:profile-minors}, the only total-weight-one
minors are, up to sign,
$\eps u$ and $\eps P$, where $P=v^2-p_2uv+p_1u^2$.
The three real rows give $\eps^3$.
Every weight-two coefficient is $O(|u|+|v|)$, and higher-weight
coefficients are bounded. Since $u^2+P^2\asymp u^2+v^4$ and
$\eps^2v^2\le(v^4+\eps^4)/2$, the numerator is comparable to
$\eps^2\mathfrak S$.
The denominator's only weight-zero minor is
$\ell_{j,2}v-\ell_{j,1}u$; another is $\eps\ell_{j,2}$,
and all others have positive weight.
Because $|\ell_{j,2}|\asymp1$ and
$\ell_{j,1}/\ell_{j,2}=p_2+s_j$, this proves
\begin{equation}
 \begin{split}
 J_j&\asymp\eps^6
       \frac{\mathfrak S}{\eps^2+[v-(p_2+s_j)u]^2},\qquad j\le3,\\
 J_{\min}^{-1}&\asymp\eps^{-6}
                     \left(1+\frac{\vartheta_G}{\mathfrak S}\right).
 \end{split}
 \label{k3p:information}
\end{equation}
For the second line, the separated $p_2+s_j$ give
$\max_j|v-(p_2+s_j)u|^2\asymp v^2+u^2$, and
$(\vartheta_G+u^2)/\mathfrak S\asymp1+\vartheta_G/\mathfrak S$ because $\vartheta_G\le1$.
Here $J_j$ includes the squared amplitude modulus.

\emph{Risk upper bound.} Equation~\eqref{eq:rp-conditional-mse} and
\eqref{k3p:information}, integrated under the stated prior, give
\begin{equation}
 \frac{\mathbb E|\widehat\tau_j-\tau_j|^2}{\eps^2}
 \le C\left[
  \min(1,\lambda^2)+
  \int_0^{v_0}\!\int_0^{u_0}
   \min\left\{1,\frac{\lambda^2\vartheta_G}{u^2+\vartheta_G^2}\right\}
                       u^{-\gamma}\,du\,dv\right].
 \label{k3p:upper-integral}
\end{equation}
Put $\varrho_\eps^2=v^2+\eps^2+u^2/(v^2+\eps^2)$ and
$\bar\varrho_\eps=\min\{1,\varrho_\eps\}$. The baseline term in~\eqref{k3p:upper-integral} makes its right side
comparable to $\mathbb E\min\{1,\lambda^2/\bar\varrho_\eps^2\}$.
Choose fixed $s_0<\min(1,v_0)$ with $s_0^2/2<u_0$, then
$\eps_0<s_0/4$. Integrating $u^{-\gamma}$ over $\{\varrho_\eps\le s\}$ gives
\[
 \mathbb P(\bar\varrho_\eps\le s)\asymp
 \int_0^{\sqrt{s^2-\eps^2}}
 [(v^2+\eps^2)(s^2-v^2-\eps^2)]^{(1-\gamma)/2}\,dv
 \asymp s^{3-2\gamma},\qquad 2\eps\le s\le s_0.
\]
The last comparison follows by $v=sz$, uniformly for $\eps/s\le1/2$.
The upper bound holds for all $s\le s_0$, and $\bar\varrho_\eps\ge\eps$.
Equation~\eqref{risk:capped-tail} and Lemma~\ref{risk:power-integral},
with $q=3-2\gamma$, therefore give
\begin{equation}
 \begin{split}
 \mathbb E\min\{1,\lambda^2/\bar\varrho_\eps^2\}
 &=2\lambda^2\int_\lambda^\infty
           \mathbb P(\bar\varrho_\eps\le s)s^{-3}\,ds\\
 &\asymp\lambda^2\left[1+
       \int_{\eps\vee\lambda}^{s_0}s^{-2\gamma}\,ds\right]
 \asymp\Phi_\gamma(\eps,\lambda),
 \end{split}
 \label{k3p:tail-profile}
\end{equation}
where the second line is for $\eps\vee\lambda\le s_0/2$.
For larger $\lambda$, bounded target diameter gives the same upper
order, since $\Phi_\gamma$ is bounded below by a positive constant.
This argument changes neither the original support nor its prior.

\emph{Paths under the same singular prior.}\label{geom:k3-paths}
We prove Lemma~\ref{geom:fiber-paths}(iii) in the real Prony
coordinates $(M,p)$, $M=(m_0,m_1,m_2)$, of
Lemma~\ref{geom:prony-secant}. With
$\mathfrak r=(\mathfrak r_0,\ldots,\mathfrak r_5)$, the additional
Jacobians are
$|\det D_{(\mathsf w,s)}\mathfrak r|
=(\prod_j\mathsf w_j)\prod_{j<k}(s_j-s_k)^4>0$ and
$|\det D_{(s,\mathsf w,t_0,v,u)}(s,\Ree c,\Imm c)|
=(\prod_j\mathsf w_j)|\det V(s)|>0$, where $V(s)$ has rows
$(1,s_j,s_j^2)$ and the last diagonal Jacobian block is
$\operatorname{diag}(\mathsf w)V(s)$. Separation and positive weights
bound these Jacobians and their inverses, so the $(M,p)$ prior
density is comparable to $u(M,p)^{-\gamma}$.

Choose fixed $t_{\max}>0$ with $1.2t_{\max}<v_0$ and
$3t_{\max}^2<u_0$. Choose compact rectangles with nonempty interiors,
$\mathcal R^-\subset\operatorname{int}\mathcal R^+
\subset\mathcal R^+\subset\operatorname{int}\mathcal R$
and $T^-\subset\operatorname{int}T^+
\subset T^+\subset\operatorname{int}T_0$.
For each dyadic layer $t=t_l=t_{\max}2^{-l}\ge\eps$ of
Lemma~\ref{geom:fiber-paths}(iii), written $t$ below, take starting points with
$(\mathfrak r,t_0)\in\mathcal R^-\times T^-$, $t\le v\le1.1t$ and
$t^2\le u\le2t^2$. For $h\ge0$ the shear
$S_{t,h}(M,p)=(M,p+h\chi_t(M))$, with
$\chi_t(M)=(\Ree M\mathbin{\times}\Imm M)/t$, fixes $M$ and has
Jacobian one for each fixed $h$.
Write $\chi_t(z)=\sum_{l=0}^2\chi_{t,l}z^l$ and set
$R_0=(\mathfrak r_0,\mathfrak r_1,\mathfrak r_2)^T$, $R_1=(\mathfrak r_1,\mathfrak r_2,\mathfrak r_3)^T$,
$R_2=(\mathfrak r_2,\mathfrak r_3,\mathfrak r_4)^T$, $\mathcal H=(R_0,R_1,R_2)$.
The real Hankel matrix $\mathcal H$ is uniformly positive definite,
$\Imm M=t_0R_0+vR_1+uR_2$ and
$\chi_t=(v/t)(R_0\mathbin{\times}R_1)+(u/t)(R_0\mathbin{\times}R_2)$.
For the polynomial $\chi$ with coefficient vector $R_0\times R_1$,
Vandermonde inversion gives
$\mathsf w_j\chi(s_j)=\det\mathcal H/\omega'(s_j)$, $j\le3$.
Thus $\chi_t$ is uniformly bounded and $|\chi_t(s_j)|\ge c>0$
after reducing $t_{\max}$.

Dots now denote derivatives along $S_{t,h}$. Orthogonality to $M$
keeps $m_3$ and $\mathfrak r_0,\ldots,\mathfrak r_3$ fixed, so
$\dot{\mathcal H}=\dot{\mathfrak r}_4e_3e_3^T$. Differentiating the recurrence
and $\Imm M=\mathcal H(t_0,v,u)^T$ gives
$\dot{\mathfrak r}_4=(u/t)\det\mathcal H$,
$(\dot t_0,\dot v,\dot u)^T=-u\dot{\mathfrak r}_4\mathcal H^{-1}e_3$ and
$\dot m_4=-\chi_t\cdot(m_1,m_2,m_3)$. On the enlarged layer
$\mathcal L_t^+=\mathcal R^+\times T^+\times[0.9t,1.2t]\times[0.8t^2,3t^2]$,
these identities and the next moment recurrence give
$|\dot u/u|\le Ct$, $|\dot v/v|\le Ct^2$, $|\dot t_0|\le Ct^3$,
$|\dot{\mathfrak r}_4|+|\dot m_4|\le Ct$ and
$|\dot{\mathfrak r}_5|+|\dot m_5|\le C$.
Choose fixed $h_0>0$ from the rectangle margins, with $Ch_0<1/25$.
Up to the first exit from $\mathcal L_t^+$, integration bounds
$u(h)/u(0)$ and $v(h)/v(0)$ between $e^{-1/25}$ and $e^{1/25}$;
the other moving coordinates cannot reach their larger-rectangle
boundaries. These strict bounds exclude a first exit before $h_0$.
The inverse moment chart therefore keeps every path in the original
admissible support for $0\le h\le h_0$, and
\begin{equation}
 \frac{u(S_{t,h}x)}{u(x)}=1+O(th),\qquad
 \frac{v(S_{t,h}x)}{v(x)}=1+O(t^2h).
 \label{k3p:prior-preservation}
\end{equation}
The original prior density in Prony coordinates consequently satisfies
$q_\gamma(x)\asymp q_\gamma(S_{t,h}x)$, uniformly as $u(x)\downarrow0$.
Moreover $\dot s_j=-\chi_t(s_j)/\omega'(s_j)$ has fixed sign and
magnitude bounded below: its leading term is
$-(v/t)\det\mathcal H/[\mathsf w_j\omega'(s_j)^2]$, with $O(t)$ remainder.
Lemma~\ref{lem:bridge} applied to the tangent and $t\ge\eps$ give,
along the entire path, $|\dot\tau_j|\ge c\eps$ and
$\|\dot\mu\|_{\nu_n}\le C(\eps^4t+\eps^5)\le C\eps^4t$.

Let $\pi_t$ be the original prior restricted to the layer and inner
rectangles; its mass is comparable to $t^{3-2\gamma}$.
For each fixed $0\le h\le h_0$, the image bands
$v(S_{t,h}x)\in[0.9t,1.2t]$ are disjoint across dyadic $t$,
since $1.2(t/2)<0.9t$. The Jacobian-one shears and
\eqref{k3p:prior-preservation} therefore give
$\sum_t(S_{t,h})_\#\pi_t\le C\pi_\gamma$ for $0\le h\le h_0$.
Integrating in $h/h_0$ verifies the occupation condition for the paths
$s\mapsto S_{t,h_0s}x$, labeled by $\sum_t\pi_t$.
Equation~\eqref{spine:two-integrals} now yields
\[
 R_{\pi_\gamma,j}\ge
  c\eps^2\sum_{\substack{\eps\le t\le t_{\max}\\t\ \mathrm{dyadic}}}
       t^{3-2\gamma}\min\{1,\lambda^2/t^2\}.
\]
The dyadic measure has the same distribution-function order and cutoff
as $\bar\varrho_\eps$. Equation~\eqref{k3p:tail-profile} evaluates the sum
as $\Phi_\gamma$, including the critical logarithm; a fixed coarse layer
gives saturation. Together with \eqref{k3p:upper-integral}, this proves
Theorem~\ref{thm:prior-effect}(b).

\begin{corollary}[Priors specified in nodes and weights]
\label{cor:k3-node-weight-prior}
Under the assumptions of Theorem~\ref{thm:prior-effect}(b) on the
designs, the orders~\eqref{eq:k3-singular-risk} hold when the prior is
specified in $(s,\mathsf w,t_0,v,u)$ on
$\mathcal B\times T_0\times(0,v_0)\times(0,u_0)$, with density
comparable to $u^{-\gamma}$, $0<\gamma<1$. Here $\mathcal B$ is a fixed
compact box with nonempty interior, ordered separated nodes and positive
weights bounded away from zero; $T_0$ is a bounded interval with
nonempty interior, $u_0>0$, $0<v_0<1$, and the induced configurations
remain in the admissible interior. The box and density bounds do not
depend on $\eps,n,\zeta_n$.
\end{corollary}

\begin{proof}
The real moment map $(s,\mathsf w)\mapsto\mathfrak r$ is injective on
ordered positive three-atom measures, and its Jacobian, given in
Appendix~\ref{geom:k3-paths}, is bounded above and below on
$\mathcal B$; so the induced density is
comparable to $u^{-\gamma}$, and~\eqref{k3p:information} gives the same
upper integral. For the lower bound, restrict to a fixed moment
rectangle inside $\mathfrak r(\operatorname{int}\mathcal B)$, an inner
interval of $T_0$ and smaller $u,v$ ranges; the original prior dominates
this restriction, and the shear argument above needs only bounded $t_0$
and an interior margin, since $t_0\Ree M$ cancels from
$\Ree M\times\Imm M$.
\end{proof}

\section{Proof of Theorem~\ref{risk:transfer}}
\label{risk:transfer-proof}

\begin{proof}[Proof of Theorem~\ref{risk:transfer}]
In a chart of $\mathcal M$,
$I(x)$ is the infimum of $\|DF(x)h\|^2/|Dg(x)h|^2$ over rational tangent
coordinates $h$ with $Dg(x)h\ne0$, hence Borel, and $\mathcal A_\pi(r;I)$ is defined.

(a) Fix $y$ and a time $s$ at which $\gamma_y$ is differentiable, and
put $h=\gamma_y'(s)\in T_{\gamma_y(s)}\mathcal M$. The chain rule gives
$(F\circ\gamma_y)'=DF\,h$ and $(g\circ\gamma_y)'=Dg\,h$. If $Dg\,h\ne0$,
the definition of $I$ applied to $h/(Dg\,h)$ gives
$d_y^2\ge\|DF\,h\|^2\ge I(\gamma_y(s))|Dg\,h|^2\ge I(\gamma_y(s))b_y^2$;
if $Dg\,h=0$ then $b_y=0$. Because the target derivative has a fixed
sign, $b_y\le|g(\gamma_y(1))-g(\gamma_y(0))|\le\sqrt L$. Hence, for almost
every $s$,
$b_y^2\min\{1,r^2/d_y^2\}\le\min\{b_y^2,r^2b_y^2/d_y^2\}
\le\min\{L,r^2/I(\gamma_y(s))\}$, including the cases $d_y=0$ (then $I(\gamma_y(s))b_y^2=0$) and
$I(\gamma_y(s))=0$. Integrate over $s\in[0,1]$ and $\Upsilon$, and use the
occupation bound in~\eqref{crit:path-family}.

(b) The function $w_x(\delta)=\min\{\sqrt L,\delta/\sqrt{S(x)}\}$ is
nondecreasing with $w_x(t\delta)\le tw_x(\delta)$ for $t\ge1$,
and~\eqref{crit:profile} is~\eqref{risk:modulus} with $C_0=C_U$, so
$\int w_x(r_n)^2\pi(dx)=\mathcal A_\pi(r_n;S)$. Lemma~\ref{spine:risk}(a),~\eqref{crit:matching}
and~\eqref{spine:two-integrals} give
\[
 \frac{c_L}{B}\mathcal L_\Upsilon(r_n)\le R_\pi^*(r_n)
 \le(1+4\bar q)C_U^2\,\mathcal A_\pi(r_n;S)
 \le(1+4\bar q)C_U^2C_M\,\mathcal L_\Upsilon(r_n).
\]
If $S\le C_SI$ with $C_S\ge1$, then
$\min\{L,r^2/I\}\le C_S\min\{L,r^2/S\}$, so part~(a) gives
$\mathcal L_\Upsilon/B\le\mathcal A_\pi(r_n;I)\le C_S\mathcal A_\pi(r_n;S)$,
and the chain closes.

(c) By Berge's
maximum theorem $x\mapsto\varpi_x(\delta)$ is upper semicontinuous; the
supremum defining $\bar\varpi_x$ may be taken over rational $u$, and
$\bar\varpi_x$ is continuous in $\delta$, so
$(x,\delta)\mapsto\bar\varpi_x(\delta)$ is jointly measurable. Since
$\bar\varpi_x$ is nondecreasing, $\bar\varpi_x(\delta)/\delta$ is
nonincreasing and $\bar\varpi_x\ge\varpi_x$, it is admissible
in~\eqref{risk:modulus} with $C_0=1$; any admissible $C_0w_x$ dominates
$\varpi_x$ and therefore $\bar\varpi_x$.
\end{proof}

\section{Proofs of Lemmas~\ref{crit:SleI} and~\ref{crit:gauge} and Propositions~\ref{crit:suff} and~\ref{prop:kappa-law}}
\label{crit:proofs}

\begin{proof}[Proof of Lemma~\ref{crit:SleI}]
At a relative-interior point, apply~\eqref{crit:profile} along a $C^1$
curve with tangent $h$ and divide by its step. This gives
$\sqrt{S(x)}|Dg(x)h|\le C_U\|DF(x)h\|$ when $S(x)<\infty$.
Taking the infimum over $Dg(x)h=1$ proves the claim. If $S(x)=\infty$,
the inverse bound forces $Dg(x)h=0$ for every $h$, so $I(x)=\infty$.
\end{proof}

\begin{proof}[Proof of Proposition~\ref{crit:suff}]
(a) Put $\delta=\|F(y)-F(x)\|$ and split at
$c_v\sqrt{LS_v(x)}$. Below this value use the assumed bound and
$S\le S_t$; above it use $|g(y)-g(x)|\le\sqrt L$ and $S\le S_v$.
These give the stated constant $C_U$.

(b) Almost everywhere in $s$,
\[
 b_y^2\min\{1,r_n^2/d_y^2\}
 \ge\frac{c_P}{\max\{1,C_P\}}\min\{L,r_n^2/S(\gamma_y(s))\}.
\]
Integrate against the occupation measure, use its lower bound on $G$,
and then the tail comparison. This gives the stated $C_M$.
\end{proof}

\begin{proof}[Proof of Proposition~\ref{prop:kappa-law}]
(a) The identities are Tonelli's layer-cake formula, as
in~\eqref{risk:capped-tail}, applied to $\min\{L,t/S\}$ and $t/S$;
$A\le\bar A$ follows from concavity of
$x\mapsto\min\{L,x\}$.

(b) For $H=\max\{s_1,t/L\}$ and $t\le Ls_0$, the two tail bounds give
\[
 A(t)\asymp t\left[s_0^{-1}
       +s_0^{-\kappa}\int_H^{s_0}u^{\kappa-2}\dd u\right].
\]
For the lower comparison integrate from $2H$ when $4H\le s_0$;
otherwise the $s_0^{-1}$ term suffices. Evaluating the integral gives
the three orders, with $H=s_1$ for the uncapped inverse moment.
If $t\ge Ls_0$, then $L\ge\bar A\ge A\ge c_-L$; if $t\le Ls_1$,
then $t/S\le L$ almost surely.

(c) If $\kappa>1$, part~(b) with $H=s_1$ gives
$t\,\mathbb E_\pi[1/S]\le(1+C_+/(\kappa-1))t/s_0$ for every
$s_1\in[0,s_0/8]$, which is at most a constant times $A(t)$ for
$t\le Ls_0$; for $t\ge Ls_0$ both are of order $L$. If $\kappa\le1$, fix
$0<s_1\le s_0/8$ and let $F_S(u)=0$ for $u\le s_1$ and
$F_S(u)=\min\{1,c_-(u/s_0)^\kappa\}$ for $u>s_1$ (an atom of mass
$c_-(s_1/s_0)^\kappa$ at $s_1$). This law
satisfies~\eqref{crit:lower-tail} with these $s_0,s_1$: since
$c_-\le\min\{1,C_+\}$, it is at most $C_+(u/s_0)^\kappa$, and it equals
$c_-(u/s_0)^\kappa$ for $s_1<u\le s_0$. For $Ls_1\le t\le Ls_0$ put
$x=\log(Ls_0/t)$ and $y=\log(s_0/s_1)$, so that $H=t/L$, $L/A(t)$ and
$t\,\mathbb E_\pi[1/S]/A(t)$ are read off part~(b): with $c>0$
depending only on $(\kappa,c_-,C_+)$,
\[
 \frac{\bar A(t)}{A(t)}\ge
 \begin{cases}
 c\min\{e^{\kappa x},e^{(1-\kappa)(y-x)}\},&\kappa<1,\\
 c\min\{e^x/(1+x),(1+y)/(1+x)\},&\kappa=1.
 \end{cases}
\]
Taking $x=y/2$ if $\kappa<1$ and $x=\log y$ if $\kappa=1$, the right side
tends to infinity as $s_1\downarrow0$, so no constant $C$ serves all
$s_1\in(0,s_0/8]$. (The law $F_S(u)=\min\{1,c_-(u/s_0)^\kappa\}$ with
$s_1=0$ has $\mathbb E_\pi[1/S]=\infty$, so $\bar A=L$, whereas $A(t)\to0$
as $t\to0$.)
\end{proof}

\begin{proof}[Proof of Lemma~\ref{crit:gauge}]
If $s_1>\min\{2^{-\mathsf p-2},1/8\}s_0$, use
$c_1s_1\le S\le C_1s_0$. Otherwise put $s_0'=C_1s_0$,
$s_1'=2^{\mathsf p-1}C_1s_1$ and $S'=\max\{S,s_1'\}$. As
$s_1'\le2^{\mathsf p-1}(C_1/c_1)S$, $S\le S'\le KS$ a.s. with
$K=\max\{1,2^{\mathsf p-1}C_1/c_1\}$, so replacing $S$ by $S'$ changes
$A(t)$ and $t\,\mathbb E_\pi[1/S]$ by factors in $[1,K]$. Now $S'$
satisfies~\eqref{crit:lower-tail} with the scales $s_0',s_1'$, $c_-=c_2$
and $C_+=C_2(C_1/c_1)^\kappa$: $S'\ge s_1'$ and $s_1'\le s_0'/8$, as
$s_1\le2^{-\mathsf p-2}s_0$;
$\pi\{S'<u\}\le\pi\{c_1s_0\mathsf r^{\mathsf p}<u\}\le C_+(u/s_0')^\kappa$;
for $2s_1'\le u\le s_0'$ and $t=(u/s_0')^{1/\mathsf p}\in[2\mathsf r_1,1]$,
$\{S'<u\}=\{S<u\}\supset\{\mathsf r<t\}$, so
$\pi\{S'<u\}\ge c_-(u/s_0')^\kappa$; and, as $\mathsf r_1<1/2$,
$c_2\le\pi\{\mathsf r<1\}\le1$, while $c_2\le C_2\le C_+$.
Proposition~\ref{prop:kappa-law}(b) applies to $S'$, and its orders with
$(s_0',s_1')$ and with $(s_0,s_1)$ agree up to factors depending only on
$\mathsf p$ and $C_1$.
\end{proof}

\section{Proof of Corollary~\ref{factor:risk}}
\label{factor:proofs}
\label{factor:lower}
We verify the conditions of Theorem~\ref{risk:transfer} for
Section~\ref{factor:section}, with $F=\Sigma$ in the Frobenius norm,
$T_n=\widehat\Sigma_n=n^{-1}\sum_i\mathbf X_i\mathbf X_i^T$ and
$r_n=n^{-1/2}$. Constants depend only on $m$, the box and the density bounds;
$\|A\|=\|(A_j)_j\|$ and $\underline dI\preceq\Sigma\preceq\Lambda_+I$,
$\Lambda_+=\overline d+V^2+\|A\|^2$.
\emph{Experiment conditions.}\label{factor:info-proof} Gaussian fourth
moments give
$\mathbb E\|\widehat\Sigma_n-\Sigma\|_F^2
=\{(\operatorname{tr}\Sigma)^2+\operatorname{tr}\Sigma^2\}/n\le\bar qr_n^2$,
$\bar q=(m+1)(m+2)\Lambda_+^2$, and Taylor's formula for $-\log\det$ along
$\Sigma'+t(\Sigma-\Sigma')\succeq\underline dI$ gives
$\operatorname{KL}(P_\Sigma^n,P_{\Sigma'}^n)\le n\|\Sigma-\Sigma'\|_F^2/(4\underline d^2)$;
in Remark~\ref{crit:fisher-rel},
$\operatorname{tr}(\Sigma^{-1}H\Sigma^{-1}H)/\|H\|_F^2
\in[\Lambda_+^{-2},\underline d^{-2}]$.

\begin{lemma}[Inverse bound]
\label{factor:inverse}
For symmetric $\widetilde\Sigma$ let
$\widetilde z=(\widetilde\Sigma_{j0})_{j\ge1}$,
$\widetilde G=\mathsf P(\widetilde\Sigma_{RR}-D_{\mathcal C})$, with
reference block $\widetilde\Sigma_{RR}$ and
$D_{\mathcal C}=\operatorname{diag}(d_j1_{j\in\mathcal C})_{j\ge1}$, and
\begin{equation}
 \widetilde d_0=\bigl[\widetilde\Sigma_{00}
 -\widetilde z^T\widetilde G\widetilde z/\|\widetilde G\|_F^2\bigr]_{[\underline d,\overline d]},
 \quad\text{or }(\underline d+\overline d)/2\text{ if }\widetilde G=0 .
 \label{factor:estimator}
\end{equation}
Then, for every $\Sigma$ of the model, with
$E=\|\widetilde\Sigma-\Sigma\|_F$ and $\delta/0=\infty$,
\begin{equation}
 |\widetilde d_0-d_0|\le C_U\min\{\sqrt L,E/\sqrt{S(a)}\},
 \qquad C_U=\max\{1,2\sqrt L,2(\|A\|+V)^2\}.
 \label{factor:global-modulus}
\end{equation}
\end{lemma}

\begin{proof}
The known $d_j$ cancel in
$\widetilde G-G=\mathsf P(\widetilde\Sigma_{RR}-\Sigma_{RR})$, so
$\|\widetilde G-G\|_F$, $\|\widetilde z-z\|$ and
$|\widetilde\Sigma_{00}-\Sigma_{00}|$ are at most $E$. Both
$\widetilde d_0$ and $d_0$ lie in $[\underline d,\overline d]$, so the
error is at most $\sqrt L$; this suffices if $S=0$ or
$2E>\sqrt S=\|G\|_F$, since
$\max\{1,2\sqrt L\}\min\{\sqrt L,1/2\}\ge\sqrt L$. If $E=0<S$, use
\eqref{factor:identification}. Let $0<2E\le\sqrt S$; then
$\|\widetilde G\|_F\ge\sqrt S/2$ and $E\le\|a\|^2/2$, because
$\sqrt S\le\|a\|^2$. As $\widetilde G$ lies in the range of $\mathsf P$
and $\mathsf P(zz^T)=v^2G$,
$z^T\widetilde Gz=v^2\langle G,\widetilde G\rangle$, and with
$\widetilde w=\widetilde z^T\widetilde G\widetilde z/\|\widetilde G\|_F^2$,
\[
 \|\widetilde G\|_F^2(v^2-\widetilde w)
 =v^2\langle\widetilde G-G,\widetilde G\rangle
 +(z-\widetilde z)^T\widetilde Gz+\widetilde z^T\widetilde G(z-\widetilde z),
\]
of modulus at most $\|\widetilde G\|_FE(v^2+2\|z\|+E)$. With
$\|z\|=|v|\|a\|\le V\|A\|$, division by $\|\widetilde G\|_F^2$ gives
$|v^2-\widetilde w|\le(2V^2+4V\|A\|+\|A\|^2)E/\sqrt S$. Since
$d_0=\Sigma_{00}-v^2$ and
$|\widetilde\Sigma_{00}-\Sigma_{00}|\le E\le\|A\|^2E/\sqrt S$, the
unclipped error is at most $2(\|A\|+V)^2E/\sqrt S$, and clipping does
not increase it.
\end{proof}

\label{factor:upper}At $\widetilde\Sigma=\Sigma(y)$,
$\widetilde d_0=d_0(y)$ unless $\widetilde G=0$, when $E\ge\sqrt S$;
so~\eqref{factor:global-modulus} gives~\eqref{crit:profile}
($\Theta^+=\Theta$), Lemma~\ref{crit:SleI} gives $S\le C_U^2I$ a.e.,
and concavity of $x\mapsto\min\{L,x/S\}$ gives
$\mathbb E|\widetilde d_0(\widehat\Sigma_n)-d_0|^2\le C_U^2\max\{1,\bar q\}\min\{L,r_n^2/S\}$.

\begin{lemma}[Prior-dominated paths for the factor model]
\label{factor:path-family}
On the factor-model class of Section~\ref{factor:section}, there is a
measurable path family with target derivative of fixed sign and
magnitude at least $b_y$, covariance speed at most $d_y$, and occupation
measure bounded by $B\pi$, with constants independent of $n$.
The paths vary $v^2$ over an interval of fixed length $h\le V^2/4$ in
$[V^2/4,V^2/2]$ at fixed
$(z,\Sigma_{00},(\Sigma_{jj})_{j\notin\mathcal C})$;
$b_y=h$, $d_y\le8h\sqrt S/V^2$ in the Frobenius norm,
with $S$ of~\eqref{factor:surrogate}, and the label measure is Lebesgue
measure on the labels.
\end{lemma}
\begin{proof}
Fix $0<h\le V^2/4$ with
$3h(1+8\|A\|^2/V^2)\le\overline d-\underline d$. On $\{v>0\}$,
$(z,v^2,\Sigma_{00},(\Sigma_{jj})_{j\notin\mathcal C})$ is a chart in
$(a,v,d_0,(d_j)_{j\notin\mathcal C})$ with Jacobian $2v^{m+1}$. A path
fixes all chart coordinates but $v^2=V^2/4+ht$, $t\in[0,1]$; the
labels (Lebesgue measure) are those of paths meeting
$\mathcal G=\{|a_j|\le A_j/2,\ V/2\le v\le(V^2/4+h)^{1/2},\
d_0,d_j\text{ in the middle third of }[\underline d,\overline d]\}$.
Along a path $|a_j|$ changes by a factor at most $\sqrt2$, and $d_0$,
$d_j$ by at most $h$, $8A_j^2h/V^2$, so it stays in the box, and the occupation density lies in
$[2(V/2)^{m+1}/h,2(V/\sqrt2)^{m+1}/h]$ on the swept
set $\supset\mathcal G$. Here $b_y=h$, and $\Sigma$ moves only through
$\mathsf P(zz^T)/v^2$, where $\|\mathsf P(zz^T)\|_F=v^2\sqrt{S(a)}$ is
constant; so $d_y=16h\|\mathsf P(zz^T)\|_F/V^4$ and
$d_y^2\le64h^2S(\gamma_y(t))/V^4$.
\end{proof}
As $\min\{L,r^2/S(a)\}$ depends only on $a$ and does not decrease under
$a\mapsto a/2$, the tail condition of Proposition~\ref{crit:suff}(b)
holds for these paths, which gives~\eqref{crit:matching}. At any point
with $v\ne0$ the path direction gives $I\le S/v^4$.

\begin{proof}[Proof of Corollary~\ref{factor:risk}]
For every $\mathcal C$, Theorem~\ref{risk:transfer}(b) shows that
$R_{\pi,d_0}(n)$, $\mathcal A_\pi(n^{-1/2};S)$ and
$\mathcal A_\pi(n^{-1/2};I)$ are of one order, attained
by~\eqref{factor:estimator} at $\widehat\Sigma_n$. As $S\le C_U^2I$, and
$I\le16S/V^4$ on $\{|v|\ge V/2\}$, of conditional probability bounded
below given $a$, the laws of $I$ and $S$ are comparable up to constant
factors in the argument, as are $\mathbb E_\pi[1/I]$ and
$\mathbb E_\pi[1/S]$. With all references calibrated,
Lemma~\ref{crit:gauge} applies to $S=\|a\|^4$ as in
Table~\ref{tab:kappa}, and Proposition~\ref{prop:kappa-law}(b) with $t=n^{-1}$
gives~\eqref{factor:rates}.
\end{proof}
\label{factor:calibration-proof}
\begin{remark}[Exact information with calibrated references]
\label{factor:exact}
If $\mathcal C=\{1,\ldots,m\}$ and $(a,v)\ne0$, then
\begin{equation}
 I=\Bigl(\frac{\|a\|^2}{\|a\|^2+v^2}\Bigr)^2
 \asymp\min\Bigl\{1,\frac{\|a\|^4}{v^4}\Bigr\}.
 \label{factor:efficient-information}
\end{equation}
Indeed, for $H=e_0e_0^T-\partial_{(v,a)}\Sigma[h_v,h_a]$ and
$h_a=ca+w$, $w\perp a$,
$\|H\|_F^2=(1-2vh_v)^2+2\|a\|^2(h_v+vc)^2+4\|a\|^4c^2+2(v^2+\|a\|^2)\|w\|^2$,
and minimizing over $w$, $c$ and $h_v$ gives the formula. Thus $I\to1$
as $v\to0$ at fixed $a\ne0$, while $S=\|a\|^4$; the two are comparable
where $|v|$ is bounded below.
\end{remark}

\section{Proof of Proposition~\ref{thm:known-count-adaptation}}
\label{prf:knownK-adaptation}
This appendix proves Proposition~\ref{thm:known-count-adaptation} and collects the window lemmas it uses.

Write $\zeta=\zeta_n$ and, for a truth $x$ with center $\tau_0$,
$J=J_{\min}(x)$. For $y=(u,c')\in\mathcal C_W$ put $v_k=u_k-\tau_0$,
$\Delta=\|\Pi_W(\mu_y-\mu_x)\|_{\nu_n}$, the physical moments
$\bar m_l(x)=\sum_jc_j(\eps s_j)^l$ and $\bar m_l(y)=\sum_kc'_kv_k^l$,
and $P(z)=\prod_k(z-v_k)=\sum_{l\le K}p_lz^l$; let $D_1\ge D_2$ be the
two largest $|v_k|$, with multiplicity. Constants depend only on the
quantities listed in Proposition~\ref{thm:known-count-adaptation}.

\emph{The window.} By Lemma~\ref{scope:small-window}(a),
Assumption~\ref{scope:AW} gives Assumption~\ref{ass:sampling} at every
center $t\in W$ with uniform constants, so Lemmas~\ref{lem:weighted}
and~\ref{geom:ellipse} and Proposition~\ref{thm:avg-inverse-info} apply
at $\tau_0$ with uniform constants. The difference of the atomic
functionals of $x$ and $y$ has $h\le2K\le N=2K_{\max}$ points in $W$, so
the Newton-basis argument in the proof of Lemma~\ref{scope:gw-bridge}
(the first $h$ moments determine the Newton coefficients through a unit
triangular matrix, and the higher moments are bounded in them) bounds
its moments of order $2K,\dots,N-1$ by those below $2K$. The comparison of Definition~\ref{scope:proj} and the
unit lower triangular binomial change of center from $\bar\tau$ to
$\tau_0\in W$ then give
\begin{equation}
 C^{-1}\max_{l<2K}|\bar m_l(x)-\bar m_l(y)|\le\Delta
 \le\|\mu_y-\mu_x\|_{\nu_n}\le C\max_{l<2K}|\bar m_l(x)-\bar m_l(y)|,
 \label{adapt:window}
\end{equation}
so~\eqref{spine:experiment} holds on $\mathcal C_W$ with
$F=\Pi_W\mu$, $T_n=\Pi_WY$, $r_n=\zeta$ and $\bar q=q_W$, and $\Delta\le C$
because $\mathcal C_W$ is bounded.

\emph{The Taylor projection.} If $D_W\le D_0$ and
Assumption~\ref{ass:sampling} holds at $\bar\tau$,
Lemma~\ref{scope:small-window}(b) makes the projection $\Pi_W=P_*$ onto
the first $2K$ Taylor coefficients at $\bar\tau$ admissible, and
then~\eqref{adapt:window} holds with $\bar q=2K$.

\begin{lemma}[Inverse bound on a fixed window]
\label{knownK:global-inverse}
There are $N_*,c_0,c_*,\eps_0>0$ and $C$ such that, for $\eps\le\eps_0$,
every truth $x$ and every $y\in\mathcal C_W$:
(i) if $\Delta\le c_*\eps\sqrt J$, the nodes of $y$ match those of $x$
in order within $\rho\eps/3$ and $|u_j-\tau_j|\le C\Delta/\sqrt{J_j(x)}$;
if $D_1\le N_*\eps$, then $|u_j-\tau_j|\le C\min\{\eps,\Delta/\sqrt J\}$;
(ii) if $D_1\ge N_*\max\{\eps,\Delta^{1/(K-1)}\}$, then $D_2\ge c_0D_1$
and $D_1^2\le C\eps\Delta/\sqrt J\le C\Delta^2/J$;
(iii) in all cases
\begin{equation}
 |u_j-\tau_j|^2\le C\min\{D_W^2,\ \Delta^2/J,\ \eps\Delta/\sqrt J\}.
 \label{knownK:global-bound}
\end{equation}
\end{lemma}

\begin{proof}
Since $\sum_kc'_kP(v_k)v_k^q=0$ and $|p_l|\le C$,
\eqref{adapt:window} gives
\begin{equation}
 \Bigl|\sum_jc_jP(\eps s_j)(\eps s_j)^q\Bigr|
 =\Bigl|\sum_{l\le K}p_l\{\bar m_{l+q}(x)-\bar m_{l+q}(y)\}\Bigr|
 \le C\Delta,\qquad q<K.
 \label{adapt:annihilation}
\end{equation}
As $P$ is real on the real line, the real vector
$b_k=|c_k|P(\eps s_k)d_j/\{|c_j|P(\eps s_j)d_k\}$, $d_k=\omega'(s_k)$,
in the amplitude-included form of Lemma~\ref{lem:weighted}, as
for~\eqref{eq:prf-psj}, and $b=e_j$ give
\begin{equation}
 |P(\eps s_j)|\le C\eps^{K-1}\Delta/\sqrt{J_j(x)},\qquad
 \sqrt J\le C\eps^{K-1}.
 \label{adapt:annihilator}
\end{equation}
If $D_1=|v_m|$ for a simple root $v_m$, the coefficients of
$R=P/(z-v_m)$ are bounded, and applying them to both moment sequences
gives $c'_mR(v_m)=\sum_jc_jR(\eps s_j)+O(\Delta)$, so
$a_*(D_1-D_2)^{K-1}\le C\{(D_2+T\eps)^{K-1}+\Delta\}$; if $v_m$ is
repeated, $D_2=D_1$. Choosing $c_0\le1/2$ small and then $N_*\ge2T/c_0$
large, $D_1\ge N_*\max\{\eps,\Delta^{1/(K-1)}\}$ forces $D_2\ge c_0D_1$.

(i) Choose $c_*$ with $Cc_*<(\rho/3)^K$. If $\Delta\le c_*\eps\sqrt J$,
then $\prod_k|s_j-v_k/\eps|=\eps^{-K}|P(\eps s_j)|<(\rho/3)^K$ for
every $j$; as in the proof of Proposition~\ref{thm:inverse} the roots
match in order and $|u_j-\tau_j|\le C\Delta/\sqrt{J_j(x)}
\le C\min\{c_*\eps,\Delta/\sqrt J\}$, without candidate gap or amplitude
bounds. Otherwise $\eps<\Delta/(c_*\sqrt J)$, and
$|u_j-\tau_j|\le D_1+T\eps\le(N_*+T)\eps$ when $D_1\le N_*\eps$.
(ii) Two roots with $|v|\ge c_0D_1\ge2T\eps$ give factors
$|\eps s_j-v|\ge c_0D_1/2$ at every $j$, and the monic quotient of
degree $K-2$ has modulus at least $(\rho\eps)^{K-2}/(K-1)$ at some
$\eps s_j$, by Lagrange interpolation on $K-1$ nodes. Hence
$\max_j|P(\eps s_j)|\ge cD_1^2\eps^{K-2}$, and~\eqref{adapt:annihilator}
gives $D_1^2\le C\eps\Delta/\sqrt J$; with $D_1\ge N_*\eps$ also
$\eps\le C\Delta/\sqrt J$.
(iii) Errors are at most $D_W$, and at most $2D_1$ if $D_1\ge N_*\eps$.
Case~(i) gives~\eqref{knownK:global-bound} by $\min\{a^2,b^2\}\le ab$,
and~(ii) directly. Otherwise $N_*\eps<D_1<N_*\Delta^{1/(K-1)}$, so
$\Delta>\eps^{K-1}$, $\Delta/\sqrt J\ge c$, and
$D_1^2\sqrt J/(\eps\Delta)\le C\Delta^{(3-K)/(K-1)}\eps^{K-2}\le C$,
where for $K=2$ we use $\Delta\le C$ from~\eqref{adapt:window}, for
$K=3$ the exponent is zero, and for $K\ge4$ we use $\Delta>\eps^{K-1}$.
Also $\eps\Delta/\sqrt J\le C\Delta^2/J$.
\end{proof}

\begin{lemma}
\label{knownK:band-small-ball}
Let $K=2r$ and, for real polynomials $q$ of degree below $K$ with
coefficient norm $\|q\|$, put $H_cq=(\sum_jc_jq(s_j)s_j^l)_{l<r}$. If,
given $(\tau_0,s)$, the law of $c$ has a Lebesgue density at most
$M_c$, then, uniformly in $(\tau_0,s)$, $\alpha>0$ and $0\le b\le1$,
\[
 \mathbb P_c\{\exists q:\|q\|=1,\ |q_{K-1}|\le b,\ \|H_cq\|\le\alpha\}
 \le C\min\{1,\alpha(b+\alpha)\}.
\]
\end{lemma}

\begin{proof}
There are $c_q,\alpha_*>0$ such that each unit $q$ with
$\|H_cq\|\le\alpha_*$ has $r$ evaluations $|q(s_j)|\ge c_q$: otherwise
compactness gives a unit $q$ with $H_cq=0$ and at most $r-1$ nonzero
evaluations, which vanish by Vandermonde rank, so $q$ has $K$ roots. If
$|q(s_j)|\ge c_q/2$ on $I$, $|I|=r$, the map $c_I\mapsto H_cq$ is affine
with real Jacobian $|\det V_I|^2\prod_I|q(s_j)|^2\ge c$, so Fubini gives
$\mathbb P_c\{\|H_cq\|\le C\alpha\}\le C\alpha^K$. The slab
$\{|q_{K-1}|\le b\}$ of the sphere has an $\eta$-net of size
$C\eta^{1-K}(b+\eta)$; for $\alpha\le\alpha_*/2$ and $\eta=c_1\alpha$,
each qualifying $q$ has a net point with $r$ evaluations at least
$c_q/2$ and residual at most $C\alpha$, and these net points do not
depend on $c$. The union bound gives $C\alpha(b+\alpha)$.
\end{proof}

For $K=2$, $c_j=e^{i\varphi_j}$ and uniform phases, $q\equiv1$ gives
$\mathbb P\{|H_cq|\le\alpha\}=(2/\pi)\arcsin(\alpha/2)$: fixed moduli do
not suffice for the lemma.

\begin{proof}[Proof of Proposition~\ref{thm:known-count-adaptation}]
(a) The truth lies in the compact class $\mathcal C_W$, so a measurable
minimizer exists \citep{brown1973measurable}, and
Lemma~\ref{spine:risk}(a) on $\Theta^+=\mathcal C_W$ with the
envelope $\bar\varpi_x$ of the local modulus, admissible with $C_0=1$
by Theorem~\ref{risk:transfer}(c) and jointly measurable by
Appendix~\ref{risk:transfer-proof}, gives
$\mathbb E|\widehat\tau_j-\tau_j|^2\le(1+4q_W)\int\bar\varpi_x(\zeta)^2d\pi$.
By~\eqref{knownK:global-bound},
$\bar\varpi_x\le w_x(\delta):=C\min\{D_W,\delta/\sqrt J,(\eps\delta/\sqrt J)^{1/2}\}$,
which is nondecreasing, satisfies $w_x(t\delta)\le tw_x(\delta)$ for
$t\ge1$, and is jointly measurable since $J_{\min}$ is continuous for
small $\eps$ (Lemma~\ref{lem:profiled}). The phase domination and
Proposition~\ref{thm:avg-inverse-info}, uniform in the moduli, give
$\mathbb E[J^{-1}\mid\tau_0,s,a]\le M_h\sum_k\mathbb E_{U_K}J_k^{-1}
\le C\eps^{-(3K-3)}$; by Jensen's inequality,
$\int w_x(\zeta)^2d\pi\le C\min\{D_W^2,\zeta^2\eps^{-(3K-3)},
\eps\zeta\eps^{-(3K-3)/2}\}$, and $\eps^{1-(3K-3)/2}=\eps^2/\zeta_{*,K}$
for odd $K$.

For $K=2r$ put $\bar\Psi(\delta)=\eps^2\delta/\zeta_{*,K}=\delta\eps^{3-3r}$.
If $H(\delta):=\mathbb E[\varpi_x(\delta)^2\mid\tau_0,s]\le C\bar\Psi(\delta)$
for all $\delta>0$, then, since
$\varpi_x(v)\min\{1,\zeta/v\}\le2^{1-k}\varpi_x(2^k\zeta)$ for
$2^{k-1}\zeta<v\le2^k\zeta$ ($k\ge1$) and $\le\varpi_x(\zeta)$ for
$v\le\zeta$,
$\int\bar\varpi_x(\zeta)^2d\pi\le\sum_{k\ge0}4^{1-k}H(2^k\zeta)
\le8C\bar\Psi(\zeta)$; with the bound through $w_x$ for the first two
terms, this is at most $C$ times the minimum of all three terms
in~\eqref{knownK:all-noise-upper}.

To bound $H(\delta)$, as $\varpi_x\le D_W$ assume
$\bar\Psi(\delta)\le\bar\Psi_0$ for a fixed $\bar\Psi_0\le1$, and put
$D_\delta=N_*\max\{\eps,\bar\Psi(\delta)^{1/2}\}$. Then
$D_\delta\ge N_*\max\{\eps,\delta^{1/(K-1)}\}$: if $\bar\Psi<\eps^2$,
then $\delta<\eps^{3r-1}\le\eps^{K-1}$; if $\bar\Psi\ge\eps^2$ and
$r\ge2$,
$\bar\Psi^{1/2}\delta^{-1/(2r-1)}=\eps^{(3-3r)/2}\delta^{(2r-3)/(4r-2)}
\ge\eps^{-r/(2r-1)}\ge1$; if $r=1$, $\bar\Psi^{1/2}=\delta^{1/2}\ge\delta$.
By Lemma~\ref{knownK:global-inverse}, as $N_*\ge T$,
$\varpi_x(\delta)\le C\min\{\eps,\delta/\sqrt J\}+2N_*\bar\Psi(\delta)^{1/2}+2D_*$,
where $D_*$, measurable by compactness, is the largest $D_1$ among
candidates with $\Delta\le\delta$ and $D_1\ge D_\delta$, to which
part~(ii) applies ($D_*=0$ if there are none). The computation of
Section~\ref{sec:randomphase-proof}, which does not depend on the
moduli (Lemma~\ref{geom:ellipse}), gives
$\mathbb E[\min\{\eps^2,\delta^2/J\}\mid\tau_0,s,a]\le CM_h\Psi_K(\eps,\delta)
\le C\bar\Psi(\delta)$. For a candidate counted in $D_*$ let
$\Xi(z)=\eps^{2-K}\{P(\eps z)-\prod_j(\eps z-\eps s_j)\}$, of degree
below $K$. By part~(ii), $\|\Xi\|\ge c\max_j|\Xi(s_j)|\ge cD_1^2$; its
coefficient of $z^{K-1}$, $\eps(\eps\sum_js_j-\sum_kv_k)$, is
$O(\eps D_1)$; and~\eqref{adapt:annihilation} gives
$|\sum_jc_j\Xi(s_j)s_j^l|\le C\eps^{2-K-l}\Delta\le C\bar\Psi(\delta)$
for $l<r$. So $q=\Xi/\|\Xi\|$ has $|q_{K-1}|\le C\eps/D_1$ and
$\|H_cq\|\le C\bar\Psi(\delta)/D_1^2$. In polar coordinates the law of
$c$ given $(\tau_0,s)$ has density at most $M_c=M_aM_h/(2\pi a_*)^K$,
and $D_*$ depends on $x$ only through $c$ given $(\tau_0,s)$; so
Lemma~\ref{knownK:band-small-ball} gives
$\mathbb P(D_*\ge v\mid\tau_0,s)\le C(\bar\Psi\eps v^{-3}+\bar\Psi^2v^{-4})$
for $v\ge D_\delta$, and the layer-cake formula with its boundary term
gives $\mathbb E[D_*^2\mid\tau_0,s]\le C(\bar\Psi\eps/D_\delta
+\bar\Psi^2/D_\delta^2)\le C\bar\Psi(\delta)$.

(b) For each center $t$ in the support of $Q_j$,
Appendix~\ref{geom:source-lower} (Lemma~\ref{geom:fiber-paths}(i) for
odd $K$; (ii) with $\beta=0$ on a product rectangle inside the box for
even $K$) gives paths in the box with occupation at most $B\pi_{0,j}$
and $\mathcal L_\Upsilon(\zeta)\ge c\Psi_K(\eps,\zeta)$, uniformly in
$t$, since Assumption~\ref{ass:sampling} holds at every center of $W$
with uniform constants (the window paragraph); the nodes stay in $W$
when $T\eps<\delta_W$. Mixing over $Q_j$ gives occupation at most
$(B/p_0)\pi$, and~\eqref{spine:two-integrals} with $F=\mu$ applies.
\end{proof}

\emph{From a window Gram bound to finite comparison.}\label{lb:window-lemma}
We keep the notation of Assumption~\ref{scope:AW} and
Definition~\ref{scope:proj}, with $\Lambda$ a Hermite functional on $W$
of total multiplicity at most $N$, and use the
Hermite--Genocchi formula \citep[Eq.~(52)]{deboor2005}
\begin{equation}
 v_j(t)=\int_{\Delta_{j-1}}\partial_\tau^{j-1}\psi
        \Bigl(\cdot,\textstyle\sum_lu_lt_l\Bigr)\dd u,\qquad
 \operatorname{vol}(\Delta_{j-1})=\frac1{(j-1)!},
 \label{scope:hg}
\end{equation}
valid for distinct and repeated knots; the argument
$\sum_lu_lt_l$ is a convex combination of knots and lies in $W$.

\begin{lemma}[Local and window Gram conditions]
\label{scope:small-window}
(a)~Assumption~\ref{scope:AW} implies Assumption~\ref{ass:sampling} at
every center $\tau\in W$, for every $K\le K_{\max}$, with
$\lambda_0=\lambda_G$, the same $A_\psi$ and
$r_0=\operatorname{dist}(W,\partial W^+)$.
(b)~Suppose that Assumption~\ref{ass:sampling} holds at the center
$\bar\tau$, with constants $r_0,A_\psi,\lambda_0$. There is $D_0>0$,
depending only on $K$ and these constants, such that $D_W\le D_0$ implies
Assumption~\ref{scope:AW} with $K_{\max}=K$,
$W^+=[\bar\tau-3r_0/4,\bar\tau+3r_0/4]$, the same $A_\psi$ and
$\lambda_G=\lambda_0/2$. Moreover, the orthogonal projection onto the
span of the first $2K$ Taylor coefficients at $\bar\tau$ is a
moment-norming projection of rank $2K$, with constants depending only on
$K$ and the constants of Assumption~\ref{ass:sampling}. One can take
$D_0=\min\{r_0/2,\ (1-2^{-1/2})\sqrt{\lambda_0}/A_\psi,\ 2\eps_0\}$, with
$\eps_0$ from Lemma~\ref{lem:bridge} for $T=1$ and $m=2K$.
\end{lemma}

\begin{proof}
(a)~If all knots equal $\tau$, then
$v_j(\tau,\ldots,\tau)=\partial_\tau^{j-1}\psi(\cdot,\tau)/(j-1)!$ is the
Taylor coefficient $g_{j-1}$ at the center $\tau$, so
$G_W(\tau,\ldots,\tau;\nu)=\Gamma_N(\nu)$ at that center. Its leading
principal $2K\times2K$ submatrix is $\Gamma_{2K}(\nu)\succeq\lambda_GI_{2K}$.
Part~1 of Assumption~\ref{ass:sampling} holds on
$(\tau-r_0,\tau+r_0)\subset W^+$.

(b)~Let $D_W\le r_0/2$, so $W\subset[\bar\tau-r_0/4,\bar\tau+r_0/4]$ and
$W^+\subset(\bar\tau-r_0,\bar\tau+r_0)$. Write $g_r$ for the Taylor
coefficients at $\bar\tau$. By~\eqref{scope:hg} and
$\|\partial_\tau^{j}\psi\|_\nu\le A_\psi$ for $j\le2K$,
$\|v_j(t)-g_{j-1}\|_\nu\le A_\psi\sup_{s\in W}|s-\bar\tau|/(j-1)!
\le A_\psi D_W/(2(j-1)!)$ for $j\le2K$. Since $\sum_{j\ge1}((j-1)!)^{-2}<4$, Cauchy--Schwarz gives, for
$a\in\mathbb C^{2K}$,
$\|\sum_ja_jv_j(t)\|_\nu\ge\|\sum_ja_jg_{j-1}\|_\nu-A_\psi D_W\|a\|_2
\ge(\sqrt{\lambda_0}-A_\psi D_W)\|a\|_2$. Hence $G_W(t;\nu)\succeq
(\lambda_0/2)I_{2K}$ for
$D_W\le(1-2^{-1/2})\sqrt{\lambda_0}/A_\psi$. Finally,
Lemma~\ref{lem:bridge} with center $\bar\tau$, scale $D_W/2$, $T=1$ and
$m=2K$ applies to every Hermite functional on $W$ of total multiplicity
at most $2K$ once $D_W/2$ is below its $\eps_0$; there
$\mathsf v_r=\Lambda[(\tau-\bar\tau)^r]=m_r(\Lambda)$, and~\eqref{eq:bridge}
is the comparison inequality of Definition~\ref{scope:proj}.
\end{proof}

\begin{lemma}[A projection for the whole search window]
\label{scope:gw-bridge}
Under Assumption~\ref{scope:AW}, there is a moment-norming projection
of rank at most
\[
 q_W=\binom{M+N+1}{N}-1,\qquad
 M=\left\lceil\frac{2D_W A_\psi\sqrt N}{\sqrt{\lambda_G}}\right\rceil .
\]
The projection uses only $\psi,\nu,W$ and the bounds $A_\psi,\lambda_G$.
Its comparison constants depend only on $N,D_W,A_\psi,\lambda_G$;
they and the rank bound are uniform over the design family, including
all knot collisions, and independent of $n,\eps$ and the noise.
\end{lemma}

\begin{proof}
List the support of a Hermite functional $\Lambda$ with its
multiplicities as $t_1\le\cdots\le t_h$, $h\le N$, and write
$\Lambda=\sum_{j\le h}\alpha_j[t_1,\ldots,t_j]$ in the Newton basis.
The matrix of these functionals on the monomials
$(\tau-\bar\tau)^k$, $k<h$, is unit triangular: its $(k,j)$ entry is
zero for $k<j-1$, one for $k=j-1$, and otherwise a complete homogeneous
polynomial in $t_1-\bar\tau,\ldots,t_j-\bar\tau$. Its entries and its
inverse are uniformly bounded on $W$; the remaining rows $h\le k<N$
are bounded too. Consequently, including collisions,
\begin{equation}
 \begin{aligned}
 c_{N,D_W}\|\alpha\|_2\le\|m(\Lambda)\|_2
 &\le C_{N,D_W}\|\alpha\|_2,\\
 \sqrt{\lambda_G}\|\alpha\|_2\le\|\Lambda\psi\|_\nu
 &\le2A_\psi\|\alpha\|_2 .
 \end{aligned}
 \label{scope:window-newton}
\end{equation}
The second pair follows by padding the knot list to length $N$,
applying Assumption~\ref{scope:AW}, and using~\eqref{scope:hg}.

Take a grid of $M+1$ equally spaced points on $W$, and project onto
the span of $v_j(s)$ for all ordered grid lists $s$ of length $j\le N$.
Formula~\eqref{scope:hg} gives
$\|v_j(t)-v_j(s)\|_\nu\le A_\psi\|t-s\|_\infty/(j-1)!$.
Rounding knots down preserves their order, so, with mesh
$d=D_W/M\le\sqrt{\lambda_G}/(2A_\psi\sqrt N)$,
$\|(I-\Pi_W)\Lambda\psi\|_\nu\le A_\psi d\sqrt N\|\alpha\|_2
\le\frac12\|\Lambda\psi\|_\nu$. Orthogonality and~\eqref{scope:window-newton} give the required lower
comparison, with $c_\Pi=\sqrt{3\lambda_G/4}/C_{N,D_W}$; the upper
comparison has $C_\Pi=2A_\psi/c_{N,D_W}$. Counting ordered grid lists
gives $\sum_{j=1}^N\binom{M+j}{j}=q_W$.
\end{proof}

The rank $q_W$ need not be small; it enters the projected-noise
constant and, through $\mathsf t$, the count-selection range of
Proposition~\ref{scope:unknownK}.

\begin{lemma}[A regular block without aliases on $W$]
\label{scope:regular-block}
For $\psi(w,\tau)=e^{i\pi w\tau}$, let
$\nu_0=N^{-1}\sum_{i=0}^{N-1}\delta_{w_0+i\mathsf h}$, where
$\mathsf h>0$ and $\mathsf hD_W<2$. Then
Assumption~\ref{scope:AW} holds for this design. It holds uniformly for
families supported on a fixed band $[-B,B]$ with
$\nu\ge\sum_k\beta_k\nu_0^{(k)}$, where each $\nu_0^{(k)}$ is a block as
above with spacing $\mathsf h_k\in[h_*,(2-\delta)/D_W]$ and
$\sum_k\beta_k\ge\beta_*>0$; $\lambda_G$ then depends only on
$N,h_*,\delta,B,\beta_*$.
Here $N=2K_{\max}$, so the same design covers every $K\le K_{\max}$.
\end{lemma}

\begin{proof}
For distinct knots, the matrix $V_{ij}=v_j(w_0+i\mathsf h;t)$ is obtained
from the evaluation matrix by the triangular divided-difference
transform. The evaluation matrix is Vandermonde in
$e^{i\pi\mathsf h t_l}$, up to unit column factors. Thus, with
$\operatorname{sinc}z=\sin z/z$ and $\operatorname{sinc}0=1$,
\[
 |\det V(t)|=\prod_{a<b}\pi\mathsf h\operatorname{sinc}
       \bigl(\pi\mathsf h(t_b-t_a)/2\bigr)>0 .
\]
Continuity extends this identity to
collisions. All factors stay positive on the compact ordered knot
simplex, and, as $|\det V|$ is the product of the singular values of
$V$, $\lambda_{\min}(V^*V)\ge|\det V|^2/\|V\|^{2(N-1)}$; so
$G_W=N^{-1}V^*V$ has a positive uniform minimum eigenvalue.
For a block with spacing in $[h_*,(2-\delta)/D_W]$ inside $[-B,B]$ each
factor is at least $\pi h_*\operatorname{sinc}(\pi(2-\delta)/2)$, and the
entries of $V$ are bounded by~\eqref{scope:hg}, which gives a common
single-block bound $\lambda_*$. Since $G_W(t;\nu)$ is linear and
monotone in $\nu$,
$G_W(t;\nu)\succeq\sum_k\beta_kG_W(t;\nu_0^{(k)})\succeq\beta_*\lambda_*I_N$.
Finally $\|\partial_\tau^r\psi\|_\nu\le\max\{1,(\pi B)^N\}$ for
$r\le N$ supplies the common regularity bound.
\end{proof}

For the equispaced design $w_k=k/n$ on $[-1/2,1/2]$, fix
$\mathsf h=\min\{1/(2N),(2-\delta)/D_W\}$ and $m=\lfloor\mathsf hn\rfloor$:
the arithmetic progressions of step $m/n$, cut into consecutive runs of
$N$ points, are disjoint blocks carrying mass at least
$1-Nm/n\ge1/2$, so the condition holds for all large $n$; a single
$N$-point block, of mass $N/n$, does not suffice. The window
restriction is needed: on a grid of step $\mathsf h$, locations $\tau$
and $\tau+2/\mathsf h$ differ only by a common phase.

\section{Proof of Proposition~\ref{scope:unknownK}}
\label{prf:scope-unknownK}
This appendix proves Proposition~\ref{scope:unknownK} on source-count selection.

Write $\zeta=\zeta_n$.

\begin{lemma}[Separation from a smaller count]
\label{scope:iso}
For a $K$-source truth, $2\le K\le K_{\max}$ and $0<\eps\le1$, every
signal $\mu'$ on fewer than $K$ distinct points of $W$, with arbitrary
complex coefficients, satisfies
\[
 \|\Pi_W(\mu_x-\mu')\|_{\nu_n}\ge c_K\eps^{2K-2}.
\]
The constant is uniform over the truths and candidate signals.
\end{lemma}
\begin{proof}
Let $u_1,\ldots,u_{K'}$, $K'<K$, be the distinct candidate nodes. One
of the disjoint intervals of radius $\rho\eps/3$ about the true nodes,
say that about $\tau_j$, contains no $u_l$. The polynomial
$p(t)=\prod_{i\ne j}(t-\tau_i)\prod_{l\le K'}(t-u_l)$
has degree at most $2K-2<2K_{\max}$. For the difference $\Lambda$ of
the two atomic functionals,
$|\Lambda p|=|c_jp(\tau_j)|\ge a_*\rho^{K-1}(\rho/3)^{K'}\eps^{K-1+K'}
\ge c\eps^{2K-2}$. Its coefficients in powers of $t-\bar\tau$ are bounded because its
roots lie in fixed $W$. Hence
$|\Lambda p|\le C\|m(\Lambda)\|\le
C c_\Pi^{-1}\|\Pi_W\Lambda\psi\|_{\nu_n}$ by
Definition~\ref{scope:proj}, proving the claim.
\end{proof}

\begin{proof}[Proof of Proposition~\ref{scope:unknownK}]
Choose $c_{\rm sel}<\min_{2\le k\le K_{\max}}c_k$.
In an orthonormal basis of the range of $\Pi_W$, the projected noise
has $q_\Pi$ independent $\CN(0,\zeta^2)$ coordinates. For
$Z_\Pi=\|\Pi_W\eta\|_{\nu_n}$, the exponential moment
$\mathbb E e^{Z_\Pi^2/(2\zeta^2)}=2^{q_\Pi}$ and Markov's inequality give
$\mathbb P(Z_\Pi>\mathsf t)\le2^{q_\Pi}e^{-\mathsf t^2/(2\bar\zeta^2)}=\alpha$.
On $\mathcal E=\{Z_\Pi\le\mathsf t\}$ the true $K$-source model is
feasible, since $r_K\le Z_\Pi$. The selected $m$ is therefore at most
$K$, and its fit satisfies
$\|\Pi_W(\mu_{\widehat x^{(m)}}-\mu_x)\|\le r_m+Z_\Pi
\le2\mathsf t<c_K\eps^{2K-2}$.
Lemma~\ref{scope:iso} rules out $m<K$ and any reduction of the count by
merging nodes or deleting zero coefficients. Thus on $\mathcal E$ the
selector returns exactly the $K$-class minimum-distance fit, while the
loss off $\mathcal E$ is at most $D_W^2$. Without conditioning the prior,
$\mathbb E\ell^{\rm cnt}_{K,j}\le
\mathbb E|\widehat\tau_j^{(K)}-\tau_j|^2+D_W^2\alpha$. The comparison through order $2K_{\max}$ implies the $2K$-moment
comparison~\eqref{adapt:window} by its Newton-coordinate argument.
Proposition~\ref{thm:known-count-adaptation}(a), applied to this same
projection and $K$-class fit, proves~\eqref{scope:count-upper}.
For $\zeta\le\zeta_{*,K}$ and $\eps\le1$, every term in the minimum
defining $\Psi_K$ is at least $\zeta^2$: use
$\zeta_{*,K}\le\eps$, $\eps^{-(3K-3)}\ge1$, and
$\eps^2\ge\zeta\zeta_{*,K}$. The stated choice of $\alpha$ therefore
absorbs $D_W^2\alpha\le C_\sigma^2D_W^2\zeta^2/4$ into $C\Psi_K$.

Finally, any count-and-location estimator gives a known-$K$ estimator
by reporting the projection of $\widehat\tau_j$ onto $W$ when the count
is correct, and $\bar\tau$ otherwise. Its squared loss is at most
$\ell^{\rm cnt}_{K,j}$. Taking infima and using
Proposition~\ref{thm:known-count-adaptation}(b) gives the lower bound.
\end{proof}

\bibliographystyle{IEEEtranN}
\bibliography{refs,refs_additions,refs_v3_additions,refs_lit,refs_v4_classics,refs_v5_additions,refs_v12_additions,refs_v14_gautschi,refs_tit,refs_v7_additions}

\begin{thebibliography}{51}
\providecommand{\natexlab}[1]{#1}
\providecommand{\url}[1]{#1}
\csname url@samestyle\endcsname
\providecommand{\newblock}{\relax}
\providecommand{\bibinfo}[2]{#2}
\providecommand{\BIBentrySTDinterwordspacing}{\spaceskip=0pt\relax}
\providecommand{\BIBentryALTinterwordstretchfactor}{4}
\providecommand{\BIBentryALTinterwordspacing}{\spaceskip=\fontdimen2\font plus
\BIBentryALTinterwordstretchfactor\fontdimen3\font minus
  \fontdimen4\font\relax}
\providecommand{\BIBforeignlanguage}[2]{{%
\expandafter\ifx\csname l@#1\endcsname\relax
\typeout{** WARNING: IEEEtranN.bst: No hyphenation pattern has been}%
\typeout{** loaded for the language `#1'. Using the pattern for}%
\typeout{** the default language instead.}%
\else
\language=\csname l@#1\endcsname
\fi
#2}}
\providecommand{\BIBdecl}{\relax}
\BIBdecl

\bibitem[Stoica and Nehorai(1989)]{stoica1989music}
P.~Stoica and A.~Nehorai, ``{MUSIC}, maximum likelihood, and {C}ram\'er--{R}ao
  bound,'' \emph{IEEE Transactions on Acoustics, Speech, and Signal
  Processing}, vol.~37, no.~5, pp. 720--741, 1989.

\bibitem[Donoho(1992)]{donoho1992superres}
D.~L. Donoho, ``Superresolution via sparsity constraints,'' \emph{SIAM Journal
  on Mathematical Analysis}, vol.~23, no.~5, pp. 1309--1331, 1992.

\bibitem[Cand\`es and Fernandez-Granda(2014)]{candes2014superres}
E.~J. Cand\`es and C.~Fernandez-Granda, ``Towards a mathematical theory of
  super-resolution,'' \emph{Communications on Pure and Applied Mathematics},
  vol.~67, no.~6, pp. 906--956, 2014.

\bibitem[Shahram and Milanfar(2006)]{shahram2006statistical}
M.~Shahram and P.~Milanfar, ``Statistical and information-theoretic analysis of
  resolution in imaging,'' \emph{IEEE Transactions on Information Theory},
  vol.~52, no.~8, pp. 3411--3437, 2006.

\bibitem[Kay(1993)]{kay1993estimation}
S.~M. Kay, \emph{Fundamentals of Statistical Signal Processing, Volume~{I}:
  Estimation Theory}.\hskip 1em plus 0.5em minus 0.4em\relax Englewood Cliffs,
  NJ: Prentice Hall, 1993.

\bibitem[Lee(1992)]{lee1992crb}
H.~B. Lee, ``The {C}ram\'er--{R}ao bound on frequency estimates of signals
  closely spaced in frequency,'' \emph{IEEE Transactions on Signal Processing},
  vol.~40, no.~6, pp. 1508--1517, 1992.

\bibitem[Smith(2005)]{smith2005resolution}
S.~T. Smith, ``Statistical resolution limits and the complexified
  {Cram\'er--Rao} bound,'' \emph{IEEE Transactions on Signal Processing},
  vol.~53, no.~5, pp. 1597--1609, 2005.

\bibitem[Swingler(1993)]{swingler1993frequency}
D.~N. Swingler, ``Frequency estimation for closely spaced sinusoids: Simple
  approximations to the {Cram\'er--Rao} lower bound,'' \emph{IEEE Transactions
  on Signal Processing}, vol.~41, no.~1, pp. 489--494, 1993.

\bibitem[Hochwald and Nehorai(1994)]{hochwald1994concentrated}
B.~Hochwald and A.~Nehorai, ``Concentrated {C}ram\'er--{R}ao bound
  expressions,'' \emph{IEEE Transactions on Information Theory}, vol.~40,
  no.~2, pp. 363--371, 1994.

\bibitem[van~der Vaart(1998)]{vandervaart1998asymptotic}
A.~W. van~der Vaart, \emph{Asymptotic Statistics}.\hskip 1em plus 0.5em minus
  0.4em\relax Cambridge: Cambridge University Press, 1998.

\bibitem[Murphy and van~der Vaart(2000)]{murphy2000profile}
S.~A. Murphy and A.~W. van~der Vaart, ``On profile likelihood,'' \emph{Journal
  of the American Statistical Association}, vol.~95, no. 450, pp. 449--465,
  2000.

\bibitem[Shen and Win(2010)]{shen2010fundamental}
Y.~Shen and M.~Z. Win, ``Fundamental limits of wideband localization---{Part
  I}: A general framework,'' \emph{IEEE Transactions on Information Theory},
  vol.~56, no.~10, pp. 4956--4980, 2010.

\bibitem[Brown and Purves(1973)]{brown1973measurable}
L.~D. Brown and R.~Purves, ``Measurable selections of extrema,'' \emph{The
  Annals of Statistics}, vol.~1, no.~5, pp. 902--912, 1973.

\bibitem[Yau and Bresler(1992)]{yau1992worst}
S.~F. Yau and Y.~Bresler, ``Worst case {C}ram\'er--{R}ao bounds for parametric
  estimation of superimposed signals with applications,'' \emph{IEEE
  Transactions on Signal Processing}, vol.~40, no.~12, pp. 2973--2986, 1992.

\bibitem[Dilavero{\u{g}}lu(1998)]{dilaveroglu1998nonmatrix}
E.~Dilavero{\u{g}}lu, ``Nonmatrix {C}ram\'er--{R}ao bound expressions for
  high-resolution frequency estimators,'' \emph{IEEE Transactions on Signal
  Processing}, vol.~46, no.~2, pp. 463--474, 1998.

\bibitem[He et~al.(2026)He, Liu, and Wang]{he2026twopoint}
X.~He, P.~Liu, and J.~Wang, ``Two-point resolution in spectral
  super-resolution,'' arXiv:2605.04601, 2026.

\bibitem[Liu and Zhang(2021)]{liu2021computational}
P.~Liu and H.~Zhang, ``A theory of computational resolution limit for line
  spectral estimation,'' \emph{IEEE Transactions on Information Theory},
  vol.~67, no.~7, pp. 4812--4827, 2021.

\bibitem[Batenkov et~al.(2021)Batenkov, Goldman, and Yomdin]{batenkov2020near}
D.~Batenkov, G.~Goldman, and Y.~Yomdin, ``Super-resolution of near-colliding
  point sources,'' \emph{Information and Inference: A Journal of the IMA},
  vol.~10, no.~2, pp. 515--572, 2021.

\bibitem[Hockmann(2025)]{hockmann2025sparse}
M.~Hockmann, ``Analysis of the sparse super resolution limit using the
  {Cram\'er--Rao} lower bound,'' \emph{IEEE Transactions on Information
  Theory}, vol.~71, no.~1, pp. 390--395, 2025.

\bibitem[Tang et~al.(2015)Tang, Bhaskar, and Recht]{tang2015minimax}
G.~Tang, B.~N. Bhaskar, and B.~Recht, ``Near minimax line spectral
  estimation,'' \emph{IEEE Transactions on Information Theory}, vol.~61, no.~1,
  pp. 499--512, 2015.

\bibitem[Ferreira Da~Costa and Chi(2020)]{dacosta2020stable}
M.~Ferreira Da~Costa and Y.~Chi, ``On the stable resolution limit of total
  variation regularization for spike deconvolution,'' \emph{IEEE Transactions
  on Information Theory}, vol.~66, no.~11, pp. 7237--7252, 2020.

\bibitem[Li et~al.(2020)Li, Liao, and Fannjiang]{li2020esprit}
W.~Li, W.~Liao, and A.~Fannjiang, ``Super-resolution limit of the {ESPRIT}
  algorithm,'' \emph{IEEE Transactions on Information Theory}, vol.~66, no.~7,
  pp. 4593--4608, 2020.

\bibitem[Yang(2023)]{yang2023esprit}
Z.~Yang, ``Nonasymptotic performance analysis of {ESPRIT} and spatial-smoothing
  {ESPRIT},'' \emph{IEEE Transactions on Information Theory}, vol.~69, no.~1,
  pp. 666--681, 2023.

\bibitem[Heinrich and Kahn(2018)]{heinrich2018strong}
P.~Heinrich and J.~Kahn, ``Strong identifiability and optimal minimax rates for
  finite mixture estimation,'' \emph{The Annals of Statistics}, vol.~46,
  no.~6A, pp. 2844--2870, 2018.

\bibitem[Ho and Nguyen(2019)]{ho2019singularity}
N.~Ho and X.~Nguyen, ``Singularity structures and impacts on parameter
  estimation in finite mixtures of distributions,'' \emph{SIAM Journal on
  Mathematics of Data Science}, vol.~1, no.~4, pp. 730--758, 2019.

\bibitem[Wu and Yang(2020)]{wu2020gaussian}
Y.~Wu and P.~Yang, ``Optimal estimation of {G}aussian mixtures via denoised
  method of moments,'' \emph{The Annals of Statistics}, vol.~48, no.~4, pp.
  1981--2007, 2020.

\bibitem[Rotnitzky et~al.(2000)Rotnitzky, Cox, Bottai, and
  Robins]{rotnitzky2000singular}
A.~Rotnitzky, D.~R. Cox, M.~Bottai, and J.~Robins, ``Likelihood-based inference
  with singular information matrix,'' \emph{Bernoulli}, vol.~6, no.~2, pp.
  243--284, 2000.

\bibitem[Tsang(2018)]{tsang2018conservative}
M.~Tsang, ``Conservative classical and quantum resolution limits for incoherent
  imaging,'' \emph{Journal of Modern Optics}, vol.~65, no.~11, pp. 1385--1391,
  2018.

\bibitem[Ekvall and Bottai(2022)]{ekvall2022confidence}
K.~O. Ekvall and M.~Bottai, ``Confidence regions near singular information and
  boundary points with applications to mixed models,'' \emph{The Annals of
  Statistics}, vol.~50, no.~3, pp. 1806--1832, 2022.

\bibitem[Andrews and Cheng(2012)]{andrews2012estimation}
D.~W.~K. Andrews and X.~Cheng, ``Estimation and inference with weak,
  semi-strong, and strong identification,'' \emph{Econometrica}, vol.~80,
  no.~5, pp. 2153--2211, 2012.

\bibitem[Bandeira et~al.(2019)Bandeira, Niles-Weed, and
  Rigollet]{bandeira2020optimal}
A.~S. Bandeira, J.~Niles-Weed, and P.~Rigollet, ``Optimal rates of estimation
  for multi-reference alignment,'' \emph{Mathematical Statistics and Learning},
  vol.~2, no.~1, pp. 25--75, 2019.

\bibitem[Abbe et~al.(2019)Abbe, Bendory, Leeb, Pereira, Sharon, and
  Singer]{abbe2019multireference}
E.~Abbe, T.~Bendory, W.~Leeb, J.~M. Pereira, N.~Sharon, and A.~Singer,
  ``Multireference alignment is easier with an aperiodic translation
  distribution,'' \emph{IEEE Transactions on Information Theory}, vol.~65,
  no.~6, pp. 3565--3584, 2019.

\bibitem[Watanabe(2009)]{watanabe2009algebraic}
S.~Watanabe, \emph{Algebraic Geometry and Statistical Learning Theory}, ser.
  Cambridge Monographs on Applied and Computational Mathematics.\hskip 1em plus
  0.5em minus 0.4em\relax Cambridge: Cambridge University Press, 2009, vol.~25.

\bibitem[Rousseau and Mengersen(2011)]{rousseau2011overfitted}
J.~Rousseau and K.~Mengersen, ``Asymptotic behaviour of the posterior
  distribution in overfitted mixture models,'' \emph{Journal of the Royal
  Statistical Society Series B: Statistical Methodology}, vol.~73, no.~5, pp.
  689--710, 2011.

\bibitem[Ziv and Zakai(1969)]{zivzakai1969}
J.~Ziv and M.~Zakai, ``Some lower bounds on signal parameter estimation,''
  \emph{IEEE Transactions on Information Theory}, vol.~15, no.~3, pp. 386--391,
  1969.

\bibitem[Chazan et~al.(1975)Chazan, Zakai, and Ziv]{chazan1975improved}
D.~Chazan, M.~Zakai, and J.~Ziv, ``Improved lower bounds on signal parameter
  estimation,'' \emph{IEEE Transactions on Information Theory}, vol.~21, no.~1,
  pp. 90--93, 1975.

\bibitem[Bell et~al.(1997)Bell, Steinberg, Ephraim, and
  Van~Trees]{bell1997extended}
K.~L. Bell, Y.~Steinberg, Y.~Ephraim, and H.~L. Van~Trees, ``Extended
  {Ziv--Zakai} lower bound for vector parameter estimation,'' \emph{IEEE
  Transactions on Information Theory}, vol.~43, no.~2, pp. 624--637, 1997.

\bibitem[Zhang et~al.(2023)Zhang, Shi, and Gu]{zhang2023zzb}
Z.~Zhang, Z.~Shi, and Y.~Gu, ``{Ziv--Zakai} bound for {DOAs} estimation,''
  \emph{IEEE Transactions on Signal Processing}, vol.~71, pp. 136--149, 2023.

\bibitem[Jeong et~al.(2025)Jeong, Dytso, and Cardone]{jeong2025zzb}
M.~Jeong, A.~Dytso, and M.~Cardone, ``A comprehensive study on {Ziv--Zakai}
  lower bounds on the {MMSE},'' \emph{IEEE Transactions on Information Theory},
  vol.~71, no.~4, pp. 3214--3236, 2025.

\bibitem[Gill and Levit(1995)]{gill1995applications}
R.~D. Gill and B.~Y. Levit, ``Applications of the van {Trees} inequality: a
  {Bayesian} {Cram\'er--Rao} bound,'' \emph{Bernoulli}, vol.~1, no. 1--2, pp.
  59--79, 1995.

\bibitem[Chen et~al.(2016)Chen, Guntuboyina, and Zhang]{chen2016bayes}
\BIBentryALTinterwordspacing
X.~Chen, A.~Guntuboyina, and Y.~Zhang, ``On {Bayes} risk lower bounds,''
  \emph{Journal of Machine Learning Research}, vol.~17, no. 218, pp. 1--58,
  2016. [Online]. Available: \url{https://jmlr.org/papers/v17/16-185.html}
\BIBentrySTDinterwordspacing

\bibitem[Donoho and Liu(1991)]{donoho1991geometrizing2}
D.~L. Donoho and R.~C. Liu, ``Geometrizing rates of convergence, {II},''
  \emph{The Annals of Statistics}, vol.~19, no.~2, pp. 633--667, 1991.

\bibitem[Cai and Low(2015)]{cailow2015framework}
T.~T. Cai and M.~G. Low, ``A framework for estimation of convex functions,''
  \emph{Statistica Sinica}, vol.~25, no.~2, pp. 423--456, 2015.

\bibitem[Duchi and Ruan(2024)]{duchiruan2024local}
J.~C. Duchi and F.~Ruan, ``The right complexity measure in locally private
  estimation: It is not the {Fisher} information,'' \emph{The Annals of
  Statistics}, vol.~52, no.~1, pp. 1--51, 2024.

\bibitem[de~Boor(2005)]{deboor2005}
\BIBentryALTinterwordspacing
C.~de~Boor, ``Divided differences,'' \emph{Surveys in Approximation Theory},
  vol.~1, pp. 46--69, 2005. [Online]. Available:
  \url{https://surveys-in-approximation-theory.com/papers/2/2.pdf}
\BIBentrySTDinterwordspacing

\bibitem[Batenkov and Yomdin(2014)]{batenkov2014geometry}
D.~Batenkov and Y.~Yomdin, ``Geometry and singularities of the {Prony}
  mapping,'' \emph{Journal of Singularities}, vol.~10, pp. 1--25, 2014.

\bibitem[Goldman et~al.(2018)Goldman, Salman, and Yomdin]{goldman2018prony}
G.~Goldman, Y.~Salman, and Y.~Yomdin, ``Geometry and singularities of {Prony}
  varieties,'' \emph{Methods and Applications of Analysis}, vol.~25, no.~3, pp.
  257--276, 2018.

\bibitem[Tsybakov(2009)]{tsybakov2009nonparametric}
A.~B. Tsybakov, \emph{Introduction to Nonparametric Estimation}, ser. Springer
  Series in Statistics.\hskip 1em plus 0.5em minus 0.4em\relax New York:
  Springer, 2009.

\bibitem[Richmond(2006)]{richmond2006threshold}
C.~D. Richmond, ``Mean-squared error and threshold {SNR} prediction of
  maximum-likelihood signal parameter estimation with estimated colored noise
  covariances,'' \emph{IEEE Transactions on Information Theory}, vol.~52,
  no.~5, pp. 2146--2164, 2006.

\bibitem[Ben-Haim and Eldar(2009)]{benhaim2009bias}
Z.~Ben-Haim and Y.~C. Eldar, ``A lower bound on the {Bayesian} {MSE} based on
  the optimal bias function,'' \emph{IEEE Transactions on Information Theory},
  vol.~55, no.~11, pp. 5179--5196, 2009.

\bibitem[Aharon and Tabrikian(2024)]{aharon2024asymptotically}
O.~Aharon and J.~Tabrikian, ``Asymptotically tight {Bayesian} {Cram{\'e}r--Rao}
  bound,'' \emph{IEEE Transactions on Signal Processing}, vol.~72, pp.
  3333--3346, 2024.

\end{thebibliography}
\clearpage
\includepdf[pages=-]{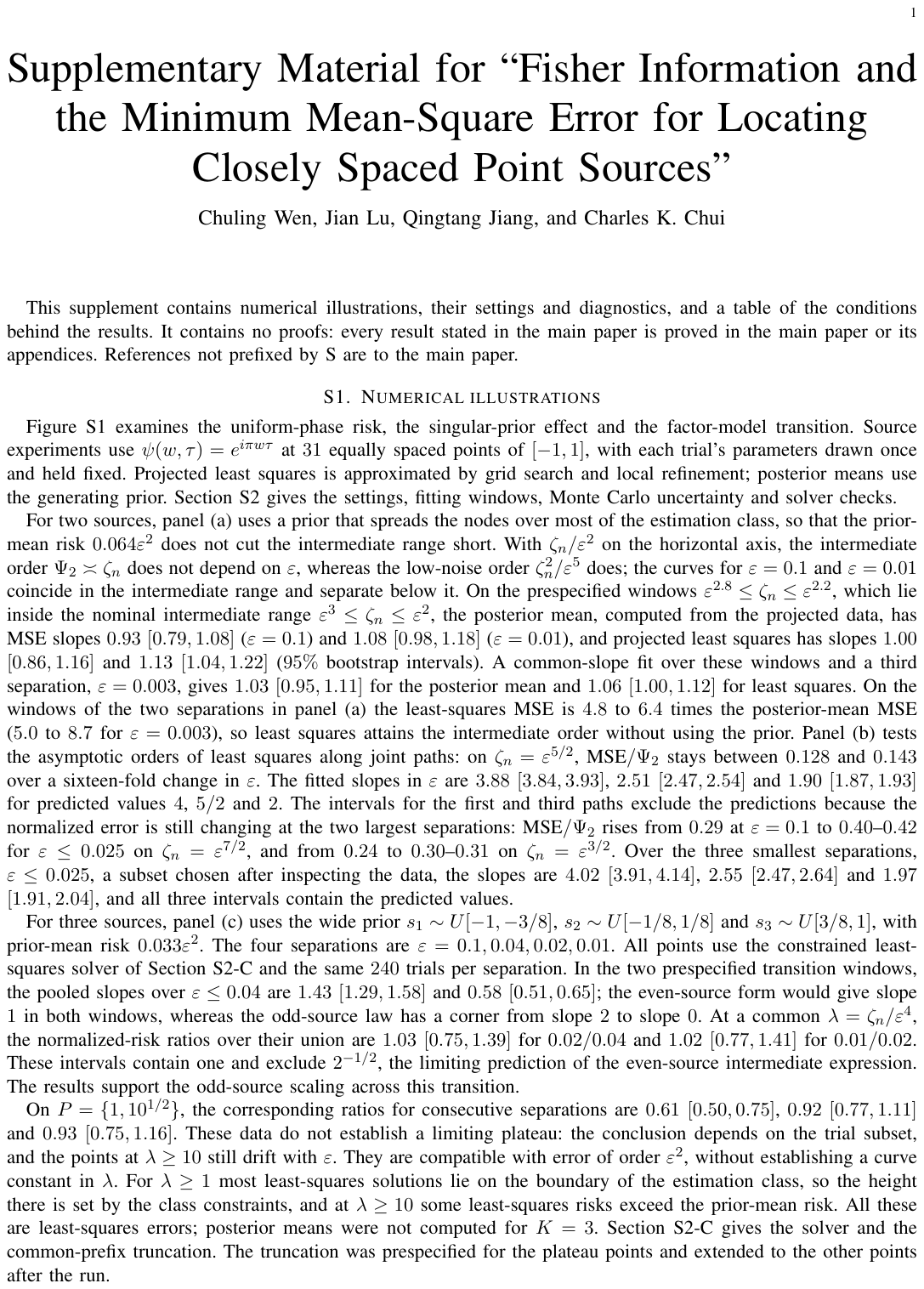}
\end{document}